\documentclass[11pt]{article}
\usepackage{style}
\usepackage{subcaption}
\usepackage{qcircuit}
\pdfoutput=1

\AtEveryBibitem{ %
    \clearfield{day}
    \clearfield{month}
    \clearfield{series}
    \clearfield{venue}
    \clearname{editor}
    \clearlist{publisher}
    \clearlist{location} %
    \clearfield{venue}
    \clearfield{issn}
    \clearfield{isbn}
    \clearfield{urldate}
    \clearfield{eventdate}
    \clearfield{pages}
    \clearfield{number}
    \clearfield{volume}
}

\newcommand{\Span}{\operatorname{span}}

\newcommand{\Useful}{\operatorname{Useful}}

\newcommand{\bools}{\{-1, 1\}^N}
\newcommand{\zo}{\{0,1\}}
\newcommand{\pmone}{\{-1,1\}}
\newcommand{\DDP}{\mathrm{DDP}}

\newcommand{\ev}{\operatorname{ev}}

\DeclareMathOperator{\supp}{supp}

\newcommand{\GHS}{\mathrm{GHS}}
\newcommand{\msg}{\mathrm{msg}}
\newcommand{\PFF}{\mathrm{PFF}}
\newcommand{\APFF}{\mathrm{APFF}}
\newcommand{\berr}{\mathbf{err}}

\newcommand{\Stab}{\mathrm{Stab}}

\usepackage{booktabs}
\usepackage{tabularx}
\usepackage{array}
\usepackage{ragged2e}
\usepackage{makecell}
\usepackage{amsmath}
\usepackage{amssymb}

\newcolumntype{L}{>{\RaggedRight\arraybackslash}X}

\newcommand{\Qoneway}{\mathsf{Q}^{1}}
\newcommand{\Roneway}{\mathsf{R}^{1}}
\newcommand{\QSMP}{\mathsf{Q}^{\parallel}}

\newcommand{\Rtwoway}{\mathsf{R}}

\providecommand{\poly}{\operatorname{poly}}

\title{Optimal and Verifiable Quantum Advantages\\in Communication Complexity}

\author{
Ryan Anselm\\
{\normalsize UT Austin}
\and
Michelle Ding\\
{\normalsize UT Austin}
\and
Dar Gilboa\\
{\normalsize Google Quantum AI}
\and
Sabee Grewal\footnote{Now a postdoc at IBM Quantum and Columbia University; work on this paper was completed while affiliated with Google Quantum AI and UT Austin.}\\
{\normalsize Google Quantum AI}\\
{\normalsize UT Austin}
}

\date{}

\begin{document}

\maketitle

\begin{abstract}
We establish optimal quantum--classical separations in communication complexity for search problems. 
We introduce a total search problem called \emph{Pelagic Fourier Fishing} and show that it admits an $n$-qubit quantum one-way protocol, whereas every randomized two-way protocol requires $\Omega(2^n)$ bits of communication. 
We then introduce a variant of this problem whose solutions can be verified in polynomial time. This variant also admits an $n$-qubit quantum one-way protocol, while every randomized one-way protocol requires $\Omega(2^n)$ bits of communication.

We also construct a family of efficiently verifiable total search problems achieving an $n$ versus $\Omega_d(n^d)$ separation between quantum one-way and randomized one-way communication for every fixed $d\ge 2$.
In the quantum protocol, Alice prepares her message using a single unitary from the $d$th level of the Clifford hierarchy, and Bob performs a Clifford measurement.
This separation is asymptotically optimal under this restriction on Alice's message. 
For $d=2$, Alice's message is a stabilizer state and Bob's measurement is Clifford, so the protocol uses no magic, yet achieves an optimal quadratic quantum advantage.

Finally, we discuss how these separations can be adapted to near-term quantum advantage experiments in which the demonstrated advantage is both \emph{unconditional} and \emph{efficiently verifiable}.
\end{abstract}

\newpage
\hypersetup{linktocpage}
\setcounter{tocdepth}{2}
\tableofcontents

\newpage

\section{Introduction}

Is it possible for $n$ qubits to do the work of exponentially many classical bits? 
Communication complexity provides one of the cleanest settings in which we can answer this question unconditionally.
Raz gave the first exponential separation between quantum and classical communication complexity~\cite{raz1999exponential}, and subsequent work has strengthened this result across a variety of communication models and problem types~\cite{BJK08-comm,GavinskyEtAl2007,klartag2010quantumonewaycommunicationexponentially,Montanaro_2019,goos2025quantum}.

Although exponential quantum advantages in communication have been known for decades, a natural question is still not fully understood:

\begin{center}
   \emph{How large of a quantum--classical separation is possible in communication complexity?}
\end{center}

\noindent
The answer depends on the type of problem being solved.
In general, any communication task admitting an $n$-qubit one-way protocol can be simulated classically up to constant additive increase in error using $O(2^n)$ bits of communication, so an $n$ versus $\Omega(2^n)$ separation is the largest possible.
Montanaro showed that this bound is tight for a \emph{sampling problem}~\cite{Montanaro_2019}, establishing that the behavior of an $n$-qubit message cannot in general be reproduced using substantially fewer than $2^n$ classical bits.
For a sampling problem, however, correctness is a property of the entire output distribution rather than of any individual output, so a single run of the protocol provides no certificate that the desired distribution was sampled correctly.

For partial Boolean functions, the largest possible separation is only $n$ versus $\Omega(2^{n/2})$.
Indeed, any $n$-qubit one-way protocol can be simulated by a randomized classical one-way protocol using $O(2^{n/2})$ bits of communication~\cite[Appendix A]{Montanaro_2019}.
This bound is known to be tight: variants of Hidden Matching achieve an $n$ versus $\Omega(2^{n/2})$ separation~\cite{GavinskyEtAl2007}.

\paragraph{Optimal separations for search problems.}
The focus of this work is communication advantages for \emph{search problems} (also called relational problems), where Alice and Bob receive inputs and must produce any output satisfying a prescribed relation on those inputs. 
The search problems we study are all \emph{total}, meaning that every pair of inputs admits at least one valid output.
Unlike sampling problems, the validity of an output can be determined on an instance-by-instance basis; unlike Boolean functions, however, there may be many valid outputs for a given input.
This additional flexibility raises a natural question: can search problems realize the full $n$ versus $2^n$ quantum--classical separation that is impossible for partial Boolean functions?

Prior to our work, the best known separation between quantum one-way and randomized one-way communication for a total search problem was the $n$ versus $\Omega(2^{n/2})$ separation given by Hidden Matching~\cite{BJK08-comm}.
Against randomized two-way communication, the best known separation was only $n$ versus $\Omega(2^{n/3})$, given by a partial function of Klartag and Regev~\cite{klartag2010quantumonewaycommunicationexponentially} (any partial function can be viewed as a total search problem by declaring every output valid on inputs outside the promise).
Thus, for roughly two decades, a wide gap had remained between the best known separations for total search problems and the general $O(2^n)$ classical simulation bound.
\emph{A priori}, it was unclear whether this gap reflected a fundamental limitation of search problems or merely a limitation of existing constructions. 

Our first result closes this gap, answering a question raised by Aaronson, Buhrman, and Kretschmer~\cite{ABK24}, and in fact proves the optimal lower bound even against randomized two-way communication.

\begin{theorem}
\label{thm:intro-two-way}
There is a total search problem that admits an $n$-qubit quantum one-way protocol with success probability at least $3/4$ on every input, while every randomized two-way protocol that succeeds with probability at least $2/3$ on average over uniformly random inputs requires $\Omega(2^n)$ bits of communication.
\end{theorem}
Thus, the general classical simulation bound is tight even for total search problems and even against two-way classical communication. Moreover, the lower bound in \cref{thm:intro-two-way} holds on average under the uniform input distribution.
This separation is based on a total search problem we introduce called \emph{Pelagic Fourier Fishing}, which we explain in \cref{subsub:pff}.

As an immediate corollary, Pelagic Fourier Fishing also gives an optimal $1$ versus $\Omega(N)$ separation between quantum and randomized query complexity for a total search problem, where $N$ is the size of the oracle's truth table. 
Indeed, by the standard communication-to-query reduction, a $T$-query classical algorithm would yield an $O(T)$-bit two-way protocol for the distributed problem~\cite[Section~3.3]{kushilevitz1997communication}. 
This improves upon the previous $1$ versus $\Omega(N/\log N)$ separation for Fourier Fishing due to Aaronson and Chen~\cite{AaronsonChen2016}.

\paragraph{Efficiently verifiable quantum advantages.}
Beyond its theoretical importance, obtaining optimal advantages in communication complexity also has concrete applications to near-term quantum advantage experiments.
Recently, Kretschmer et al.~\cite{kretschmer2025demonstratingunconditionalseparationquantum} demonstrated an unconditional quantum advantage on a trapped-ion device: using 12 qubits, they solved a task that provably requires at least 62 bits of classical memory to reproduce.
Their experiment is based on a separation between quantum one-way and randomized one-way communication complexity. 
Thus, improved quantum--classical separations in communication complexity can translate directly into stronger experimental demonstrations of quantum space advantage.

This experimental perspective motivates additional requirements on the underlying communication problem.
In particular, it is desirable for the success of the quantum protocol to be \emph{efficiently verifiable}, meaning that a proposed solution can be checked in time polynomial in the number $n$ of communicated qubits, and by only querying polynomially-many bits of the inputs.
One drawback of the communication problem used by Kretschmer et al.~\cite{kretschmer2025demonstratingunconditionalseparationquantum} is that it is not efficiently verifiable in this sense.
By contrast, Hidden Matching has efficiently verifiable solutions, but its $n$ versus $\Omega(2^{n/2})$ separation is quantitatively too small to yield the same kind of unconditional advantage at the instance sizes accessible to current experiments~\cite{kretschmer2025demonstratingunconditionalseparationquantum}.
Can the optimal $n$ versus $2^n$ separation be achieved by a search problem whose solutions are efficiently verifiable?
Our second result answers this question affirmatively.

\begin{theorem}
\label{thm:intro-one-way}
There is a total search problem with efficiently verifiable solutions that admits an $n$-qubit quantum one-way protocol with success probability at least $3/4$ on every input, while every randomized one-way protocol that succeeds with probability at least $2/3$ on average under a uniform input distribution requires $\Omega(2^n)$ bits of communication.
\end{theorem}

Two further features of the problem in \cref{thm:intro-one-way} are that Bob's quantum computation consists only of a Clifford measurement and that his input is exponentially shorter than Alice's.
Aaronson~\cite{aaronson2007learnability} showed that, for any partial or total Boolean function, a one-way quantum protocol can be simulated by a randomized one-way protocol with only a multiplicative overhead proportional to the length of Bob's input.
Thus, when Bob's input is short, Boolean functions cannot exhibit the kind of large quantum--classical separation obtained here.
Alongside~\cite{li2024-classicallyaccessibleoracle,ABK24}, our result gives another example showing that relational problems can evade this limitation and support much larger separations even when Bob's input is short.

\paragraph{Resource-efficient quantum protocols.}
Our final result studies quantum--classical separations under restrictions on the complexity of the quantum protocol.
Suppose that Alice's $n$-qubit messages can each be prepared using at most $G$ gates from a fixed finite gate set of one- and two-qubit gates. A description of such a circuit requires only $O(G\log n)$ bits, because each gate is specified by its type and the qubits on which it acts. Alice could therefore send this circuit description classically, after which Bob could simulate the quantum protocol. 
Consequently, any such protocol has randomized classical one-way communication complexity at most $O(G\log n)$. 
In particular, achieving an $n$ versus $\Omega(K)$ quantum--classical separation requires $G=\Omega(K/\log n).$ 
Our final result exhibits a separation that saturates this bound.

\begin{theorem}
\label{thm:intro-ddp}
For every fixed $d\ge 2$, there is a total search problem with efficiently verifiable solutions that admits an exact $n$-qubit quantum one-way protocol, while every randomized one-way protocol that succeeds with probability at least $2/3$ on average over uniformly random inputs requires $\Omega_d(n^d)$ bits of communication.
Moreover, in the quantum protocol, Alice's message can be prepared using $O_d(n^d/\log n)$ gates, and Bob performs a Clifford measurement.
\end{theorem}

Since a classical description of Alice's preparation circuit uses $O_d(n^d)$ bits, the separation in \cref{thm:intro-ddp} is asymptotically optimal among protocols with this state-preparation complexity.
This restriction is also natural from an experimental perspective: any experiment demonstrating a communication complexity separation must physically implement the quantum protocol.
\Cref{thm:intro-ddp} therefore gives a tunable tradeoff between the size of the quantum--classical separation and the complexity of the required quantum circuit.

\Cref{thm:intro-ddp} is optimal from a second perspective, in terms of the Clifford hierarchy~\cite{gottesman1999demonstrating}.
In the quantum protocol, Alice's message is an $n$-qubit phase state of the form
\[
    \ket{\psi_f}
    =
    2^{-n/2}\sum_{x\in\F_2^n}(-1)^{f(x)}\ket{x},
\]
where $f:\F_2^n\to\F_2$ is a degree-at-most-$d$ polynomial.
Such a state can be prepared by applying a single diagonal unitary from the $d$th level $\calC_d$ of the Clifford hierarchy to $\ket{+}^{\otimes n}$, since any lower level of the Clifford hierarchy is also contained in $\calC_d$.
The separation in \cref{thm:intro-ddp} is asymptotically optimal among protocols in which Alice's message is prepared by a single unitary from $\calC_d$.

\begin{corollary}
\label{cor:intro-clifford-hierarchy}
Fix $d\ge 2$, and consider the search problem from \cref{thm:intro-ddp}.
Suppose Alice is restricted to sending an $m$-qubit state prepared by a single unitary from the $k$th level $\calC_k$ of the Clifford hierarchy.
Then, for every fixed $1 \le k \le d$, every one-way protocol that succeeds with probability at least $2/3$ over uniformly random inputs satisfies 
\[
    m=\Omega_{d,k}(n^{d/k}).
\]
In particular, the $n$-qubit $\calC_d$ protocol of \cref{thm:intro-ddp} is asymptotically optimal, while for every $k<d$, $n$ qubits do not suffice for all sufficiently large $n$.
\end{corollary}

Equivalently, $n$ versus $\Theta_d(n^d)$ is the largest possible separation achievable using $n$-qubit messages prepared by a single unitary from $\calC_d$. Moreover, achieving this separation requires going all the way to level $d$: messages from any fixed lower level of the Clifford hierarchy require superlinear qubit complexity.

\paragraph{Quantum advantage without magic.}
Specializing \cref{thm:intro-ddp} to $d=2$, Alice's message is an $n$-qubit stabilizer state and Bob's measurement is Clifford, while the randomized classical one-way communication complexity is $\Omega(n^2)$.
This gives an optimal $n$ versus $\Theta(n^2)$ separation within this restricted model, since every $n$-qubit stabilizer state has an $O(n^2)$-bit classical description.
In particular, this answers a question raised by Kalai and Aaronson on Aaronson's blog, \emph{Shtetl-Optimized}~\cite{aaronson2025quantum}.

This special case is somewhat surprising. We are restricting the quantum protocol to use \emph{no magic}---that is, only classically simulable operations---yet we are still able to obtain a (quadratic) quantum advantage. Moreover, the evidence available before our work for whether such an advantage should be possible is murky at best. On the one hand, a superlinear communication advantage is \emph{impossible} when the quantum protocol is restricted to stabilizer-state messages over odd-prime-dimensional qudits and Clifford measurements.\footnote{The classical simulation here follows from the so-called nonnegative discrete-Wigner representation, which applies to odd-prime-dimensional stabilizer states and Clifford measurements, but not to qubit stabilizer states~\cite{gross2006hudson,veitch2012negative}.} On the other hand, Karanjai, Wallman, and Bartlett showed that any classical simulation of the $n$-qubit stabilizer subtheory within a broad ontological-model framework requires $\Omega(n^2)$ bits of memory~\cite{karanjai2018contextuality}.

Recent work of Girish, May, Parham, and Yuen~\cite{girish2026magic} gives several connections between magic and communication complexity, showing that in a number of settings, low magic implies low classical communication complexity. Our results exhibit a complementary phenomenon in one-way communication: a quantum advantage can persist even with no magic at all, while access to progressively higher levels of the Clifford hierarchy enables progressively larger optimal separations. Thus, while magic provides an important measure of quantum communication power, it is not itself necessary for a quantum communication advantage.

\subsection{Overview}
\label{subsec:overview}

We give a high-level overview of the total search problems underlying our separations in \cref{thm:intro-two-way,thm:intro-one-way,thm:intro-ddp}.
For reference, \cref{tab:qc-separations} summarizes our results alongside the most relevant prior work.

\begin{table}[t]
\centering
\caption{Selected quantum--classical communication separations. 
In each row, we define $n$ to be the worst-case number of qubits communicated by the quantum protocol that establishes the separation. 
The gate counts include message preparation and the implementation of input-dependent operations, and assume no free input oracles. 
$\QSMP$ denotes quantum communication complexity in the simultaneous message passing model.
A check mark in the last column indicates that the solution can be verified by a classical procedure using $\poly(n)$ time and input queries.
}
\label{tab:qc-separations}

\small
\setlength{\tabcolsep}{3pt}
\renewcommand{\arraystretch}{1.18}

\begin{tabularx}{\linewidth}{
    @{}
    >{\RaggedRight\arraybackslash}X
    c c c l c
    @{}
}
\toprule
\textbf{Problem}
& \textbf{Quantum}
& \textbf{Classical}
& \textbf{Gates}
& \textbf{Type}
& \textbf{Verif.}
\\
\midrule

Vector in Subspace~\cite{klartag2010quantumonewaycommunicationexponentially}
& $\Qoneway: n$
& $\Rtwoway: \Omega(2^{n/3})$
& $2^{\Theta(n)}$
& partial function
& --
\\
\addlinespace[3pt]

Hidden Matching~\cite{BJK08-comm}
& $\Qoneway: n$
& $\Roneway: \Theta(2^{n/2})$
& $2^{\Theta(n)}$
& total relation
& $\checkmark$
\\
\addlinespace[3pt]

Distributed Fourier Sampling~\cite{Montanaro_2019}
& $\Qoneway: n$
& $\Rtwoway: \Theta(2^n)$
& $2^{\Theta(n)}$
& sampling
& --
\\
\addlinespace[3pt]

Bipartite NullCodeword~\cite{goos2025quantum}
& $\QSMP: n$
& $\Rtwoway: 2^{n^{\Omega(1)}}$
& $\ge 2^{n^{\Omega(1)}}$
& total relation
& $\checkmark$
\\

\midrule

\textbf{Pelagic Fourier Fishing (Thm.~\ref{thm:intro-two-way})}
& $\Qoneway: n$
& $\Rtwoway: \Theta(2^n)$
& $2^{\Theta(n)}$
& total relation
& --
\\
\addlinespace[3pt]

\textbf{Affine Pelagic Fourier Fishing (Thm.~\ref{thm:intro-one-way})}
& $\Qoneway: n$
& $\Roneway: \Theta(2^n)$
& $2^{\Theta(n)}$
& total relation
& $\checkmark$
\\
\addlinespace[3pt]

\textbf{$\boldsymbol{\DDP_{n,d}}$ (Thm.~\ref{thm:intro-ddp})}
& $\Qoneway: n$
& $\Roneway: \Theta_d(n^d)$
& $O_d(n^d/\log n)$
& total relation
& $\checkmark$
\\

\bottomrule
\end{tabularx}
\end{table}

\subsubsection{Pelagic Fourier Fishing} \label{subsub:pff}

The problems underlying \cref{thm:intro-two-way,thm:intro-one-way} are based on the Fourier Fishing problem~\cite{Aaronson10BQPandPH}.
In the query complexity setting, Fourier Fishing asks, given oracle access to a Boolean function $f:\{0,1\}^n\to\{-1,1\}$, to find a Fourier coefficient of large magnitude.
The problem admits a particularly simple quantum algorithm: prepare the phase state
\[
2^{-n/2}\sum_{x\in \zo^n} f(x)\ket{x},
\]
apply $H^{\otimes n}$, and measure.
The resulting outcome $y$ is sampled with probability $\abs{\wh{f}(y)}^2$, so the quantum algorithm is naturally biased toward large Fourier coefficients.

A natural candidate for a quantum communication advantage is to adapt this problem to a distributed setting. 
Here, Alice gets a Boolean function $f\colon \zo^n \to \pmone$, Bob gets a Boolean function $g\colon \zo^n \to \pmone$, and Bob must find a large Fourier coefficient of $h \coloneqq fg$. 
The quantum protocol is simply for Alice to send the $n$-qubit phase-state encoding of $f$, after which Bob applies the diagonal unitary given by $\ket{x}\mapsto g(x)\ket{x}$, applies $H^{\otimes n}$, and measures.
This samples $y$ with probability $\abs{\wh{h}(y)}^2$.

\paragraph{Going beyond Fourier Fishing.} To obtain an optimal separation of $n$ vs. $\Theta(2^n)$, we introduce a variant of Fourier Fishing that we call \emph{Pelagic Fourier Fishing}.\footnote{Pelagic fish inhabit the open water column, away from the seafloor and shore. The name reflects that Pelagic Fourier Fishing likewise asks one to fish in the ``interior'' of the Fourier spectrum: for a coefficient whose magnitude is neither too small nor too large.}
Pelagic Fourier Fishing keeps the same distributed inputs $f$ and $g$, but changes what counts as a successful output.
Rather than asking Bob to find a Fourier coefficient above some fixed threshold, we ask him to find one whose magnitude is neither too small nor too large. 
Writing $h=fg$, Bob must output a $y\in\{0,1\}^n$ such that
\begin{equation} \label{eq:Pelagic-condition}
\sum_{z\,:\,\abs{\widehat h(z)}\le \abs{\widehat h(y)}}\abs{\wh{h}(z)}^2\ge \frac18,
\qquad
\sum_{z\,:\,|\wh h(z)|\ge |\widehat h(y)|}|\widehat h(z)|^2\ge \frac18.
\end{equation}
We call this the \emph{Pelagic condition}.
Bob must find a coefficient $\widehat h(y)$ whose magnitude lies away from both extremes of the spectrum: at least $1/8$ of the total squared Fourier mass should come from coefficients of magnitude at most $|\widehat h(y)|$, and at least $1/8$ should come from coefficients of magnitude at least $|\widehat h(y)|$. Informally, Bob is asked to find a coefficient ``in the bulk'' of the Fourier spectrum, weighted by squared Fourier mass. 
We note that the same quantum protocol for Fourier Fishing succeeds automatically with constant probability.
Indeed, the set of outcomes violating the lower Pelagic inequality carry less than $1/8$ of the total squared Fourier mass, and likewise for the upper Pelagic inequality.
Since the quantum protocol samples $y$ with probability $|\widehat h(y)|^2$, and by Parseval's identity
\[ \sum_{z\in \zo^n} |\wh{h}(z)|^2 = 1,\]
it therefore outputs a valid $y$ with probability at least $3/4$ for every Boolean function $h$.

In the non-distributed version of Pelagic Fourier Fishing, an algorithm is given oracle access to a Boolean function $h\colon \zo^n\to \pmone$ and must output a $y$ satisfying the Pelagic condition. Note that the two-way randomized communication complexity of the distributed version is at most twice the randomized query complexity of the non-distributed version, because Alice and Bob can simulate the query algorithm, fixing randomness at the start and then exchanging two bits, $f(x)$ and $g(x)$, whenever the query algorithm queries $h(x)$. Thus the same asymptotic randomized query lower bound immediately follows from a randomized two-way communication lower bound.

\paragraph{Classical hardness of Fourier Fishing for communication.} The difficulty with Fourier Fishing classically is that Alice's input alone gives no indication of where the Fourier mass of $h{=}fg$ will lie. Indeed, fix any $f$. As $g$ varies, $h$ ranges over all Boolean functions. Thus Alice's message must encode enough information about $f$ to remain useful for whichever $g$ Bob receives. This is particularly challenging because the Fourier coefficient $\wh{h}(y)$ depends with equal weight on every entry of the truth table of $h$.
Moreover, this difficulty is symmetric between the two players. The same argument applies with the roles of $f$ and $g$ reversed, so messages from Bob to Alice face the same basic obstacle. 
This symmetry is what makes the problem a natural candidate for hardness even against two-way classical communication.

The above intuition suggests that this distributed variant of Fourier Fishing may be hard even for two-way classical communication, but does not indicate the communication scale at which it becomes easy. In fact, its randomized two-way communication complexity is $\Theta(2^n/n)$; we next give an informal explanation for this scale. To see where the bound comes from, and why it falls short of optimal, it helps to view Fourier coefficients as rescaled inner products. Write $F,G\in\pmone^N$ for the truth tables of $f$ and $g$, and $\bchi_y\in\pmone^N$ for the truth table of $\chi_y(x) = (-1)^{x\cdot y}$. Then the rescaled Fourier coefficient of $h$ at $y$ is
\[\wh{H}_y\coloneqq\sqrt N\cdot \wh h(y)= 2^{-n/2}\sum_{x\in \zo^n} f(x)g(x)\chi_y(x) = \frac{\langle F,\,G\odot\bchi_y\rangle}{\sqrt N},\]
where $\odot$ denotes the entrywise product. With this normalization, suppose the distributed version of Fourier Fishing asks Bob to output some $y$ with $|\wh{H}_y|\ge1/2$. By Parseval, the quantum protocol succeeds with probability at least $3/4$ on every input. When $F$ and $G$ are uniformly random, each $\wh{H}_y$ is a $1/\sqrt{N}$-normalized sum of $N$ independent uniformly random signs in $\pmone$, which for sufficiently large $N$ is approximately a standard Gaussian distribution $\mathcal{N}(0,1)$ by the central limit theorem. A guess with no communication whatsoever thus succeeds with probability about $\Pr_{Z\sim \mathcal N(0,1)}[|Z|\ge1/2]\approx0.617$.

Suppose that Alice's and Bob's truth tables are $\bF,\bG\sim\mathrm{Unif}(\pmone^N)$, and first consider a deterministic one-way protocol in which Alice sends $c$ bits. After receiving her message, Bob's posterior belief about $\bF$ is that it is uniformly distributed over the set $A\subseteq\pmone^N$ of inputs consistent with that message, which we think of as having typical size about $2^{N-c}$, since each communicated bit can eliminate half of the possibilities. Crucially, Alice does not know $\bG$, so she cannot tailor $A$ to any of the vectors $\bG\odot\bchi_y$. Heuristically, when $c = o(N)$, Alice's message can shift the coefficients, but substantial Gaussian-like noise remains. The shift $s_y$ is the average value of $\wh{H}_y$ over $A$, which Bob can compute from Alice's message and his own input, while the variation of $\bF$ within $A$ contributes the variance.

How much can Alice's message shift each mean $s_y$? Consider a simple strategy, in which Alice sends the values of her truth table on a fixed set of $c$ coordinates agreed upon in advance. Then $A$ consists of all truth tables agreeing with these values, so the shift $s_y$ is exactly the contribution of the $c$ known terms in the sum defining $\wh{H}_y$, while the remaining $N-c$ terms make up the Gaussian-like noise. Since Alice does not know $\bG$, each vector $\bG\odot\bchi_y$ looks to her like a uniformly random sign vector. Hence each shift $s_y$ is distributed as a sum of $c$ independent uniformly random signs scaled by $1/\sqrt N$, which is approximately Gaussian with mean $0$ and standard deviation $\sqrt{c/N}$.

Bob, however, knows $\bG$, so he can choose the $y^*$ whose shift $s_{y^*}$ is largest in magnitude. The maximum of $N$ roughly independent Gaussians with standard deviation $\sigma$ is about $\sigma\sqrt{2\ln N}$ (see, e.g.,~\cite[Section~2.5]{boucheron2013concentration}), so the largest shift available to Bob is about $\sqrt{2c\ln N/N}$. Once $c$ is a sufficiently large constant multiple of $N/\ln N$, this shift is a large constant, and $|\wh{H}_{y^*}|\ge1$ with probability close to $1$. This calculation suggests why revealing raw coordinates can become useful around $c = O(N/\log N)$, so this version of Fourier Fishing cannot achieve an optimal separation. It also turns out that this simple strategy is asymptotically optimal.

The same qualitative picture holds for two-way protocols. Alice may now learn about $\bG$ from Bob's messages, but conditioned on the full $c$-bit transcript, $\bF$ and $\bG$ are independent and uniformly distributed over sets $A$ and $B$ with $|A|\cdot|B|\approx2^{2N-c}$ typically. Since $\bG$ is still uniform over a huge set, substantial Gaussian-like uncertainty remains, and the same heuristic picture applies.

\paragraph{Classical hardness of Pelagic Fourier Fishing for communication.}

We now show why the factor-of-$1/n$ sub-optimality is removed for Pelagic Fourier Fishing. For a typical random Boolean function $h\colon \zo^n\to \pmone$, the empirical distribution of the rescaled Fourier coefficients $\wh{H}_y=\sqrt N\,\wh h(y)$ is close to standard Gaussian. Since a normalized coefficient $\widehat H_y$ contributes $\widehat H_y^2/N$ to the squared Fourier mass, the fraction of squared Fourier mass carried by coefficients with $|\widehat H_y|\le t$ is typically close to $C(t) \coloneqq \mathbb E_{Z\sim\mathcal N(0,1)}[Z^2\mathbf 1\{|Z|\le t\}]$. Consequently, for a typical random $h$, every coefficient satisfying the Pelagic condition also satisfies a fixed \emph{window condition}, e.g. that $|\wh{H}(y)|$ must satisfy 
\[C_{\textrm{lower}} \leq |\wh{H}(y)| \leq C_{\textrm{upper}}\]
for some constants $0 < C_{\textrm{lower}} < C_{\textrm{upper}}$. For atypical inputs, however, such a window condition need not admit any correct output. For example, if $h$ is exactly a character $\chi_y$, then all of its Fourier mass sits on $y$, for which $|\wh{H}(y)| = \sqrt{N}$, exceeding any constant upper threshold. Therefore, we officially use the Pelagic condition, but it suffices to analyze a window condition for a $1-o(1)$ fraction of inputs and ignore the rest as too low probability to contribute significantly to the error over the uniform input distribution.

This two-sided window condition allows the lower bound to become an optimal $\Theta(2^n)$. The lower-bound argument shows that the selected coefficient can be represented, up to a vanishing error, as a standard Gaussian-like fluctuation plus an independent random shift. The shift may itself be random, so the resulting distribution need not be a single Gaussian. However, a two-sided window can be chosen so that
\[ \sup_{b\in \R} \Pr_{Z\sim \calN(0,1)}[C_{\textrm{lower}} \leq | Z + b | \leq C_{\textrm{upper}}] < 2/3. \]
Hence the same bound holds if $b$ is replaced by any independent random shift. Thus, unlike the one-sided threshold of Fourier Fishing, a large shift is no longer intrinsically useful. This difference forces the classical protocol to communicate $\Omega(2^n)$ bits in total.

\subsubsection{Affine Pelagic Fourier Fishing}

Checking whether a Fourier coefficient satisfies the Pelagic condition~\eqref{eq:Pelagic-condition} appears to require essentially complete knowledge of the Fourier spectrum of $h$, since a coefficient could be arbitrarily close to failing one of the Pelagic inequalities. Thus verification seems inherently expensive for our first problem. We therefore seek a variant in which correctness can be checked using only polynomially many queries to the inputs and polynomial computation time, while preserving the $n$ vs. $\Theta(2^n)$ quantum--classical separation. 

Inspired by the ``hidden partition'' paradigm initiated by Hidden Matching and generalized in subsequent work~\cite{BJK08-comm,doriguello2020exponential}, we give Bob a hidden structure that partitions Alice's input into many much smaller instances of Pelagic Fourier Fishing. Bob is allowed to solve any one of these instances, provided he also specifies which instance he solved. This ``hidden partition'' trick means that for verifying an answer, we can restrict our attention to one smaller instance of the problem, which depends on only a small fraction of Alice's input. Note that this trick creates an easy two-way protocol, so our lower bound will only hold against one-way communication from Alice to Bob.

Concretely, Bob receives a linear subspace $V \leq \F_2^n$. Its affine cosets partition $\F_2^n$, and on each coset we obtain a smaller instance of Pelagic Fourier Fishing by considering the restriction of $f$ to that coset and its Fourier spectrum over $V$. Choosing $V$ to have logarithmic dimension makes each coset polynomially-sized, so once Bob specifies a coset, the corresponding Pelagic condition can be checked efficiently by querying $f$ only on the $\poly(n)$ coordinates in that coset. We call this variant \emph{Affine Pelagic Fourier Fishing}.

Why should a partition of $\F_2^n$ into the affine cosets of a small random subspace $V$ preserve the classical hardness of the original problem? Intuitively, because Alice does not know $V$, she cannot tailor her message to any particular coset of $V$. If Alice employed the simple strategy of revealing some arbitrary $c$ coordinates of her truth table, then a random affine partition would give each coset only a small share of these coordinates on average. The nontrivial difficulty is that Bob may choose whichever coset contains the most revealed information, so we need to control even the most heavily loaded coset. Results of Alon et al.~\cite{alon1999linear}, recently sharpened and generalized by Jaber, Kumar, and Zuckerman~\cite{jaber25}, provide exactly this type of control for the maximum load of linear hash functions. In our setting, these results can be interpreted as bounds on how heavily an arbitrary fixed subset of $\F_2^n$ can intersect any coset of a random subspace. Thus, even after choosing the most favorable coset, Bob cannot concentrate the information revealed by Alice's message onto one small problem instance. This allows the classical hardness of Pelagic Fourier Fishing to survive the partition, ultimately yielding the same $c=\Omega(2^n)$ lower bound against randomized one-way communication.

\subsubsection{The Directional Derivative Problem} \label{subsub:DDP}

Our final family of separations is based on a restricted version of Hidden Matching~\cite{BJK08-comm}.
Alice receives a degree-at-most-$d$ multilinear polynomial $q:\F_2^n\to\F_2$,
and Bob receives a nonzero direction $a\in\F_2^n$.
Their goal is to jointly output the directional derivative of $q$ in the direction $a$ at any point $x \in \F_2^n$, which we denote by $D_aq(x)\coloneqq q(x+a)+q(x)$.
Equivalently, the direction $a$ defines the perfect matching
\[
M_a=\bigl\{\{x,x+a\}:x\in\F_2^n\bigr\},
\]
and Bob must output any edge of this matching together with the parity of $q$ on its two endpoints.
Thus, the Directional Derivative Problem can be viewed as Hidden Matching with two restrictions: Alice's input is constrained to be a low-degree polynomial, and Bob's matching is restricted to a subset of matchings that can be defined by an affine shift. 

The quantum protocol is again extremely simple.
Alice sends the low-degree phase state
\[
\ket{\psi_q}=2^{-n/2}\sum_{x\in\F_2^n}(-1)^{q(x)}\ket{x}.
\]
Given $a$, Bob measures in the basis consisting of states $(\ket{x}+(-1)^b\ket{x+a})/\sqrt2$. An outcome $(x,b)$ necessarily satisfies $b=q(x)+q(x+a)=D_aq(x)$, so the protocol is exact. Moreover, Bob's measurement is Clifford and can be implemented using at most $n-1$ $\mathrm{CNOT}$ gates and one Hadamard gate.

Classically, there is a simple $O_d(n^d)$-bit protocol: Alice sends the coefficients of $q$, after which Bob can compute $D_aq(x)$ for any $x$.
Our lower bound shows that this na\"ive protocol is essentially optimal.
One difficulty is that, for each direction $a$, Bob is free to choose whichever point $x$ makes the derivative $D_aq(x)$ easiest to predict from Alice's message.

To see how we overcome this, think of the coefficients of $q$ as $\Theta_d(n^d)$ independent bits.
After receiving Alice's message, Bob will not know $q$ exactly, but he knows that it belongs to some smaller set of possible polynomials that are consistent with Alice's message.
For a given direction $a$, if $D_aq(x)$ is nearly equally likely to be $0$ or $1$ over this set for every choice of $x$, then Bob cannot predict the derivative with much advantage over random guessing.
Thus, for Bob to succeed on many directions, Alice's message must make some derivative $D_aq(x)$ noticeably biased for many different $a$.

For fixed $x$ and $a$, each derivative $D_aq(x)$ is a parity of some subset of the coefficients of $q$.
We learn some information about this parity whenever the derivative $D_aq(x)$ is biased after conditioning on Alice's message. 
The main structural ingredient in our proof shows that biased derivatives each reveal information about \emph{different} parities. 
Specifically, we show that if Bob has biased derivatives for many directions, then among them there are $\Omega_d(n^d)$ derivatives that each reveal a genuinely new piece of information about $q$; it turns out that the parities must be linearly independent, so the value of none of these parities is determined by the others.
Thus Alice's message must reveal information about $\Omega_d(n^d)$ different degrees of freedom of $q$, forcing the message to have length $\Omega_d(n^d)$.

\subsection{Toward verifiable quantum information supremacy}
\label{subsec:toward}

As discussed above, separations in one-way communication complexity can be converted directly into quantum advantage experiments. This idea was proposed by Aaronson, Buhrman, and Kretschmer~\cite{ABK24}, who called the resulting notion \emph{quantum information supremacy}, and was recently realized experimentally by Kretschmer et al.~\cite{kretschmer2025demonstratingunconditionalseparationquantum}.

The basic idea is to reinterpret a one-way communication protocol as a computation carried out on a single device at two different points in time. Alice's input determines an $n$-qubit state $\ket{\psi_x}$, which is prepared on the device. Once this state has been prepared, Bob's input determines a measurement to be performed on it, and the resulting measurement outcome is taken as the output of the protocol. Thus, Alice's quantum message is simply the quantum state retained by the device between the state-preparation and measurement stages.

Under this reinterpretation, a lower bound on classical communication becomes a lower bound on classical memory. Indeed, any classical device that could reproduce the same behavior while retaining only $m$ bits of information between the two stages, under the usual one-way information boundary, would immediately give an $m$-bit classical one-way communication protocol: Alice sends the contents of the device's memory to Bob, who continues the computation using his input. Consequently, an $n$-qubit versus $M$-bit communication separation yields an unconditional experimental separation between $n$ qubits of quantum memory and $M$ bits of classical memory.

These experiments serve two key purposes. First, they provide operational benchmarks for quantum hardware. Morally, demonstrating a large $n$ versus $M$ advantage requires Alice's messages to be prepared with sufficiently high fidelity that the communication advantage survives experimental noise.
Thus, for an implementation of the intended protocol, a large quantum information supremacy advantage provides evidence that the hardware can prepare complex states with high fidelity, although formally certifying the intended state was prepared would further require a self-testing theorem.
Second, these experiments probe a more conceptual question: whether the exponentially large state space predicted by quantum mechanics is operationally accessible. If reproducing the behavior of an $n$-qubit experiment requires $M$ bits of classical memory, then any classical finite-state model reproducing that behavior must have at least $2^M$ distinguishable internal states. Thus, quantum information supremacy gives an unconditional way to test how large a classical state space is required to reproduce the observable behavior of an $n$-qubit system, providing direct evidence for the effective exponentiality of Hilbert space.

Kretschmer et al.~\cite{kretschmer2025demonstratingunconditionalseparationquantum} recently realized this proposal experimentally on a trapped-ion device, demonstrating an unconditional separation between $12$ qubits and $62$ classical bits.
As discussed above, one drawback is that the communication task underlying their experiment is not efficiently verifiable in our sense.
Two other natural goals are to reproduce such experiments on additional hardware platforms, including superconducting and neutral-atom devices, and to demonstrate substantially larger separations.
Our results provide several candidate communication problems for pursuing these directions.

Kretschmer et al.\ introduced a new communication task in part because the $n$ versus $\Omega(2^{n/2})$ separation from Hidden Matching~\cite{BJK08-comm} was not quantitatively strong enough for the advantage to appear at the small values of $n$ accessible experimentally.
The separation in \cref{thm:intro-one-way} removes this asymptotic bottleneck: it achieves the largest possible $n$ versus $\Omega(2^n)$ separation in the one-way model while retaining efficiently verifiable solutions.
It therefore provides a natural candidate for an unconditional and efficiently verifiable demonstration of quantum information supremacy.

For finite-size experiments, however, asymptotic separation is only part of the story: constant factors and circuit complexity are equally important.
This makes the Directional Derivative Problem underlying \cref{thm:intro-ddp} another potentially useful candidate.
Its solutions are efficiently verifiable, while the parameter $d$ provides a tunable tradeoff between the size of the classical lower bound and the complexity of Alice's state preparation.
For a particular noisy device, one can therefore choose $d$ to obtain the largest separation compatible with the circuits that can be implemented reliably. The separation in \cref{sec:ghs} offers yet another experimentally viable candidate.
As in the $d=2$ case of \cref{thm:intro-ddp}, Alice sends a stabilizer state, and Bob performs a Clifford measurement.
However, in \cref{sec:ghs}, Bob's quantum circuit consists only of Hadamard gates, whereas the Directional Derivative protocol also requires $\mathrm{CNOT}$ gates.
This simpler measurement may make the former problem attractive, because one-qubit gates typically are not the dominant source of noise.

Turning any of these proposals into an experiment will require a hardware-specific optimization.
For a given platform, one must first determine which of the communication problems developed here gives the best finite-size separation at experimentally accessible circuit depths and noise rates.
This will require numerical studies of the relevant success probabilities and classical bounds, careful optimization of constant factors, and possibly variants of our constructions tailored to particular hardware architectures.
We leave these questions---and the resulting experimental demonstrations---to future work.

\subsection{Open problems}

As a complement to the directions given in \cref{subsec:toward},   
we conclude with several other open questions suggested by our results.

\begin{enumerate}
    \item \textbf{Partial functions.}
Can the optimal $n$ versus $2^{n/2}$ separation for partial Boolean functions be achieved between quantum one-way and randomized two-way communication?
The best known separation is $n$ versus $\Omega(2^{n/3})$, given by the Vector in Subspace problem of Klartag and Regev~\cite{klartag2010quantumonewaycommunicationexponentially}.
Closing the gap between the exponents $1/3$ and $1/2$ remains open.

    \item \textbf{Weaker quantum models.}
Our first result achieves the maximal $n$ versus $2^n$ separation using one-way quantum communication.
Can a comparable exponential-scale advantage be achieved in still weaker quantum communication models?
For example, is there a total search problem with $O(n)$ quantum simultaneous-message complexity but $\Omega(2^n)$ randomized two-way communication complexity?
G\"o\"os, Gur, Jain, and Li~\cite{goos2025quantum} obtained an exponential separation between quantum simultaneous-message and randomized two-way communication for a total search problem, but their separation is quantitatively far from optimal.

    \item \textbf{Efficient verifiability.}
    Our first result (\cref{thm:intro-two-way}) is hard against randomized two-way protocols, while our second result (\cref{thm:intro-one-way}) is efficiently verifiable but only establishes hardness against randomized one-way protocols.
    Can these two properties be achieved simultaneously?
    In particular, is there an efficiently verifiable total search problem admitting an $n$-qubit quantum one-way protocol but requiring $\Omega(2^n)$ randomized two-way communication?

    \item \textbf{Total functions.}
    Can an exponential quantum--classical separation for a total Boolean function be achieved in the one-way model?
    Recent work has established exponential separations for total functions using quantum protocols with two or more messages~\cite{gavinsky2026quantumcommunicationcomplexitytotal,gall2026improvedseparationsquantumclassical}.
    It remains open whether such an advantage is possible with a single quantum message.
\end{enumerate}

\pdfbookmark[1]{Acknowledgments}{acknowledgments}
\paragraph*{Acknowledgments.}
We thank Scott Aaronson, Uma Girish, Siddhartha Jain, Olivier Lalonde, Jarrod McClean, Natalie Parham, and Ronak Ramachandran for helpful discussions. R.A. is supported by an NSF Graduate Research Fellowship. M.D. is supported by the Texas Quantum Initiative Graduate Research Fellowship and U.S. Department of Energy through a quantum computing program at the Pacific Northwest National Laboratory.

\pdfbookmark[1]{Tool and computational resource disclosure}{tool-disclosure}
\paragraph*{Tool and computational resource disclosure.}
This project began without generative-AI assistance, and several of its core results were obtained before the authors used AI tools in the research process. In particular, the authors designed the Directional Derivative Problem $\DDP_{n,d}$ and proved the $d=2$ case of \cref{thm:intro-ddp} without AI assistance. They also independently obtained an $n$ versus $\widetilde{\Omega}(2^n)$ separation between quantum one-way and randomized one-way communication using Fourier Fishing, as well as the $n$ versus $\Theta(n^2)$ separation in \cref{sec:ghs}. 

All improvements on these results involved AI assistance.
In particular, at a later stage of the project, the authors used ChatGPT 5.6 and ChatGPT 6 Astra iteratively to explore improvements to these human-generated results. The authors supplied the models with their existing constructions and proofs. The models generalized the $d=2$ Directional Derivative lower bound to arbitrary fixed $d$ in \cref{thm:intro-ddp}, and designed a variant of the Fourier Fishing problem (Pelagic Fourier Fishing) that removed a logarithmic loss from the classical lower bound. The models were also used for literature search, checking arguments, and revising portions of the exposition. The authors independently verified the mathematical arguments incorporated into the paper, revised and organized the resulting proofs and exposition, and take full responsibility for the correctness, attribution, and presentation of the work.

\section{Preliminaries}

\subsection{Communication complexity}
We refer to \cite{kushilevitz-nisan-textbook} for background on
communication complexity. A two-party search problem is specified
by a relation $R\subseteq X\times Y\times Z$ over finite sets.
Alice receives $x\in X$, Bob receives $y\in Y$, and Bob must
output a solution $z\in R(x,y)$, where $R(x,y)\coloneqq\{z\in Z:(x,y,z)\in R\}$.
The relation is \emph{total} if $R(x,y)\neq\varnothing$ for every
$(x,y)\in X\times Y$.

We write $\Rtwoway_\varepsilon(P)$, $\Roneway_\varepsilon(P)$, and $\Qoneway_\varepsilon(P)$ for the minimum costs of randomized classical two-way, randomized classical one-way, and quantum one-way protocols for a problem $P$, respectively, that succeed with probability at least $1-\varepsilon$ on every input. An omitted error subscript means $\varepsilon=1/3$. In all communication models we consider, the players have access to shared randomness but no prior entanglement.

\subsection{Fourier analysis} \label{sub:fourier-analysis}
We first recall standard Boolean function Fourier-analytic conventions, following \cite{odonnell2014boolean}. For $f:\zo^n\to\R$ and $y\in\zo^n$, define
\[\widehat f(y)\coloneqq \E_{x\in\zo^n}[f(x)\chi_y(x)]=\frac{1}{2^n}\sum_{x\in\zo^n}f(x)\chi_y(x),\]
where $\chi_y(x)\coloneqq(-1)^{y\cdot x}$. We freely identify a frequency $y\in\F_2^n$ with the subset
\[S_y\coloneqq\{i\in[n]:y_i=1\}\subseteq[n].\]
Thus Fourier coefficients indexed by subsets $S\subseteq[n]$ are the same coefficients under the identification $y\leftrightarrow S_y$.

More generally, for a finite-dimensional vector space $V$ over $\F_2$, let $V^*$ denote its dual space. For $y\in V^*$, write
\[\chi_y(v)\coloneqq(-1)^{y(v)}.\]
If $U=a+V$ is a coset of $V$ and $f:\F_2^n\to\R$, define the Fourier coefficient of $f$ on $U$ at $y\in V^*$ by
\[\widehat f_U(y)\coloneqq\frac{1}{|V|}\sum_{v\in V}f(a+v)\chi_y(v).\]
Replacing $a$ by $a+w$ for $w\in V$ multiplies $\widehat f_U(y)$ by $\chi_y(w)\in\{\pm1\}$, so $|\widehat f_U(y)|$ depends only on $U$ and $y$.

We will repeatedly use Parseval's identity. For $f:\zo^n\to\R$,
\[\sum_{y\in\zo^n}|\widehat f(y)|^2=\E_{x\in\zo^n}[|f(x)|^2]=\frac{1}{2^n}\sum_{x\in\zo^n}|f(x)|^2.\]
More generally, for a coset $U=a+V$ and $f:U\to\R$,
\[\sum_{y\in V^*}|\widehat f_U(y)|^2=\frac{1}{|V|}\sum_{x\in U}|f(x)|^2.\]
In particular, if $f$ is $\pmone$-valued, then
\[\sum_{y\in\zo^n}|\widehat f(y)|^2=1,\qquad \textrm{and}\qquad \sum_{y\in V^*}|\widehat f_U(y)|^2=1\]
on every coset $U=a+V$.\\

\subsection{Information theory}

We first recall some standard information-theoretic notions. Let $\bX$ be a random variable on a finite set $\calX$, with probability mass function $p_{\bX}(x)\coloneqq\Pr[\bX=x]$. Its \emph{entropy} is
\[H(\bX)\coloneqq\sum_{x\in\calX}p_{\bX}(x)\log_2\!\left(\frac{1}{p_{\bX}(x)}\right).\]
For a Bernoulli random variable with success probability $p$ (using the convention $0\log 0=0$), the \emph{binary entropy function} is
\[H_2(p)\coloneqq-p\log_2p-(1-p)\log_2(1-p).\]
For jointly distributed random variables $\bX,\bY$, the \emph{conditional entropy} of $\bX$ given $\bY$ is
\[H(\bX\mid\bY)\coloneqq\sum_y\Pr[\bY=y]H(\bX\mid\bY=y).\]
The \emph{mutual information} between $\bX$ and $\bY$ is
\[I(\bX;\bY)\coloneqq H(\bX)-H(\bX\mid\bY)=H(\bY)-H(\bY\mid\bX),\]
and, for jointly distributed $\bX,\bY,\bZ$, the \emph{conditional mutual information} is
\[I(\bX;\bY\mid\bZ)\coloneqq H(\bX\mid\bZ)-H(\bX\mid\bY,\bZ).\]
We will also use the chain rule
\[H(\bX,\bY\mid\bZ)=H(\bY\mid\bZ)+H(\bX\mid\bY,\bZ),\]
as well as the nonnegativity of mutual information and conditional mutual information. We use the following result about mutual information.

\begin{lemma}
\label{lem:cov-mutual-information}
For jointly distributed random variables $\bX,\bY\in\pmone$,
\[\Cov(\bX,\bY)^2\leq 2I(\bX;\bY).\]
Consequently, for any jointly distributed $\bX,\bY \in \pmone$, and any finite random variable $\bZ$,
\[\E\big[\Cov(\bX,\bY\mid\bZ)^2\big]\leq 2I(\bX;\bY\mid\bZ).\]
\end{lemma}
\begin{proof}
Let $h(b)=H_2((1+b)/2)$. Since $h''(b)\leq-1$, concavity of $h(b)+b^2/2$ gives
\[I(\bX;\bY)\geq\frac12\Var(\E[\bX\mid\bY]).\]
By Cauchy--Schwarz,
\[\Cov(\bX,\bY)^2=\Cov(\E[\bX\mid\bY],\bY)^2\leq\Var(\E[\bX\mid\bY])\Var(\bY)\leq\Var(\E[\bX\mid\bY]).\]
The conditional statement follows by applying the first inequality for each value of $\bZ$ and averaging. \end{proof}
The \emph{min-entropy} of $\bX$ is
\[H_\infty(\bX)\coloneqq-\log_2\!\left(\max_{x\in\calX}p_{\bX}(x)\right).\]
Note that $H(\bX)\geq H_\infty(\bX)$. We will primarily use these notions for random variables on the Boolean hypercube $\pmone^s$.

\begin{definition}[Min-entropy deficiency]
\label{def:min-entropy-def}
For a $\pmone^s$-valued random variable $\bX$, define its \emph{min-entropy deficiency} by
\[D_\infty(\bX)\coloneqq\log_2\left(2^s\max_{x\in\pmone^s}\Pr[\bX=x]\right)=s-H_\infty(\bX).\]
In particular,
\[0\leq s-H(\bX)\leq D_\infty(\bX)\leq s.\]
\end{definition}

\subsection{Couplings}

We develop coupling tools in \Cref{subsec:shared-lemmas} and apply them in \Cref{sec:PFF,sec:APFF}.

\begin{definition}[Coupling]
\label{def:coupling}
Let $\mu$ and $\nu$ be probability distributions on $\calX$ and $\calY$, respectively. A \emph{coupling} of $\mu$ and $\nu$ is a joint distribution $\pi$ on $\calX\times\calY$ whose marginals are $\mu$ and $\nu$. Equivalently, if $(\bX,\bY)\sim\pi$, then $\bX\sim\mu$ and $\bY\sim\nu$.
\end{definition}

A coupling lets us place two random variables with fixed individual distributions on the same joint probability space, while choosing their dependence in a useful way. In particular, although the marginal laws of $\bX$ and $\bY$ are fixed, they need not be independent. In this paper, we will often couple a structured or high-entropy random variable $\bW\in\pmone^s$ with a perfectly uniform random variable $\bU\sim\mathrm{Unif}(\pmone^s)$ so that $\bW$ and $\bU$ are typically close. This allows us to compare $\bW$ to the much simpler uniform distribution while preserving the exact marginal distribution of $\bW$.\\

\subsection{Linear hashing}

We will use the following optimal expected maximum-load theorem for linear hash functions due to Jaber, Kumar, and Zuckerman~\cite{jaber25} for proving the randomized one-way lower bound in \Cref{subsec:apff-lb}. 

\begin{theorem}[Theorem 2 of \cite{jaber25}]
\label{thm:jkz-linear-hashing}
Let $n\geq\ell\geq1$, put $N=2^n$ and $B=2^{n-\ell}$, and sample a uniformly random linear map $\bL:\F_2^n\to\F_2^{n-\ell}$. For every nonempty set $D\subseteq\F_2^n$ of size $m$,
\[\E_{\bL}\max_{y\in\F_2^{n-\ell}}|D\cap\bL^{-1}(y)|\leq16\,\operatorname{opt}(m,B),\]
where
\[\operatorname{opt}(m,B)\coloneqq\begin{cases}\displaystyle\frac{\log_2B}{\log_2(B\log_2B/m)},&m\leq\frac12B\log_2B,\\[8pt]\displaystyle\frac{m}{B},&m>\frac12B\log_2B.\end{cases}\]
\end{theorem}

In particular, we refer to the following corollary, which expresses the theorem in terms of cosets of a uniformly random subspace.

\begin{corollary}[Corollary of \Cref{thm:jkz-linear-hashing}]
\label{cor:jkz-subspace-intersection}
Let $1\leq\ell<n$, let $N=2^n$ and $K=2^\ell$, and let $\bV$ be a uniformly random $\ell$-dimensional subspace of $\F_2^n$. Then, for every fixed set $D\subseteq\F_2^n$,
\[\E_{\bV}\max_{a\in\F_2^n}\frac{|D\cap(a+\bV)|}{K}\leq32\left(\frac{|D|}{N}+\frac{n-\ell}{K}\right).\]
\end{corollary}

\begin{proof}
The claim is trivial for $D=\varnothing$, so let $m=|D|>0$ and set $B=N/K=2^{n-\ell}$. Sample $\bL:\F_2^n\to\F_2^{n-\ell}$ uniformly from all linear maps, and let $\mathcal E$ be the event that $\bL$ is surjective. Since $\bL$ is represented by a uniformly random $(n-\ell)\times n$ matrix over $\F_2$,
\[\Pr[\mathcal E]=\prod_{i=0}^{n-\ell-1}(1-2^{i-n})\geq1-\sum_{i=0}^{n-\ell-1}2^{i-n}>1-2^{-\ell}\geq\frac12.\]

Conditioned on $\mathcal E$, the kernel $\bV=\ker\bL$ is a uniformly random $\ell$-dimensional subspace, and each fiber $\bL^{-1}(y)$ is a coset of $\bV$. Therefore, using $\Pr[\mathcal E]\geq1/2$ and \Cref{thm:jkz-linear-hashing},
\[\E_{\bV}\max_{a\in\F_2^n}|D\cap(a+\bV)|=\E_{\bL}\left[\max_{y\in\F_2^{n-\ell}}|D\cap\bL^{-1}(y)|\,\middle|\,\mathcal E\right]\leq2\E_{\bL}\max_y|D\cap\bL^{-1}(y)|\leq32\,\operatorname{opt}(m,B).\]

If $m>\frac12B\log_2B$, then $\operatorname{opt}(m,B)/K=m/N$. Otherwise, $\log_2(B\log_2B/m)\geq1$, so $\operatorname{opt}(m,B)/K\leq\log_2B/K=(n-\ell)/K$. Thus in either case
\[\frac{\operatorname{opt}(m,B)}{K}\leq\frac{m}{N}+\frac{n-\ell}{K}.\]
Dividing the preceding bound by $K$ proves the result.
\end{proof}

\section{Classical simulation and lower bound tools} \label{subsec:shared-lemmas}

In this section, we explain how the classical simulation protocol for one-way quantum communication works in more detail, and build up a foundation of techniques that will be used to prove the lower bounds of \Cref{thm:intro-two-way,thm:intro-one-way}.

\subsection{The classical simulation protocol}

We first note that for any problem which admits an $n$-qubit one-way protocol, the bounded-error classical communication complexity is at most $O(2^n)$ due to a classical simulation protocol that uses an $\eps$-net, as previously observed by Montanaro~\cite{Montanaro_2019}. Alice and Bob agree on an indexed set of pure states $\{\ket{\psi_i}\}_i$ such that for all $n$-qubit pure states $\ket{\psi}$, $\lVert\, \ket{\psi}\bra{\psi} - \ket{\psi_i}\bra{\psi_i}\, \rVert_1 \leq \eps$ for some $i$. Alice sends the index $i^*$ of the closest state in the $\eps$-net to the state that would have been sent in the quantum protocol. 

It is known that there are $\eps$-nets of size $(5/\eps)^{2^{n+1}}$ for the space of $n$-qubit pure states with respect to trace norm (see e.g., \cite{hayden2004randomizing}, Lemma II.4). Hence, for any fixed $\eps>0$, the index can be transmitted using $O(2^n)$ bits. Then it suffices for Bob to classically simulate his measurement on $\ket{\psi_{i^*}}$. For a mixed state message $\rho=\sum_j p_j|\psi_j\rangle\langle\psi_j|$, Alice samples $j$ according to $p_j$ and sends the index of a state approximating $|\psi_j\rangle$. By convexity of trace norm, the resulting mixture is within $\eps$ of $\rho$, so the same simulation bound applies.

This classical simulation argument helps us define what it means for a
quantum--classical separation in communication to be optimal. Since
$\lVert\, \ket{\psi}\bra{\psi} - \ket{\psi_{i^*}}\bra{\psi_{i^*}}\, \rVert_1 \leq \eps$,
the distribution of any measurement outcome changes by at most $\eps/2$ in
total variation distance, so the simulation loses at most $\eps/2$ in success
probability. Hence for every problem $P$, every $\delta \geq 0$, and every $0 < \eps < 1$,
\begin{equation}
    \label{eq:communication-simulation}
    \Rtwoway_{\delta + \eps/2}(P) \;\leq\; \Roneway_{\delta + \eps/2}(P) \;\leq\; \left\lceil 2^{\Qoneway_{\delta}(P) + 1}\log_2(5/\eps) \right\rceil.
\end{equation}
In other words, for fixed error tolerances $\delta, \eps$, the classical communication cost of a problem is at most two to the power of its one-way quantum communication cost, up to a constant factor. We therefore say that a problem exhibits an \emph{optimal} quantum--classical separation if it admits a one-way quantum protocol using $n$ qubits, while every classical protocol requires $\Theta(2^n)$ bits of communication. Note that this simulation argument does not apply to interactive quantum communication protocols.

\subsection{Entropy, sampling, and coupling tools}
The following lemma shows how to couple an arbitrary random variable $\bW$ over $\pmone^s$ to a uniform random variable $\bU$ over $\pmone^s$, and gives an upper bound on the expected squared distance between the coupled random variables in terms of $H(\bW)$, the entropy of $\bW$. It can be seen as a uniform and Boolean-cube specialization of Marton's transportation inequality (see e.g., Theorem 8.2 of~\cite{boucheron2013concentration}). We give an elementary proof of the version we will need below.
\begin{lemma}
\label{lem:coupling-by-entropy}
For every random variable $\bW$ over $\pmone^s$, there exists a coupling $(\bW, \bU)$ of $\bW$ with  $\bU \sim \mathrm{Unif}(\pmone^s)$ such that
\[\E \lVert \bW-\bU\rVert_2^2 \leq 4\sqrt{s(s-H(\bW))}.  \]
\end{lemma}
\begin{proof}
    For $p \in [0,1]$, we can use the joint sampling rule shown in the table below to sample from $\pmone\times \pmone$ where one marginal distribution satisfies $\Pr[\bW = 1] = p$ and the other is uniform.
    \begin{center}
    \renewcommand{\arraystretch}{1.25}
    \begin{tabular}{c|cc|c}
     & $1$ & $-1$ & row sum\\ \hline
    $1$ & $\min(p,1/2)$ & $\max(p-1/2,0)$ & $p$\\
    $-1$ & $\max(1/2-p,0)$ & $\min(1-p,1/2)$ & $1-p$\\ \hline
    column sum & $1/2$ & $1/2$ & $1$
    \end{tabular}
    \end{center}
    We use this rule to give a procedure for sampling from a coupling $(\bW,\bU)$ coordinate-by-coordinate. Let $p_1 = \Pr[\bW_1 =1]$ and for $2\leq i\leq s$ let $p_i(w_1,\dots, w_{i-1}) = \Pr[\bW_i = 1 | \bW_1=w_1, \dots, \bW_{i-1}=w_{i-1}]$. We first sample $(w_1,u_1)$ using the above rule with $p=p_1$, then sample $(w_2,u_2)$ using $p = p_2(w_1)$, etc., until sampling $(w_s,u_s)$ using $p = p_s(w_1,\dots, w_{s-1})$. We check that both marginals are correct to verify that this is a coupling. Let $w = w_1\dots w_s$ and $u = u_1\dots u_s$. Looking at the column sums, $u_i = 1$ occurs with probability $1/2$ for each $i$, independent of prior samples, so each $u\in \pmone^s$ occurs with probability $2^{-s}$. Looking at the row sums, $w_i = 1$ occurs with probability $p_i(w_1,\dots, w_{i-1})$. By the chain rule, each $w\in \pmone^s$ occurs with probability \[
    \prod_{i=1}^s \Pr[\bW_i = w_i \mid \bW_1 = w_1,\dots,\bW_{i-1} = w_{i-1}] = \Pr[\bW = w].
    \]
    Write $\bp_i := p_i(\bW_1,\dots,\bW_{i-1})$ for the random variable of the $i$th coordinate's conditional bias. Conditioned on $\bW_1,\dots,\bW_{i-1}$, the pair $(\bW_i,\bU_i)$ is sampled from the table with
    $p = \bp_i$, and $\bW_i\neq\bU_i$ occurs exactly when an off-diagonal entry is sampled. Hence
    \[\Pr[\bW_i \neq \bU_i \mid \bW_1,\dots,\bW_{i-1}] = \max(\bp_i -1/2,0) + \max(1/2-\bp_i,0) = \lvert \bp_i - 1/2\rvert,\]
    and taking an expectation gives $\Pr[\bW_i\neq\bU_i] = \E\,\lvert \bp_i - 1/2\rvert$. Since $\bW_i,\bU_i\in\pmone$, we have $(\bW_i-\bU_i)^2 = 4\cdot\mathbbm{1}[\bW_i\neq\bU_i]$. Therefore 
    \[\E \lVert \bW - \bU\rVert_2^2 = \E \left[\sum_{i=1}^s (\bW_i - \bU_i)^2\right] = 4\sum_{i=1}^s \Pr[\bW_i\neq\bU_i] = 4\sum_{i=1}^s \E\,\lvert \bp_i - 1/2\rvert.\]
    Recall the chain rule for entropy:
    \[ H(\bW) = \sum_{i=1}^s H(\bW_i \mid \bW_1,\dots,\bW_{i-1}) = \sum_{i=1}^s \E\, H_2(\bp_i)\]
    where $H_2$ denotes binary entropy and note that $s - H(\bW) = s - \sum_{i=1}^s\E\, H_2(\bp_i) =  \sum_{i=1}^s \E[1 - H_2(\bp_i)].$ Now using the elementary fact that for $p\in [0,1]$,
    \[ 1 - H_2(p) \geq 2(p-1/2)^2,\]
    we obtain that 
    \begin{equation}
        s - H(\bW) \geq 2\sum_{i=1}^s \E[(\bp_i-1/2)^2].
        \label{eq:coupling-entropy}
    \end{equation}
    Finally, we have that
    \begin{align*}
        \E \lVert \bW - \bU\rVert_2^2 &= 4\sum_{i=1}^s \E\lvert \bp_i - 1/2\rvert\\
        &\leq 4\sqrt{s \sum_{i=1}^s \E [(\bp_i - 1/2)^2]}\\
        &\leq 4 \sqrt{\frac{s(s-H(\bW))}{2}}\\
        &\leq 4\sqrt{s(s-H(\bW))}
    \end{align*}
    where the second line is due to Cauchy--Schwarz and the third line uses (\ref{eq:coupling-entropy}).    
\end{proof}

The next two lemmas concern independently and uniformly sampled indices $\bI_1,\dots, \bI_k \in [N]$, and an associated reveal of the values of a random variable $\bW \in \pmone^N$ at those indices. The first lemma uses the revealed values to estimate the inner product between $\bW$ and another vector, and the second lemma builds upon \Cref{lem:coupling-by-entropy} to construct a coupling between $\bW$ and a uniform random variable that remains uniform even after conditioning on the revealed values of $\bW$.

\begin{lemma} \label{lem:sampling-estimate}
Let $\bW$ be a random variable over $\pmone^N$. Independently sample $\bI_1,\dots, \bI_k$ uniformly from $[N]$ with replacement. For any fixed $v \in \mathbb{R}^N$, let
\[ \wh{\bL}(v) := \frac{1}{k}\sum_{j=1}^k \wh{\bL}_j(v)  \quad \textrm{where}\quad \wh{\bL}_j(v)  =Nv_{\bI_j}\bW_{\bI_j}.\]
Then
\begin{equation}
    \label{eq:sampling-estimate}
    \E (\wh{\bL}(v) - \langle v,\bW\rangle)^2 \leq \frac{N}{k}\lVert v\rVert_2^2.
\end{equation}
\end{lemma}
\begin{proof}
    Fix $w\in\pmone^N$. For each $j\in[k]$,
    \[\E[\wh{\bL}_j(v)\mid \bW=w]
    =\sum_{i=1}^N\frac1N Nv_iw_i
    =\langle v,w\rangle,
    \qquad
    \E[\wh{\bL}_j(v)^2\mid \bW=w]
    =\sum_{i=1}^N\frac1N(Nv_iw_i)^2
    =N\|v\|_2^2.
    \]
    Hence $\E[\wh{\bL}(v)\mid\bW=w]=\langle v,w\rangle$ and
    \[
    \Var[\wh{\bL}_j(v)\mid\bW=w]
    =N\|v\|_2^2-\langle v,w\rangle^2
    \le N\|v\|_2^2.
    \]
    Since the $\wh{\bL}_j(v)$ are independent conditioned on $\bW=w$,
    \[
    \begin{aligned}
    \E\!\left[(\wh{\bL}(v)-\langle v,\bW\rangle)^2\mid\bW=w\right]
    &=\Var[\wh{\bL}(v)\mid\bW=w] \\
    &=\frac1{k^2}\sum_{j=1}^k\Var[\wh{\bL}_j(v)\mid\bW=w]
    \le \frac Nk\|v\|_2^2.
    \end{aligned}
    \]
    Averaging over $w$ proves~\eqref{eq:sampling-estimate}.
\end{proof}

\begin{lemma} \label{lem:coupling-after-reveal}
Let $\bW$ be a random variable over $\pmone^N$. Independently sample $\bI_1,\dots,\bI_k$ uniformly from $[N]$ with replacement, and let
\[ \bR=(\bI_1,\bW_{\bI_1},\dots,\bI_k,\bW_{\bI_k}), \qquad \bI=\{\bI_1,\dots,\bI_k\}.\]
Then there exists a coupling between $\bW$ and an auxiliary random variable $\bU$ preserving the joint law of $(\bW,\bR)$. Specifically, for every $R$ in the support of $\bR$,
\[ \bU \mid (\bR=R)\sim \mathrm{Unif}(\pmone^N), \]
and
\[ \E\sum_{i\notin\bI}(\bW_i-\bU_i)^2 \le 4\sqrt{ND_\infty(\bW)}, \qquad \E\|\bW-\bU\|_2^2 \le 4\sqrt{ND_\infty(\bW)}+2k. \]
\end{lemma}
\begin{proof}
    We view $\bR$ as a ``reveal'' of the values of $\bW$ at indices $\bI_1,\dots, \bI_k$. For a fixed reveal $R$, let $I(R) \subseteq [N]$ be the set of distinct indices appearing in $R$. Conditional on a reveal $\bR =R$, $\bW$ will have $s_R := N - |I(R)|$ unrevealed coordinates, while the coordinates in $I(R)$ are fixed. Let $\bW'$ be the random variable taking values in $\pmone^{s_R}$ obtained by restricting $\bW \mid (\bR =R)$ to its $s_R$ unrevealed coordinates. Since the revealed coordinates of $\bW$ are fixed, $H(\bW') = H(\bW \mid \bR =R)$.
    Now for each fixed $R$ we apply \Cref{lem:coupling-by-entropy} to $\bW'$. Then there exists a coupling of $\bW' \in \pmone^{s_R}$ with $\bU' \sim \mathrm{Unif}(\pmone^{s_R})$, denoted $(\bW', \bU')$, such that
    \begin{equation} \label{eq:coupling-after-reveal-1}
        \E \lVert \bW' - \bU' \rVert_2^2 \leq 4\sqrt{s_R(s_R-H(\bW'))} = 4\sqrt{s_R(s_R - H(\bW \mid \bR =R))}. 
    \end{equation}

    Index the coordinates of $\bW'$ and $\bU'$ by $[N]\setminus I(R)$. Extend $\bU'$ to a random string $\bU \in \pmone^N$ by setting $\bU_i = \bU'_i$ for $i \in [N]\setminus I(R)$, and sampling each $\bU_i$ for $i \in I(R)$ uniformly from $\pmone$, independently of each other and of $(\bW', \bU')$. Since $\bU'$ is uniform on $\pmone^{s_R}$, it follows that for all $R$ in the support of $\bR$,
    \[ \bU \mid (\bR = R) \sim \mathrm{Unif}(\pmone^N), \]
    as desired. Now, using the entropy chain rule, observe that
    \begin{align}
    H(\bW) &= H(\bW\mid\bI_1,\dots,\bI_k) \notag\\
    &= H(\bW_{\bI_1},\dots,\bW_{\bI_k}\mid\bI_1,\dots,\bI_k) + H(\bW\mid\bW_{\bI_1},\dots,\bW_{\bI_k},\bI_1,\dots,\bI_k) \notag\\
    &= H(\bW_{\bI_1},\dots,\bW_{\bI_k}\mid\bI_1,\dots,\bI_k) + H(\bW\mid\bR) \notag\\
    &\leq \E|\bI|+H(\bW\mid\bR). \label{eq:reveal-entropy}
    \end{align}
    since $\bW$ is independent of the sampled indices and, for each sampling $\bI_1,\dots, \bI_k$, the observed values have at most $2^{|\bI|}$ possible values so their conditional entropy is at most $|\bI|$. Averaging over the indices gives the upper bound $\E|\bI|$. Consequently,
    \begin{align}
    \sum_R\Pr[\bR=R](s_R-H(\bW\mid\bR=R))
    &= N-\E|\bI|-H(\bW\mid\bR) \notag\\
    &\leq N-H(\bW) \notag\\
    &\leq D_\infty(\bW). \label{eq:average-reveal-deficiency}
    \end{align}
    where the first line uses that $\E |\bI| = \sum_R \Pr[\bR = R]|I(R)|$ and $s_R = N - |I(R)|$, the second line uses \eqref{eq:reveal-entropy}, and the last line follows from Definition~\ref{def:min-entropy-def}.
    Given $\bR = R$, the unrevealed coordinates are exactly $[N] \setminus I(R)$, and on these coordinates $(\bW, \bU)$ is the coupled pair $(\bW', \bU')$. Hence
    \begin{align*}
        \E\sum_{i\notin\bI}(\bW_i-\bU_i)^2
        &= \sum_R \Pr[\bR=R]\; \E\Bigl[\sum_{i \notin I(R)} (\bW_i - \bU_i)^2 \Bigm| \bR = R\Bigr]
        = \sum_R \Pr[\bR=R]\; \E\lVert \bW' - \bU' \rVert_2^2 \\
        &\leq 4\sum_R\Pr[\bR=R]\sqrt{s_R(s_R-H(\bW\mid\bR=R))}\\
        &\leq 4\sqrt{N\sum_R\Pr[\bR=R](s_R-H(\bW\mid\bR=R))}\\
        &\leq 4\sqrt{ND_\infty(\bW)},
    \end{align*}
    where the first inequality is \eqref{eq:coupling-after-reveal-1}, the second uses $s_R \leq N$ and concavity of the square root, and the third is \eqref{eq:average-reveal-deficiency}.
    This gives the first coupling bound. Finally, conditional on $\bR=R$, each $\bW_i$ with $i\in I(R)$ is fixed, while $\bU_i$ is a uniform sign. Thus each revealed coordinate has expected squared discrepancy $2$, so
    \[ \E\sum_{i\in\bI}(\bW_i-\bU_i)^2 = 2\E|\bI| \leq 2k. \]
    Adding the contributions from the revealed and unrevealed coordinates yields
    \[ \E\|\bW-\bU\|_2^2 \leq 4\sqrt{ND_\infty(\bW)}+2k. \]
\end{proof}

The next lemma shows that a uniform random variable $\bU\sim\mathrm{Unif}(\pmone^N)$, which retains most of its entropy after conditioning on some jointly distributed random variable $\bW$, can be made to have small total conditional covariance by revealing a small, deterministically chosen subset of its coordinates. 

\begin{lemma}
\label{lem:entropy-budget-conditioning}
Let $\bU\sim\mathrm{Unif}(\pmone^N)$, and let $\bW$ be any finite random variable jointly distributed with $\bU$ such that
\[H(\bU\mid\bW)\geq N-c.\]
For every positive integer $k$, there exists a fixed set $I\subseteq[N]$ with $|I|\leq k$ such that, writing $\bU_I = (\bU_i)_{i\in I}$ and  $\bZ=(\bW,\bU_I)$,
\[\E \Big[ \sum_{\substack{i,j\in[N]\\ i\neq j}}\Cov(\bU_i,\bU_j\mid\bZ)^2 \Big] \leq \frac{2Nc}{k}.\]
Moreover, $H(\bU\mid\bZ)\geq N-c-k$.
\end{lemma}

\begin{proof}
For an arbitrary fixed set $J\subseteq[N]$, write $\bZ_J=(\bW,\bU_J)$ and define
\[\Gamma(J):=\left(\sum_{i=1}^N H(\bU_i\mid\bZ_J)\right)-H(\bU\mid\bZ_J).\]
By subadditivity of entropy, $\Gamma(J)\geq0$. Moreover,
\[\Gamma(\varnothing)\leq N-H(\bU\mid\bW)\leq c,\]
since $H((\bW,\bU_{\varnothing})) = H(\bW)$ and each coordinate $\bU_i$ has conditional entropy $H(\bU_i \mid \bW) \leq 1$.

We first compute how much $\Gamma$ decreases when one additional coordinate is revealed.
\[
\begin{aligned}
\Gamma(J)-\Gamma(J\cup\{j\})
&=\sum_{i=1}^N\Big(H(\bU_i\mid\bZ_J)-H(\bU_i\mid\bZ_J,\bU_j)\Big)-\Big(H(\bU\mid\bZ_J)-H(\bU\mid\bZ_J,\bU_j)\Big) \\
&=\left(\sum_{i=1}^N I(\bU_i;\bU_j\mid\bZ_J)\right)-H(\bU_j\mid\bZ_J) \\
&=\sum_{i:\,i\neq j} I(\bU_i;\bU_j\mid\bZ_J) \\
&\geq \frac12\sum_{i:\,i\neq j}\E\left[\Cov(\bU_i,\bU_j\mid\bZ_J)^2\right],
\end{aligned}
\]
where the second equality uses the definition of conditional mutual information and the chain rule, the third uses $I(\bU_j;\bU_j\mid\bZ_J)=H(\bU_j\mid\bZ_J)$, and the final inequality uses \Cref{lem:cov-mutual-information}. 
Summing over $j\in[N]$ gives
\[\sum_{j=1}^N\big(\Gamma(J)-\Gamma(J\cup\{j\})\big)\geq\frac12\E\sum_{\substack{i,j\in[N]\\ i\neq j}}\Cov(\bU_i,\bU_j\mid\bZ_J)^2.\]
Suppose
\[\E\sum_{\substack{i,j\in[N]\\ i\neq j}}\Cov(\bU_i,\bU_j\mid\bZ_J)^2>\frac{2Nc}{k}.\]
Then
\[\sum_{j=1}^N\big(\Gamma(J)-\Gamma(J\cup\{j\})\big)>\frac{Nc}{k}.\]
Hence, by averaging over the $N$ coordinates, there exists some $j\in[N]$ such that $\Gamma(J)-\Gamma(J\cup\{j\})>\frac{c}{k}$.
Moreover, such a $j$ cannot already lie in the set $J$, since if $j\in J$ then $\bU_j$ is already determined by $\bZ_J = (\bW, \bU_J)$, so $\Gamma(J\cup\{j\})=\Gamma(J)$.

Starting from $J=\varnothing$, repeatedly add such a coordinate whenever the covariance bound fails. This process cannot continue for $k$ steps: otherwise $\Gamma$ would decrease by more than $c$, contradicting $\Gamma(\varnothing)\leq c$ and $\Gamma(J)\geq0$. Hence it stops at some fixed set $I$ with $|I|\leq k$ satisfying
\[\E\sum_{i\neq j}\Cov(\bU_i,\bU_j\mid\bW,\bU_I)^2\leq\frac{2Nc}{k}.\]
Finally, revealing $|I|$ signs can decrease conditional entropy by at most $|I|$, so
\[H(\bU\mid\bW,\bU_I)\geq H(\bU\mid\bW)-|I|\geq N-c-k.\]
At each step, the next coordinate is chosen using only the deterministic quantities $\Gamma(J)-\Gamma(J\cup\{j\})$, which are determined by the joint distribution of $(\bU,\bW)$. Thus the resulting set $I$ is deterministically chosen.
\end{proof}

\subsection{Gaussian window tools}

We will use three numerical estimates for a standard Gaussian. These estimates follow by direct evaluation and elementary optimization of the standard Gaussian density and cumulative distribution function.

\begin{fact}[Gaussian estimates]\label{fact:gaussian}
Let $Z \sim \calN(0,1)$ be a standard Gaussian random variable, and set $a := \sqrt{2/3}$ and $p_* := 0.6 < 2/3$. Then
\[
    \E\bigl[Z^2 \mathbbm{1}_{\{|Z| \leq a\}}\bigr] < \frac{1}{8}, \qquad
    \E\bigl[Z^2 \mathbbm{1}_{\{|Z| > 3a\}}\bigr] < \frac{1}{8}, \qquad
    \sup_{b \in \R} \Pr\bigl[a \leq |b + Z| \leq 3a\bigr] < p_*.
\]
\end{fact}
The next lemma shows that the third estimate in Fact~\ref{fact:gaussian} holds even when $Z$ is replaced by a normalized sum of $s$ independent uniform signs, for all sufficiently large $s$. The sum may be shifted by an independent random offset and perturbed by a small error term.

\begin{lemma}\label{lem:independent-noise}
For each positive integer $s$, let $\bT = \bA + \bS + \berr$, where $\bA$ and $\bS$ are independent,
$\bS$ is distributed as $s^{-1/2}\sum_{i=1}^s \boldsymbol{\xi}_i$ for
independent uniform signs $\boldsymbol{\xi}_i \in \pmone$, and $\E[\berr^2] = o(1)$
as $s \to \infty$. The error term $\berr$ may depend on $\bA$ and $\bS$.
Then for all sufficiently large $s$,
\[
    \Pr\bigl[a \leq |\bT| \leq 3a\bigr] \leq p_*.
\]
\end{lemma}
\begin{proof}
Let $Z \sim \calN(0,1)$ and $q := \sup_{b \in \R}\Pr[a \leq |b + Z| \leq 3a]$,
so $q < p_*$ by Fact~\ref{fact:gaussian}. Fix a small constant $\delta > 0$ with
$q + 4\delta < p_*$.

By Markov's inequality, $\Pr[|\berr| > \delta] \leq \E[\berr^2]/\delta^2 = o(1)$,
and when $|\berr| \leq \delta$ the event $a \leq |\bT| \leq 3a$ implies
$a - \delta \leq |\bA + \bS| \leq 3a + \delta$. Conditioning on $\bA = b$, which
is independent of $\bS$, it therefore suffices to bound
$\Pr[a - \delta \leq |b + \bS| \leq 3a + \delta]$ uniformly in $b$. This event
is a union of at most two intervals, so by the Berry--Esseen theorem its probability is
within $O(s^{-1/2})$ of the corresponding Gaussian probability
$\Pr[a - \delta \leq |b + Z| \leq 3a + \delta]$. Since the density of $Z$ is at
most $1$, that Gaussian probability is at most $q + 4\delta$. Altogether,
\[
    \Pr[a \leq |\bT| \leq 3a] \leq q + 4\delta + O(s^{-1/2}) + o(1),
\]
which is at most $p_*$ for all sufficiently large $s$.
\end{proof}

We now give a lemma that relates the Pelagic condition of \emph{Pelagic Fourier Fishing}, which seeks a coefficient with roughly balanced total squared Fourier mass above and below its magnitude, to a simpler condition based on a fixed window of accepted magnitudes. We do not use this simpler window condition in \emph{Pelagic Fourier Fishing} itself, since for some inputs there would be no coefficient that falls into the window, preventing the problem from being total.

\begin{lemma}\label{lem:magnitude-to-fourier-mass}
Let $s=2^t$, and consider $t\to \infty$. Let $(\bW,\bU)$ be coupled random variables over $\pmone^s$ such that $\bU\sim\mathrm{Unif}(\pmone^s)$ and
\[ \frac{1}{s}\E\|\bW-\bU\|_2^2=o(1). \]
Regard $\bW$ as the truth table of a Boolean function $w:\zo^t\to\pmone$. Then with probability $1-o(1)$, every $y\in\zo^t$ satisfying the Pelagic condition:
\[ \sum_{z:|\wh{w}(z)|\leq|\wh{w}(y)|}|\wh{w}(z)|^2\geq\frac18,\qquad \sum_{z:|\wh{w}(z)|\geq|\wh{w}(y)|}|\wh{w}(z)|^2\geq\frac18 \]
also satisfies the window condition:
\[ \frac{a}{\sqrt{s}}\leq|\wh{w}(y)|\leq\frac{3a}{\sqrt{s}}\]
where $a = \sqrt{2/3}$.
\end{lemma}
\begin{proof}
For a magnitude cutoff $r > 0$, let $q_r(x) := x^2\,\mathbbm{1}_{\{|x| \leq r\}}$, let $Z \sim \calN(0,1)$ be a standard Gaussian, and let $\mu_r := \E[q_r(Z)]$. For a function $v: \zo^t \to \pmone$, let
\[ L_r(v) := \sum_{y : |\wh{v}(y)| \leq r/\sqrt{s}} \wh{v}(y)^2 \]
be the squared Fourier mass of $v$ at or below level $r/\sqrt{s}$. By Parseval, the mass strictly above this level is $1 - L_r(v)$.

The proof has two parts. First, we show that for the uniformly random function $u: \zo^t \to \pmone$ with truth table $\bU$, the quantity $L_r(u)$ is concentrated around its Gaussian value $\mu_r$ for each fixed $r>0$. Second, we use the coupling to show that $w$ and $u$ are close in Fourier $\ell_2$ distance, and deduce that $L_a(w) < 1/8$ and $1 - L_{3a}(w) < 1/8$ for $a = \sqrt{2/3}$ with probability $1 - o(1)$. This implies the lemma. Write
\[G(y) := \sqrt{s}\,\wh{u}(y) = \frac{1}{\sqrt{s}}\sum_{x\in \zo^t}u(x)\chi_y(x)\] 
for the rescaled Fourier coefficients of the uniformly random function $u$, so that $L_r(u) = \frac{1}{s}\sum_y q_r(G(y))$. We show that for each fixed $r > 0$, $|L_r(u) - \mu_r| = o(1)$ with probability $1 - o(1)$, by bounding the mean and variance of $L_r(u)$. Each $G(y)$ is distributed as $s^{-1/2}\sum_{i=1}^s \boldsymbol{\xi}_i$ for independent uniform signs $\boldsymbol{\xi}_i \in \pmone$ due to each value of $u(x)$ being independently uniform in $\pmone$. By the central limit theorem, $G(y)$ converges in distribution to $Z$. Since $q_r$ is bounded and its only discontinuities are at $\pm r$, which have Gaussian probability zero, $\E[q_r(G(y))]=\mu_r+o(1)$. Hence $\E[L_r(u)]= \E [\frac{1}{s}\sum_y q_r(G(y))] = \mu_r+o(1)$.

For the variance, fix distinct labels $y \neq z$. The characters $\chi_y$ and $\chi_z$ agree on a set $A = \{ x\in \zo^t: \chi_y(x)=\chi_z(x)\}$ containing exactly half of $\zo^t$. Let $\bV$ and $\bV'$ be the sums of $u(x)\chi_y(x)$ over $x \in A$ and over $x \in \zo^t \setminus A$ respectively, each divided by $\sqrt{s/2}$. These are independent normalized sums of $s/2$ uniform signs, and since $\chi_z = \chi_y$ on $A$ and $\chi_z = -\chi_y$ off $A$,
\[ \bigl(G(y), G(z)\bigr) = \frac{1}{\sqrt{2}}\bigl(\bV + \bV',\ \bV - \bV'\bigr). \]
Let $Z_1, Z_2$ be independent standard Gaussians. By independence and the central limit theorem, $(\bV, \bV')$ converges in distribution to $(Z_1, Z_2)$. Observe that $(Z_1, Z_2)$ has joint density $\frac{1}{2\pi}e^{-(z_1^2 + z_2^2)/2}$, which depends only on $z_1^2 + z_2^2$, and the linear map $(z_1, z_2) \mapsto \bigl(\frac{z_1 + z_2}{\sqrt{2}}, \frac{z_1 - z_2}{\sqrt{2}}\bigr)$ preserves $z_1^2 + z_2^2$ and has determinant of absolute value $1$, so it leaves this density unchanged. Therefore $(G(y),G(z)) = \frac{1}{\sqrt{2}}(\bV + \bV', \bV - \bV')$ also converges in distribution to $(Z_1, Z_2)$. The function $(v,v')\mapsto q_r(v)q_r(v')$ is bounded, and its discontinuities lie on the lines $v=\pm r$ or $v'=\pm r$, which have probability zero under the limiting Gaussian law. Its expectations therefore converge. So for $y\neq z$,
\begin{equation} \label{eq:asymptotic-indep-qrG}
\E\bigl[q_r(G(y))\,q_r(G(z))\bigr] = \E[q_r(Z_1)]\,\E[q_r(Z_2)] + o(1) = \mu_r^2 + o(1). \end{equation}
The joint law of $(G(y),G(z))$ is identical for every distinct pair $y,z$, since the two independent half-sums always have the same laws. Thus the $o(1)$ term in~\eqref{eq:asymptotic-indep-qrG} is uniform over these pairs. Expanding the square in $\E[L_r(u)^2]$, we see that
\[ \E[L_r(u)^2] = \frac{1}{s^2}\sum_{y \in \zo^t}\sum_{z\in \zo^t} \E\bigl[q_r(G(y))\,q_r(G(z))\bigr] \leq \frac{r^4}{s} + \mu_r^2 + o(1) \]
where the inequality follows from bounding each of the $s$ diagonal terms by noting that $q_r(x) \leq r^2$, and bounding each of the $s^2-s$ off-diagonal terms by~\eqref{eq:asymptotic-indep-qrG}.
Hence $\mathrm{Var}[L_r(u)] = \E[L_r(u)^2] - \E[L_r(u)]^2 \leq r^4/s + o(1) = o(1)$, and Chebyshev's inequality gives $|L_r(u) - \mu_r| = o(1)$ with probability $1 - o(1)$. This completes the first part.

By Parseval's identity applied to $w-u$,
\[ \sum_{y\in\zo^t}(\wh{w}(y)-\wh{u}(y))^2 = \frac1s\|\bW-\bU\|_2^2. \]
Taking expectations and using the lemma's coupling assumption,
\[ \E\sum_{y\in\zo^t}(\wh{w}(y)-\wh{u}(y))^2= \frac1s\E \|\bW-\bU\|_2^2 = o(1). \] 
We can then use Markov's inequality to obtain that the event $\sum_{y} (\wh{w}(y) - \wh{u}(y))^2 = o(1)$ occurs with probability $1 - o(1)$. Fix a constant $0<\eps<1$, to be chosen below, and work on this high-probability event. Call a $y$ \emph{bad} if 
\[ | \wh{w}(y) - \wh{u}(y) | > \eps | \wh{w}(y)| \]
and \emph{good} otherwise. The total squared Fourier mass of bad coefficients is
\[ \sum_{y \textrm{ bad}} |\wh{w}(y)|^2 \leq \frac{1}{\eps^2} \sum_{y \textrm{ bad}} | \wh{w}(y) - \wh{u}(y)|^2 \leq \frac{1}{\eps^2} \sum_{y \in \zo^t} | \wh{w}(y) - \wh{u}(y)|^2 = o(1). \]
For every good $y$, $(1-\eps)|\wh{w}(y)| \leq | \wh{u}(y)| \leq (1+\eps) |\wh{w}(y)|$. Hence for $a = \sqrt{2/3}$,
\begin{align*} 
    L_a(w) = \sum_{y:\, |\wh{w}(y)| \leq a/\sqrt{s}} \wh{w}(y)^2
    &\leq o(1)+\sum_{\substack{y\text{ good}:\\ |\wh{w}(y)|\leq a/\sqrt{s}}}\wh{w}(y)^2\\
    &\leq o(1)+\frac{1}{(1-\eps)^2}\sum_{\substack{y\text{ good}:\\ |\wh{w}(y)|\leq a/\sqrt{s}}}\wh{u}(y)^2\\
    &\leq o(1)+\frac{1}{(1-\eps)^2}L_{(1+\eps)a}(u)
\end{align*}
where the first inequality uses that bad $y$ have total squared Fourier mass $o(1)$, the second that for good $y$, $|\wh{w}(y)| \leq \frac{|\wh{u}(y)|}{1-\eps}$, and the third that for good $y$, if $|\wh{w}(y)| \leq a/\sqrt{s}$, then $|\wh{u}(y)| \leq (1+\eps)a/\sqrt{s}$. Similarly, we can show that
\begin{align*}
    1-L_{3a}(w) 
    &\leq o(1)+\frac{1}{(1-\eps)^2}\bigl(1-L_{(1-\eps)3a}(u)\bigr).
\end{align*}
By Fact~\ref{fact:gaussian}, $\mu_a<1/8$ and $1-\mu_{3a}<1/8$. We now choose $\eps>0$ sufficiently small such that
\[ \frac{\mu_{(1+\eps)a}}{(1-\eps)^2}<\frac18,\qquad \frac{1-\mu_{(1-\eps)3a}}{(1-\eps)^2}<\frac18. \]
By the first part of the proof, $L_{(1+\eps)a}(u)=\mu_{(1+\eps)a}+o(1)$ and $L_{(1-\eps)3a}(u)=\mu_{(1-\eps)3a}+o(1)$ with probability $1-o(1)$. The two preceding inequalities imply that, with probability $1-o(1)$,
\[ L_a(w) = \sum_{y:\, |\wh{w}(y)| \leq a/\sqrt{s}} \wh{w}(y)^2 <\frac18,\qquad 1-L_{3a}(w) = \sum_{y:\, |\wh{w}(y)| > 3a/\sqrt{s}} \wh{w}(y)^2<\frac18. \]
Therefore, if $|\wh{w}(y)|\le a/\sqrt{s}$, the first Pelagic Fourier Fishing inequality fails, while if $|\wh{w}(y)|>3a/\sqrt{s}$, the second fails. Thus every successful $y$ satisfies
\[ \frac{a}{\sqrt{s}} \leq |\wh{w}(y)|\le\frac{3a}{\sqrt{s}}, \]
which proves the lemma.
\end{proof}

\section{Pelagic Fourier Fishing} \label{sec:PFF}
Our first problem, which we call \emph{Pelagic Fourier Fishing}, is a distributed variant of the \emph{Fourier Fishing} problem introduced by Aaronson~\cite{Aaronson10BQPandPH} and further studied by Aaronson and Chen~\cite{AaronsonChen2016}. Conceptually, this variant asks for the players to ``fish'' for a Fourier coefficient that is neither too small nor too large in magnitude, rather than asking for a Fourier coefficient with magnitude above a certain threshold as in \emph{Fourier Fishing}.  

Rather than fixing upper and lower thresholds in advance, we define success relative to the Fourier spectrum of the input. For a coefficient to be acceptable, the total squared Fourier mass of coefficients with lesser than or equal magnitude and the total squared Fourier mass of coefficients with greater than or equal magnitude must both be sufficiently large. In other words, it asks for an approximate median of the Fourier spectrum weighted by squared Fourier mass. We call this the \emph{Pelagic condition}. 

We establish the definitions of the Fourier-analytic notation we use in Section~\ref{sub:fourier-analysis}.
\begin{tcolorbox}[problembox]
\begin{problem}[Pelagic Fourier Fishing, $\PFF_n$]
\label{problem:FM}
Let $N=2^n$. Alice receives a function $f: \zo^n \to \pmone$ and Bob receives a function $g: \zo^n \to \pmone$, each given as a truth table of length $N$. Let $h(x) = f(x)g(x)$ for $x\in \zo^n$. Bob must output a $y \in \zo^n$ satisfying
\begin{equation}
    \sum_{z: |\widehat{h}(z)|\leq |\widehat{h}(y)|} |\widehat{h}(z)|^2 \geq \frac{1}{8}, \qquad \sum_{z: |\widehat{h}(z)|\geq |\widehat{h}(y)|} |\widehat{h}(z)|^2 \geq \frac{1}{8}.
\end{equation}
\end{problem}
\end{tcolorbox}

Note that sending Alice's entire truth table gives an exact $2^n$-bit classical one-way protocol. Our first result is that \emph{Pelagic Fourier Fishing} has an optimal separation between one-way quantum communication complexity and two-way randomized communication complexity:

\begin{theorem}[Restatement of \Cref{thm:intro-two-way}]
\label{thm:PFF-combined}
For all $n\in \mathbb{N}$, $\PFF_n$ admits an $n$-qubit quantum one-way communication protocol that succeeds with probability at least $3/4$ on every input. Meanwhile, for all sufficiently large $n$, every randomized public-coin (two-way) communication protocol for $\PFF_n$ with success probability at least $2/3$ on average under a uniform input distribution communicates $\Omega(2^n)$ bits. In particular,
\[ \Qoneway(\PFF_n) \leq n \qquad \textrm{and} \qquad \Rtwoway(\PFF_n) = \Theta(2^n).\]
\end{theorem}
\subsection{Quantum upper bound}
We now show that $\PFF_n$ admits a simple quantum one-way communication protocol using $n = \log_2 N$ qubits that succeeds with probability at least $3/4$ on all inputs. Alice prepares and sends the $n$-qubit binary phase state 
    \[ \ket{\psi_f} = \frac{1}{\sqrt{N}}\sum_{x \in \zo^n} f(x)\ket{x}.\]
Bob applies the diagonal unitary $U_g$ which maps $\ket{x} \mapsto g(x)\ket{x}$, then applies Hadamard gates to all $n$ qubits and measures in the computational basis. The sequence of transformations is
\begin{align*}
    \ket{\psi_f} =\frac{1}{\sqrt{N}}\sum_{x \in \zo^n} f(x)\ket{x} &\overset{U_g}{\longmapsto} \frac{1}{\sqrt{N}}\sum_{x \in \zo^n} f(x)g(x)\ket{x} \\
    &\overset{H^{\otimes n}}{\longmapsto} \frac{1}{N} \sum_{y \in \zo^n} \sum_{x \in \zo^n} f(x)g(x) (-1)^{x \cdot y} \ket{y}\\
    &= \sum_{y \in \zo^n} \widehat{h}(y)\ket{y}
\end{align*} 
Measuring in the computational basis gives $y$ with probability $| \widehat{h}(y)|^2$. Consider the total squared Fourier mass of all $y\in \zo^n$ which satisfy the \emph{Pelagic condition} of \emph{Pelagic Fourier Fishing}, which we restate here:
\[ \sum_{z: |\widehat{h}(z)|\leq |\widehat{h}(y)|} |\widehat{h}(z)|^2 \geq \frac{1}{8}, \qquad \sum_{z: |\widehat{h}(z)|\geq |\widehat{h}(y)|} |\widehat{h}(z)|^2 \geq \frac{1}{8}.\]
Let $L\subseteq \zo^n$ be the set of $y$ violating the first inequality, and let $y^*\in L$ maximize $|\wh{h}(y)|$ over $L$ (if $L$ is nonempty; otherwise $L$ has zero squared Fourier mass). Every $y\in L$ has $|\wh{h}(y)| \leq |\wh{h}(y^*)|$, so
\[ \sum_{y\in L} \lvert \wh{h}(y)\rvert^2 \leq \sum_{z: |\wh{h}(z)| \leq |\wh{h}(y^*)|} \lvert \wh{h}(z)\rvert^2 < \frac{1}{8}.\]
Symmetrically, we can argue that the set of $y$ violating the second inequality also has total squared Fourier mass at most $1/8$. Parseval's identity says that for all $h: \zo^n \to \pmone$,
\[ \sum_{y \in \zo^n} |\widehat{h}(y)|^2 = 1.\]
The total squared Fourier mass of all $y$ which satisfy both inequalities is therefore at least $1 - 1/8 - 1/8 = 3/4$. Sampling $y$ according to its squared Fourier mass therefore fishes a correct output with probability at least $3/4$.

\subsection{Randomized two-way lower bound}
\label{subsec:pff-lb}
In this section we prove the $\Omega(2^n)$-bit randomized two-way communication lower bound in \Cref{thm:PFF-combined}. 

\subsubsection{Key technical lemma}
The key technical lemma needed for this lower bound is stated as follows:

\begin{lemma}[High-entropy inner-product decomposition] \label{lem:inner-prod-decomp}
Let $\bX$ and $\bY$ be independent random variables over $\pmone^N$ such that their combined min-entropy deficiencies satisfy the following condition:
\[ D_\infty(\bX) + D_\infty(\bY) = o(N).\]
Then we can decompose the normalized inner product of $\bX$ and $\bY$ as the sum of three random variables:
\[ \frac{\langle \bX, \bY \rangle }{\sqrt{N}} = \bA + \bS + \textbf{err},\]
where $\bA$ and $\bS$ are independent, $\bS$ is distributed as $N^{-1/2}\sum_{i=1}^N \boldsymbol{\xi}_i$ for independent uniform signs $\boldsymbol{\xi}_i \in \pmone$, and $\E[\berr^2] = o(1)$. 
\end{lemma}

\begin{proof}
    We apply the coupling construction of \Cref{lem:coupling-by-entropy} to $\bX$ to obtain a coupling $(\bX,\bU)$ between $\bX$ and $\bU \sim \mathrm{Unif}(\pmone^N)$ such that, taking an expectation over the joint distribution of the coupling $(\bX,\bU)$, we have
    \[ \E \| \bX - \bU \|_2^2 \leq 4\sqrt{N(N-H(\bX))}.\]
    By the min-entropy deficiency condition, we have that $N - H(\bX) \leq D_\infty(\bX) = o(N)$, hence 
    \begin{equation} \E \| \bX - \bU \|_2^2 = o(N). \end{equation}
    So when sampling from $(\bX,\bU)$, $\bX = \bU + (\bX - \bU)$ where $\bU$ is perfectly uniform and $\bX - \bU$ has small expected squared norm. Now set an integer parameter
    \[ k = \Bigl\lceil\sqrt{N\cdot \left(1+ \E\| \bX - \bU \|_2^2\right)} \ \Bigr\rceil.\]
    This choice of $k$ is asymptotically at an intermediate scale between $\E\| \bX - \bU \|_2^2$ and $N$: we have that $\E\| \bX - \bU \|_2^2 = o(k)$ and $k = o(N)$.

    Independently of the first coupling $(\bX,\bU)$, sample $k$ indices $\bI_1,\dots,\bI_k \in [N]$ independently and uniformly at random and reveal the values of $\bY$ at those indices. The information revealed is captured by a random variable $\bR = (\bI_1, \bY_{\bI_1}, \dots, \bI_k, \bY_{\bI_k})$. Then \Cref{lem:coupling-after-reveal} asserts that there exists a coupling between $\bY$ and an auxiliary random variable $\bU'$ such that for every fixed $R$ in the support of $\bR$,
    \begin{equation}\label{eq:uniformity} \bU' \mid (\bR = R) \sim \mathrm{Unif}(\pmone^N)\end{equation}
    and
    \[ \E \| \bY - \bU'\|_2^2 \leq 4\sqrt{ND_\infty(\bY)} + 2k.\]
    But by the min-entropy deficiency condition, $D_\infty(\bY) = o(N)$ and $k = o(N)$, so
    \begin{equation}\label{eq:Y-U'-small} \E \| \bY - \bU'\|_2^2 = o(N). \end{equation}
    Also define 
    \[\wh\bL(v) := \frac{1}{k}\sum_{j=1}^k \wh\bL_j(v)\quad \textrm{where}\quad \wh\bL_j(v) = Nv_{\bI_j}\bY_{\bI_j}\]
    for all $v\in \R^N$ as in \Cref{lem:sampling-estimate}.
    Now we define $\bA$, $\bS$, and $\berr$, the random variables in the decomposition:
    \begin{equation}
        \bA := \frac{\wh\bL(\bX - \bU)}{\sqrt{N}}, 
        \qquad \bS := \frac{\langle \bU, \bU'\rangle}{\sqrt{N}},
        \qquad \berr := \underbrace{\frac{\langle \bX -\bU, \bY\rangle - \wh\bL(\bX-\bU)}{\sqrt{N}}}_{\berr_1}  + \underbrace{\frac{\langle \bU, \bY - \bU'\rangle}{\sqrt{N}}}_{\berr_2}.
    \end{equation}
    One can verify that
    \[ \bA + \bS + \berr = \frac{\langle \bX,\bY \rangle}{\sqrt{N}}\]
    using the bilinearity of the inner product $\langle \cdot, \cdot \rangle$. Now observe that if we condition on $\bX$, $\bU$, and $\bR$, then $\bA$ is completely determined, while $\bU'$ is distributed uniformly over $\pmone^N$ due to the property \eqref{eq:uniformity} from its construction. Therefore $\langle \bU, \bU'\rangle$ is a sum of independent and uniform signs, and $\bS$ is distributed as $N^{-1/2}\sum_{i=1}^N \boldsymbol{\xi}_i$ for independent uniform signs $\boldsymbol{\xi}_i \in \pmone$. Thus, $
    \bS$ is independent of $\bX,\bU,\bR$, and hence independent of $\bA$ as well.

    Finally, we show that $\E[\berr^2] = o(1)$ by bounding the two terms in $\berr$ separately. By construction, $\bX-\bU$ is independent of $(\bY,\bI_1,\dots,\bI_k)$, so conditioned on $(\bX-\bU)=v$, \Cref{lem:sampling-estimate} gives
    \[ \E[(\langle v,\bY\rangle-\wh\bL(v))^2\mid(\bX-\bU)=v]\leq \frac{N}{k}\|v\|_2^2. \]
    Therefore
    \[ \E[\berr_1^2]\leq \frac{1}{k}\E\|\bX-\bU\|_2^2=o(1), \]
    where the last equality follows from our choice of $k$ such that $\E\|\bX-\bU\|_2^2 = o(k)$.

    For $\berr_2$, the random variable $\bU$ is uniform and independent of $(\bY,\bU')$. For every fixed $v\in\R^N$,
    \[\E_{\bU}\langle \bU,v\rangle^2= \E_{\bU}\sum_{i,j}\bU_i\bU_jv_iv_j = \sum_i v_i^2 = \|v\|_2^2, \]
    since $\E[\bU_i\bU_j]=0$ for $i\neq j$. Conditioning on $\bY-\bU'$ therefore gives
    \[ \E[\berr_2^2] = \frac1N\E\|\bY-\bU'\|_2^2 = o(1). \]
    due to \eqref{eq:Y-U'-small}. Finally, using $(x+y)^2\leq 2x^2+2y^2$,
    \[ \E[\berr^2]\leq 2\E[\berr_1^2]+2\E[\berr_2^2]=o(1), \]
    which completes the proof.
\end{proof}

\subsubsection{Proof of \cref{thm:PFF-combined}}

We now complete the proof of the randomized two-way communication lower bound of~\Cref{thm:PFF-combined}.

\begin{proof}[Proof of the randomized lower bound in~\Cref{thm:PFF-combined}]
    Let $N=2^n$. For functions $f,g:\zo^n\to\pmone$, write $F,G\in\bools$ for their truth tables. Also write $\bchi_y\in\bools$ for the truth table of $\chi_y(x)=(-1)^{x\cdot y}$, and let $h=fg$ denote the pointwise product of the two functions.
    
    By Yao's principle~\cite{yao1983lower}, it suffices to show that all deterministic protocols using $c=o(N)$ bits of communication succeed with probability less than $2/3$ under a hard input distribution, for all sufficiently large $n$. Consider the distribution where Alice's and Bob's truth tables $\bF$ and $\bG$ are independent and uniform on $\bools$. Fix an arbitrary deterministic $c$-bit two-way communication protocol with $c=o(N)$. We will prove that on this distribution it succeeds with probability less than $2/3$, which implies the lower bound.
    
    The protocol has at most $2^c$ transcript leaves. For each nonempty leaf $\tau$, the inputs producing $\tau$ form a combinatorial rectangle
    \[ A_\tau\times B_\tau\subseteq\bools\times\bools. \]
    Let $\alpha_\tau:=|A_\tau|/2^N$ and $\beta_\tau:=|B_\tau|/2^N$, so the probability of transcript $\tau$ (on the uniform input distribution) is $\alpha_\tau\beta_\tau$. Starting from the uniform input distribution and conditioning on seeing a transcript leaf $\tau$, Alice's and Bob's input distributions are independent, and conditionally distributed as
    \[ \bF_\tau\sim\mathrm{Unif}(A_\tau),\qquad \bG_\tau\sim\mathrm{Unif}(B_\tau). \]
    Call a transcript $\tau$ \emph{rare} if $\alpha_\tau\beta_\tau<2^{-c}/N$. There are at most $2^c$ transcripts, so rare transcripts have total probability at most $1/N$. For every non-rare transcript, $\alpha_\tau\beta_\tau \geq 2^{-c}/N = 2^{-(c+n)}$, so
    \[ D_\infty(\bF_\tau)+D_\infty(\bG_\tau)=\log_2\frac1{\alpha_\tau\beta_\tau}\leq c+n. \]
    Since he is deterministic, Bob's output on seeing transcript $\tau$ is determined by a function of his input $y_\tau:B_\tau\to\zo^n$. Define
    \[ V_\tau(G):=G\odot\bchi_{y_\tau(G)}, \]
    where $\odot$ denotes the entrywise product. The random variables $\bF_\tau$ and $V_\tau(\bG_\tau)$ are independent, since $V_\tau(\bG_\tau)$ is a function of $\bG_\tau$. We next show that applying this function increases the min-entropy deficiency by at most $n$. That is,
    \[ D_\infty(V_\tau(\bG_\tau))\leq D_\infty(\bG_\tau)+n. \]
    Fix $v\in\bools$. If $V_\tau(G)=v$, then $G=v\odot\bchi_{y_\tau(G)}$. Each possible output label in the range of $y_\tau:B_\tau \to \zo^n$ determines at most one such $G$, and there are only $2^n=N$ output labels. Thus $v$ has at most $N$ pre-images under $V_\tau$. Since $\bG_\tau$ is uniform on $B_\tau$,
    \[ \max_v\Pr[V_\tau(\bG_\tau)=v]\leq\frac{N}{|B_\tau|}. \]
    By the definition of min-entropy deficiency, this gives
    \[ D_\infty(V_\tau(\bG_\tau))\leq\log_2\frac{2^N N}{|B_\tau|}=D_\infty(\bG_\tau)+n. \]
    Consequently, every non-rare transcript satisfies
    \[ D_\infty(\bF_\tau)+D_\infty(V_\tau(\bG_\tau))\leq c+2n=o(N). \]
    Let $a=\sqrt{2/3}$, and let $\mathcal{W}$ be the event that Bob's output $y$ satisfies the ``window condition'':
    \[ a\leq\sqrt N\,|\wh h(y)|\leq3a. \]
    For inputs $F,G$ producing transcript $\tau$, the Fourier coefficient selected by Bob can be written as
    \[ \wh h(y_\tau(G))=\frac1{N}\sum_{x\in\zo^n}f(x)g(x)\chi_{y_\tau(G)}(x)=\frac{\langle F,V_\tau(G)\rangle}{N}. \]
    Therefore
    \[ \Pr[\mathcal{W}\mid\tau]=\Pr\left[a\leq\frac{|\langle\bF_\tau,V_\tau(\bG_\tau)\rangle|}{\sqrt N}\leq3a\right]. \]
    
    Choose a non-rare transcript $\tau_*$ maximizing $\Pr[\mathcal{W}\mid\tau]$ among non-rare transcripts. Since rare transcripts have total probability at most $1/N$,
    \[ \Pr[\mathcal{W}]\leq\frac1N+\Pr[\mathcal{W}\mid\tau_*]. \]
    The independent pair $\bF_{\tau_*},V_{\tau_*}(\bG_{\tau_*})$ has combined min-entropy deficiency at most $c+2n=o(N)$. Applying \Cref{lem:inner-prod-decomp}, we obtain a decomposition
    \[ \bT:=\frac{\langle\bF_{\tau_*},V_{\tau_*}(\bG_{\tau_*})\rangle}{\sqrt N}=\bA+\bS+\berr, \]
    where $\bA$ and $\bS$ are independent, $\bS$ is distributed as $N^{-1/2}\sum_{i=1}^N\boldsymbol{\xi}_i$ for independent uniform signs $\boldsymbol{\xi}_i\in\pmone$, and $\E[\berr^2]=o(1)$. \Cref{lem:independent-noise}, applied with $s=N$, therefore gives, for all sufficiently large $n$,
    \[ \Pr[\mathcal{W}\mid\tau_*]=\Pr[a\leq|\bT|\leq3a]\leq p_*<\frac23. \]
    It follows that $\Pr[\mathcal{W}]\leq p_*+1/N$.
    
    Finally, we relate the window condition to the Pelagic condition. Under our hard input distribution, $\bF$ and $\bG$ are independent uniform truth tables, and hence $\bF\odot \bG$ is itself a uniform truth table. We may therefore apply \Cref{lem:magnitude-to-fourier-mass} in the special case with $s=N$ and $\bW=\bU=\bF\odot\bG$. So, except with probability $o(1)$ over the inputs, every valid $\PFF_n$ output satisfies the window condition. Hence
    \[ \Pr[\text{the protocol succeeds}]\leq\Pr[\mathcal{W}]+o(1)\leq p_*+\frac1N+o(1)<\frac23 \]
    for all sufficiently large $n$. This proves the randomized two-way communication lower bound.
\end{proof}

\section{Affine Pelagic Fourier Fishing} \label{sec:APFF}
Our second problem, which we call \emph{Affine Pelagic Fourier Fishing}, builds on the idea of our first problem by adding structure to Bob's input that makes the problem efficiently verifiable while still preserving classical hardness, because the structure is hidden to Alice. This motif is inspired by the Hidden Matching problem~\cite{BJK08-comm}, though the nature of the structure we introduce is quite different. \emph{Affine Pelagic Fourier Fishing} attains an optimal $n$ versus $\Theta(2^n)$ quantum--classical separation in the one-way communication setting. The key modification is that we partition the input space $\zo^n$ of Alice's function into cosets of a linear subspace, breaking the problem into many smaller problem instances on affine subspaces, where we ask for a correct answer to any one of them. 

A key ingredient in our analysis is a line of work on the expected maximum load of linear hash functions initiated by Alon et al.~\cite{alon1999linear} and recently sharpened and generalized by Jaber, Kumar, and Zuckerman~\cite{jaber25}. We use the latter result, stated here as \Cref{thm:jkz-linear-hashing}, by interpreting it as a bound on how heavily an arbitrary fixed set can intersect the cosets of a random subspace in expectation.

We establish the definitions of the Fourier-analytic notation we use in Section~\ref{sub:fourier-analysis}.
\begin{tcolorbox}[problembox]
\begin{problem}[Affine Pelagic Fourier Fishing, $\APFF_n$]
\label{problem:AFM}
Let $N = 2^n$, and for $n \geq 4$, set $\ell = \lceil 2\log_2 n\rceil$ and $K = 2^\ell$, so that $\ell \leq n$ and $n^2 \leq K < 2n^2$. \\

Alice receives a function $f: \zo^n \to \pmone$ given as a truth table of length $N$. Bob receives an $\ell$-dimensional subspace $V \leq \zo^n$, specified by a publicly fixed canonical basis $b_1(V), \dots, b_\ell(V)$. Bob must output a string $a \in \zo^n$, specifying a coset $U = V + a$, and a string $y \in \zo^\ell$, identified with the linear map in $V^*$ sending $b_i \mapsto y_i$. The output is correct if
\begin{equation}
    \sum_{z \in V^*:\, |\widehat{f}_U(z)| \leq |\widehat{f}_U(y)|} |\widehat{f}_U(z)|^2 \geq \frac{1}{8},
    \qquad
    \sum_{z \in V^*:\, |\widehat{f}_U(z)| \geq |\widehat{f}_U(y)|} |\widehat{f}_U(z)|^2 \geq \frac{1}{8}.
\end{equation}
\end{problem}
\end{tcolorbox}

We pick $\ell = \lceil 2\log_2n\rceil$ such that the coset size is $K = 2^\ell = \Theta(n^2)$ as a convenient scale that simultaneously allows for efficient verification, while also being large enough for the proof of the lower bound to go through.

\begin{theorem}[Restatement of \Cref{thm:intro-one-way}]
\label{thm:APFF-combined}
For all $n \geq 4$, $\APFF_n$ admits an $n$-qubit quantum one-way communication protocol that succeeds with probability at least $3/4$ on every input. Meanwhile, for all sufficiently large $n$, every randomized public-coin one-way communication protocol for $\APFF_n$ with success probability at least $2/3$ on average when Alice's truth table and Bob's subspace are sampled independently and uniformly at random communicates $\Omega(2^n)$ bits. In particular,
\[ \Qoneway(\APFF_n) \leq n \qquad \textrm{and} \qquad \Roneway(\APFF_n) = \Theta(2^n).\]
Furthermore, answers to $\APFF_n$ are efficiently verifiable. Specifically, there exists a deterministic algorithm that, given a basis of $V$, a candidate answer $(a,y)$, and query access to $f$, decides whether $(a,y)$ is a valid answer using $O(n^2)$ queries to $f$ and $\poly(n)$ time.
\end{theorem}
Essentially, the verification algorithm for $\APFF_n$ queries $f$ on all of $U = V+a$, using $K < 2n^2$ queries. It then computes all $K$ coefficients $\wh{f}_U(z)$ for $z\in V^*$ and checks whether the two inequalities in the definition of $\APFF_n$ hold. All the Fourier coefficients on $V$ can be computed from the queried values of $f$ by applying the Walsh--Hadamard transform in $\tilde{O}(n^2)$ time.

\subsection{Quantum upper bound}

We now show that $\APFF_n$ admits a quantum one-way communication protocol using $n = \log_2 N$ qubits that succeeds with probability at least $3/4$ on all inputs. Additionally, we observe that Bob is efficient in this protocol, requiring only $O(n^2/\log n)$ Clifford gates. Alice prepares and sends the $n$-qubit binary phase state
\[ \ket{\psi_f} = \frac{1}{\sqrt{N}}\sum_{x \in \zo^n} f(x)\ket{x}.\]
Bob is given a basis $b_1, \dots, b_\ell$ of $V$, which he extends to a basis $b_1,\dots, b_n$ of $\zo^n$. Let $M\colon \zo^\ell \times \zo^{n-\ell} \to \zo^n$ be the invertible linear map
\[
    M(c, d) = \sum_{i=1}^{\ell} c_i b_i + \sum_{j=1}^{n-\ell} d_j b_{\ell+j}.
\]
Bob applies the permutation unitary $\Pi_{M^{-1}}$, which maps $\ket{x} \mapsto \ket{M^{-1}(x)}$, to $\ket{\psi_f}$, obtaining
\[
    \frac{1}{\sqrt{N}} \sum_{d \in \zo^{n-\ell}} \sum_{c \in \zo^\ell} f(M(c,d))\ket{c}\ket{d}.
\]
Note that $\Pi_{M^{-1}}$ can be implemented with $O(n^2/\log n)$ CNOT gates~\cite{patel08-synthesis}. Bob then measures the second register, which holds $n-\ell$ qubits. Each outcome $d\in \zo^{n-\ell}$ occurs with probability $K/N = 1/2^{n-\ell}$, and it identifies the coset $U = V + a$ with $a = M(0, d)$. The post-measurement state of the first register is
\[
    \frac{1}{\sqrt{K}} \sum_{c \in \zo^\ell} f\Big(a + \sum_{i=1}^{\ell} c_i b_i\Big)\ket{c}.
\]
Finally, Bob applies $H^{\otimes \ell}$ to the first register. Writing each $y \in V^*$ as the string $(y(b_1), \dots, y(b_\ell)) \in \zo^\ell$, we have $\chi_y\big(\sum_i c_i b_i\big) = (-1)^{c \cdot y}$, so the resulting state is
\[
    \sum_{y \in V^*} \Big(\frac{1}{K} \sum_{c \in \zo^\ell}
    f\Big(a + \sum_{i=1}^{\ell} c_i b_i\Big)(-1)^{c \cdot y}\Big)\ket{y}
    = \sum_{y \in V^*} \Big(\frac{1}{K} \sum_{v \in V} f(v + a)\chi_y(v)\Big)\ket{y}
    = \sum_{y \in V^*} \widehat{f}_U(y)\ket{y}.
\]
Measuring in the computational basis gives $y$ with probability $|\widehat{f}_U(y)|^2$, and Bob outputs $a$ and $y$. Since $f$ takes values in $\pmone$, Parseval's identity gives that, for every coset $U$ of $V$,
\[
    \sum_{y \in V^*} |\widehat{f}_U(y)|^2 = 1.
\]
By the same argument as for $\PFF_n$, the measured $y$ is a valid answer for $U$ with
probability at least $3/4$, for whichever coset $U$ was measured.

\subsection{Randomized one-way lower bound}
\label{subsec:apff-lb}

In this section we prove the randomized one-way communication lower bound in \Cref{thm:APFF-combined}. We first briefly restate some information from~\Cref{sub:fourier-analysis} and define some new notation that we will use.\\

Let $N=2^n$, $\ell=\lceil2\log_2 n\rceil$, and $K=2^\ell=\Theta(n^2)$. This choice ensures that $(n-\ell)/K = o(1)$ and $K^2/N = o(1)$, which will be necessary in the proof. We identify $\zo^n$ with $\F_2^n$ and use addition modulo $2$ throughout. Coordinates of length-$N$ vectors will sometimes be indexed by $\F_2^n$, using a fixed identification with $[N]$.

An $\ell$-dimensional subspace $V\leq\F_2^n$ has $N/K$ cosets, each of size $K$. Write $V^*$ for its dual space and $\chi_y(v)=(-1)^{y(v)}$ for $y\in V^*$. For a Boolean function $f:\F_2^n\to\pmone$ and a coset $U=a+V$ with a chosen representative $a$, recall that the Fourier coefficient of $f$ on $U$ at $y$ is
\[\widehat f_U(y):=\frac1K\sum_{v\in V}f(a+v)\chi_y(v).\]
Using the same representative, define the \emph{coset character vector} $\varphi_{U,y}\in\R^N$ by
\begin{equation}\label{eq:coset-character-vec}(\varphi_{U,y})_z:=\begin{cases}\displaystyle \frac1{\sqrt K}\chi_y(z-a), & z\in U,\\[6pt]0, & z\notin U.\end{cases}\end{equation}
In particular, $|(\varphi_{U,y})_z|=K^{-1/2}\mathbbm{1}[z\in U]$ for every $z\in\F_2^n$. If $F\in \pmone^N$ is the truth table of a function $f\colon \zo^n \to \pmone$, then
\[\|\varphi_{U,y}\|_2=1\qquad\textrm{and}\qquad \langle F,\varphi_{U,y}\rangle=\sqrt K\,\widehat f_U(y).\]

Fix an integer $c$. For each possible message $m\in\zo^c$, suppose Bob's output after receiving a subspace $V$ as input is specified by a coset representative $a_m(V)\in\F_2^n$ and a character $y_m(V)\in V^*$. Write
\[U_m(V):=a_m(V)+V,\qquad \varphi_m(V):=\varphi_{U_m(V),y_m(V)}.\]
Thus, for each fixed message $m$, the functions $a_m(\cdot)$ and $y_m(\cdot)$ describe Bob's output as a function of his input subspace $V$. Throughout this section, $\bV$ denotes a uniformly random $\ell$-dimensional subspace of $\F_2^n$, and for each $m$ and $z,z'\in\F_2^n$ we write
\begin{equation} \label{eq:apff-gmzz}
g_m(z,z'):=\E_{\bV}[\varphi_m(\bV)_z\varphi_m(\bV)_{z'}]\end{equation}
for the correlation between coordinates $z$ and $z'$ of the selected coset character vector $\varphi_m(V)$. The lemmas below hold for any collection of output functions $a_m(\cdot),y_m(\cdot)$.

\subsubsection{Statement of \cref{lem:coset-character-correlations,lem:selected-coset-disagreement,lem:apff-decomposition}}

First, we state a lemma which bounds the sum of all $g_m(z,z')^2$ for distinct $z\neq z'$.

\begin{lemma}\label{lem:coset-character-correlations}
    For every $m\in\zo^c$ and all sufficiently large $n$,
    \[\sum_{z\neq z'}g_m(z,z')^2\leq\frac3N.\]
\end{lemma}

We defer the proof of this lemma until the end of this section. Next, we prove a consequence of \Cref{cor:jkz-subspace-intersection} which shows that two globally close truth tables remain close, on average, even on a coset selected as a function of Bob's input. This step is where the results of~\cite{alon1999linear, jaber25} on the expected maximum load of linear hash functions get used for our lower bound.

\begin{lemma}\label{lem:selected-coset-disagreement}
Let $(\bX,\bY,\bM)$ be jointly distributed random variables with $\bX,\bY\in\pmone^N$ and $\bM\in\zo^c$, and suppose that $(\bX,\bY,\bM)$ is independent of $\bV$, a uniformly random $\ell$-dimensional subspace. Then
\[\E\left[\frac1K\sum_{z\in U_{\bM}(\bV)}(\bX_z-\bY_z)^2\right]\leq32\frac{\E\|\bX-\bY\|_2^2}{N}+128\frac{n-\ell}{K}.\]
In particular, if $\E\|\bX-\bY\|_2^2=o(N)$, then
\[\E\left[\frac1K\sum_{z\in U_{\bM}(\bV)}(\bX_z-\bY_z)^2\right]=o(1).\]
\end{lemma}
\begin{proof}
Condition on $\bX,\bY,\bM$, and write $\bM=m$. Then the \emph{disagreement set} of $\bX$ and $\bY$,
\[
D:=\{z\in\F_2^n:\bX_z\neq\bY_z\},
\]
is fixed, while $\bV$ remains a uniformly random $\ell$-dimensional subspace. Since $U_m(V)$ is one of the cosets of $V$ for every $V$, \Cref{cor:jkz-subspace-intersection} gives
\[
\E_{\bV}\left[\frac{|D\cap U_m(\bV)|}{K}\right]\leq\E_{\bV}\left[\max_a\frac{|D\cap(a+\bV)|}{K}\right]\leq32\left(\frac{|D|}{N}+\frac{n-\ell}{K}\right).
\]
Since $(\bX_z-\bY_z)^2=4\mathbbm{1}[z\in D]$, for every realization of $\bV$,
\[
\frac1K\sum_{z\in U_m(\bV)}(\bX_z-\bY_z)^2=\frac{4|D\cap U_m(\bV)|}{K}.
\]
Averaging the preceding bound also over $(\bX,\bY,\bM)$ therefore yields
\[
\E\left[\frac1K\sum_{z\in U_{\bM}(\bV)}(\bX_z-\bY_z)^2\right]\leq128\left(\frac{\E|D|}{N}+\frac{n-\ell}{K}\right)=32\frac{\E\|\bX-\bY\|_2^2}{N}+128\frac{n-\ell}{K}.
\]
The final claim follows from $K=\Theta(n^2)$, so $(n-\ell)/K=o(1)$.
\end{proof}

Next, we state the key technical lemma of this section. This lemma is analogous to \Cref{lem:inner-prod-decomp} for Pelagic Fourier Fishing, but requires much more work to prove due to the affine setting.

\begin{lemma}[Coset character inner-product decomposition]\label{lem:apff-decomposition}
    Let $\bF\sim\mathrm{Unif}(\pmone^N)$ and let $\bM$ be a $\zo^c$-valued random variable jointly distributed with $\bF$, where $c=o(N)$. Suppose that $(\bF,\bM)$ is independent of $\bV$. Then
    \[\langle\bF,\varphi_{\bM}(\bV)\rangle=\bA+\bS+\berr,\]
    where $\bA$ and $\bS$ are independent, $\bS$ is distributed as $\frac1{\sqrt K}\sum_{i=1}^K\bxi_i$
    for independent uniform signs $\bxi_i\in\pmone$, and $\E[\berr^2]=o(1)$.
\end{lemma}

\subsubsection{Proof of \cref{thm:APFF-combined}}

We defer the proof of the key technical lemma stated above and now prove the randomized one-way communication lower bound, assuming \Cref{lem:apff-decomposition} for now.

\begin{proof}[Proof of the randomized lower bound in \Cref{thm:APFF-combined}]
By Yao's principle~\cite{yao1983lower}, it suffices to show that all deterministic protocols using $c=o(N)$ bits of communication succeed with probability less than $2/3$ under a hard input distribution, for all sufficiently large $n$. 

For every $\ell$-dimensional subspace $V\leq\F_2^n$, use the publicly fixed canonical basis $b_1(V),\dots,b_\ell(V) \in \F_2^n$. Consider the input distribution where Alice's truth table $\bF$ is uniformly random on $\pmone^N$ and Bob independently receives a uniformly random $\ell$-dimensional subspace $\bV\leq\F_2^n$, presented by the basis $b_1(\bV),\dots,b_\ell(\bV)$. Now fix an arbitrary deterministic $c$-bit one-way communication protocol with $c=o(N)$. We will show that on this distribution it succeeds with probability less than $2/3$, which implies the lower bound.

Let $\msg:\pmone^N\to\zo^c$ denote Alice's message function, and write $\bM:=\msg(\bF)$. Upon receiving message $m$ and subspace $V$, Bob outputs a coset representative $a_m(V)\in\F_2^n$ and a string $r_m(V)\in\zo^\ell$. Using the fixed ordered basis $b_1(V),\ldots,b_\ell(V)$ of $V$, we identify $r_m(V) \in \zo^\ell$ with the character $y_m(V)\in V^*$ defined by
\[y_m(V)(b_i(V))=(r_m(V))_i\qquad\textrm{for each }i\in[\ell].\]
We then use the notation $U_m(V)$ and $\varphi_m(V)$ introduced at the beginning of the section for the corresponding selected coset and coset character vector.

Since $(\bF,\bM)$ is independent of $\bV$, \Cref{lem:apff-decomposition} gives that
\[ \langle\bF,\varphi_{\bM}(\bV)\rangle=\bA+\bS+\berr,\]
where $\bA$ and $\bS$ are independent, $\bS$ is distributed as a normalized sum of $K$ independent uniform signs, and $\E[\berr^2]=o(1)$. \Cref{lem:independent-noise}, applied with $s=K$, yields
\begin{equation}\label{eq:apff-window-bound}
    \Pr[a\leq|\langle\bF,\varphi_{\bM}(\bV)\rangle|\leq3a]\leq p_*
\end{equation}
for all sufficiently large $n$, where $a=\sqrt{2/3}$ and $p_*<2/3$ are the constants from Fact~\ref{fact:gaussian}.

\medskip
Next, we relate this window condition to the Pelagic condition. Since $\bM$ is a deterministic function of $\bF$ and $H(\bM)\leq c$,
\[H(\bF\mid\bM)=N-H(\bM)\geq N-c.\]
Now construct a separate coupling by applying \Cref{lem:coupling-by-entropy} to $\bF\mid(\bM=m)$ for each $m$ in the support of $\bM$. This preserves the law of $(\bF,\bM)$ and gives a random variable $\bG$, with $(\bF,\bM,\bG)$ independent of $\bV$, such that
\[\bG\mid(\bM=m)\sim\mathrm{Unif}(\pmone^N)\]
and
\[
\E\bigl[\|\bF-\bG\|_2^2\bigm|\bM=m\bigr]\leq4\sqrt{N\bigl(N-H(\bF\mid\bM=m)\bigr)}\]
for every $m$ in the support of $\bM$. Averaging this over all $m$ and applying Jensen's inequality, we get
\begin{equation}\label{eq:apff-main-FG} \E\|\bF-\bG\|_2^2\leq4\sqrt{N\bigl(N-H(\bF\mid\bM)\bigr)}\leq4\sqrt{Nc}=o(N).\end{equation}

Since $(\bF,\bM,\bG)$ is independent of $\bV$, \Cref{lem:selected-coset-disagreement} together with~\eqref{eq:apff-main-FG} gives
\begin{equation}\label{eq:apff-main-jkz-bound}
\E\left[\frac1K\sum_{z\in U_{\bM}(\bV)}(\bF_z-\bG_z)^2\right]\leq32\frac{\E\|\bF-\bG\|_2^2}{N}+128\frac{n-\ell}{K}=o(1).
\end{equation}
\medskip

For each message $m$ and subspace $V$, define the affine bijection
\[\psi_{m,V}:\F_2^\ell\to U_m(V),\qquad \psi_{m,V}(t):=a_m(V)+\sum_{i=1}^{\ell}t_i b_i(V).\]
Thus $\psi_{m,V}$ simply uses the fixed basis of $V$ to label the $K$ points of the coset $U_m(V)$ by strings $t\in\F_2^\ell$.

Using the random pair $(\bM,\bV)$, let $\bF',\bG'\in\pmone^K$ denote the restrictions of $\bF,\bG$ to Bob's selected coset under this labeling:
\[\bF'_t:=\bF_{\psi_{\bM,\bV}(t)},\qquad \bG'_t:=\bG_{\psi_{\bM,\bV}(t)}\qquad(t\in\F_2^\ell).\]
Conditional on $(\bM,\bV)$, the map $\psi_{\bM,\bV}$ is fixed, while $\bG$ is uniform on $\pmone^N$. Hence $\bG'\sim\mathrm{Unif}(\pmone^K)$. Moreover,~\eqref{eq:apff-main-jkz-bound} gives
\[\frac1K\E\|\bF'-\bG'\|_2^2=o(1)\]
since $\bF_z - \bG_z$ for $z\in U_m(V)$ are exactly the terms summed over in $\| \bF' - \bG'\|_2^2$ by the bijection.

Let $f':\F_2^\ell\to\pmone$ be the Boolean function with truth table $\bF'$. For any $r\in\F_2^\ell$, let $y_r\in\bV^*$ be the character determined by
\[y_r(b_i(\bV))=r_i\qquad\textrm{for each }i\in[\ell].\]
Then, by the definition of $\psi_{\bM,\bV}$,
\[
\widehat f'(r)
=\frac1K\sum_{t\in\F_2^\ell}\bF'_t(-1)^{t\cdot r}
=\widehat f_{U_{\bM}(\bV)}(y_r).
\]
In particular, for Bob's actual output string $r_{\bM}(\bV)$,
\[\widehat f'\bigl(r_{\bM}(\bV)\bigr)=\widehat f_{U_{\bM}(\bV)}\bigl(y_{\bM}(\bV)\bigr),\]
and therefore
\[\left|\langle\bF,\varphi_{\bM}(\bV)\rangle\right|
=\sqrt K\,\left|\widehat f'\bigl(r_{\bM}(\bV)\bigr)\right|.\]

We may now apply \Cref{lem:magnitude-to-fourier-mass} to the coupling $(\bF',\bG')$ with $s=K$. Except with probability $o(1)$, every string $r\in\F_2^\ell$ satisfying the Pelagic condition on the selected coset obeys
\[\frac{a}{\sqrt K}\leq|\widehat f'(r)|\leq\frac{3a}{\sqrt K}.\]
Hence, except with probability $o(1)$, if Bob's output is valid then
\[a\leq\left|\langle\bF,\varphi_{\bM}(\bV)\rangle\right|\leq3a.\]
Combining this with the window bound in \eqref{eq:apff-window-bound} proved above,
\[
\Pr[\text{the protocol succeeds}]
\leq
\Pr\!\left[a\leq\left|\langle\bF,\varphi_{\bM}(\bV)\rangle\right|\leq3a\right]+o(1)
\leq p_*+o(1)<\frac23
\]
for all sufficiently large $n$. This proves the claimed randomized one-way communication lower bound.
\end{proof}

\subsubsection{Proof of the key technical lemma}

We now prove the key technical step, \Cref{lem:apff-decomposition}, assuming \Cref{lem:coset-character-correlations}, whose proof is deferred until the end of this subsection. Recall that \cref{lem:coset-character-correlations} shows that the inner product $\langle \bF, \varphi_{\bM}(\bV)\rangle$ admits a decomposition into $\bA + \bS + \berr$ where $\bS$ is approximately Gaussian, $\bA$ and $\bS$ are independent, and $\E[\berr^2] = o(1)$.

\begin{proof}[Proof of \Cref{lem:apff-decomposition}]
Set an integer parameter
\[k:=\Bigl\lceil\sqrt{N(c+1)}\,\Bigr\rceil.\]
This choice of $k$ is asymptotically at an intermediate scale between $c$ and $N$: we have that $c=o(k)$ and $k=o(N)$.

We first use \Cref{lem:entropy-budget-conditioning} to reveal a fixed set of at most $k$ coordinates of $\bF$ in order to remove most of its pairwise correlations. Since $\bM$ takes at most $2^c$ values, $H(\bF\mid\bM)\geq H(\bF)-H(\bM)\geq N-c$. Applying \Cref{lem:entropy-budget-conditioning} to $\bF$ with $\bM$ jointly distributed with it gives a fixed set $I\subseteq\F_2^n$ with $|I|\leq k$ such that, writing $\bZ:=(\bM,\bF_I)$ where $\bF_I = (\bF_i)_{i\in I}$,
\begin{equation}\label{eq:apff-cov}\E\sum_{z\neq z'}\Cov(\bF_z,\bF_{z'}\mid\bZ)^2\leq\frac{2Nc}{k}\qquad\textrm{and}\qquad H(\bF\mid\bZ)\geq N-c-k.\end{equation}
Note that $\bM$ is determined by $\bZ$. We will eventually use these properties to bound $\E[\berr^2]$.

\medskip
Next, for every fixed $Z$ in the support of $\bZ$, we apply the coupling construction of \Cref{lem:coupling-by-entropy} to $\bF\mid(\bZ=Z)$ to obtain a coupling between $\bF\mid(\bZ=Z)$ and $\bG\sim\mathrm{Unif}(\pmone^N)$ such that
\[\E\bigl[\|\bF-\bG\|_2^2\bigm|\bZ=Z\bigr]\leq4\sqrt{N\bigl(N-H(\bF\mid\bZ=Z)\bigr)}.\]
We construct this coupling using only $(\bF,\bZ)$ and auxiliary randomness, so $(\bF,\bZ,\bG)$ remains independent of $\bV$. Additionally, we have that for every $Z$ in the support of $\bZ$, 
\begin{equation}\label{eq:apff-uniformity}\bG\mid(\bZ=Z)\sim\mathrm{Unif}(\pmone^N).\end{equation}
Averaging the conditional bound over $\bZ$, applying Jensen's inequality, and using $H(\bF\mid\bZ)=\E_{\bZ}\,H(\bF\mid\bZ=\bZ)$, we obtain that
\[\E\|\bF-\bG\|_2^2\leq4\sqrt{N\bigl(N-H(\bF\mid\bZ)\bigr)}.\]
By~\eqref{eq:apff-cov}, we then have that $N-H(\bF\mid\bZ)\leq c+k=o(N)$, hence
\begin{equation}\label{eq:apff-Delta}\Delta:=\E\|\bF-\bG\|_2^2=o(N).\end{equation}
Note that we have obtained $\bF=\bG+(\bF-\bG)$, where $\bG$ is perfectly uniform even conditional on $\bZ$, and $\bF-\bG$ has small expected squared norm.

Since $(\bF,\bZ,\bG)$, and hence $(\bF,\bG,\bM)$, is independent of $\bV$, \Cref{lem:selected-coset-disagreement} together with~\eqref{eq:apff-Delta} gives
\begin{equation}\label{eq:apff-eta}
\eta:=\E\left[\frac1K\sum_{z\in U_{\bM}(\bV)}(\bF_z-\bG_z)^2\right]\leq32\frac{\Delta}{N}+128\frac{n-\ell}{K}=o(1).
\end{equation}

Now we consider the random variable $\bF - \bG$. Recall from \eqref{eq:apff-Delta} that $\Delta =\E\| \bF - \bG\|_2^2 = o(N)$. Define
\[\mu(\bZ):=\E[\bF\mid\bZ]\qquad\textrm{and}\qquad\be:=(\bF-\bG) -\mu(\bZ).\]
By the fact that $\bG \mid (\bZ = Z)$ is a uniform random variable \eqref{eq:apff-uniformity}, $\E[\bG\mid\bZ] = 0$. So $\mu(\bZ) = \E[\bF \mid \bZ] = \E[\bF-\bG \mid \bZ]$. Therefore, $\be = (\bF - \bG) - \E[\bF-\bG \mid \bZ]$. 

\medskip
Now let us define $\bA$, $\bS$, and $\berr$, the random variables we get from the decomposition:
\begin{equation}\label{eq:apff-decomp-def}\bA:=\langle\mu(\bZ),\varphi_{\bM}(\bV)\rangle,\qquad\bS:=\langle\bG,\varphi_{\bM}(\bV)\rangle,\qquad\berr:=\langle\be,\varphi_{\bM}(\bV)\rangle.\end{equation}
One can verify that
\[\bA+\bS+\berr=\langle\bF,\varphi_{\bM}(\bV)\rangle\]
using that $\bF=\mu(\bZ)+\bG+\be$. Observe that if we condition on $\bZ=(\bM,\bF_I)$ and $\bV$, then $\bA$ is completely determined, where $\varphi_{\bM}(\bV)$ is a vector with exactly $K$ nonzero entries, each equal to $\pm1/\sqrt K$. Meanwhile, $\bG$ is distributed uniformly over $\pmone^N$ due to~\eqref{eq:apff-uniformity}. Therefore $\sqrt K\,\bS$ is a sum of $K$ independent uniform signs, and $\bS$ is therefore distributed as $K^{-1/2}\sum_{i=1}^K\bxi_i$ for independent uniform signs $\bxi_i\in\pmone$. Since this conditional law does not depend on the values of $\bZ$ and $\bV$, $\bS$ is independent of $(\bZ,\bV)$, and hence independent of $\bA$ as well. Thus we have shown two of the conditions of our lemma.
\medskip

Finally, we must show that $\E[\berr^2]=o(1)$. Expanding the square in the definition of $\berr$ gives
\[\berr^2= \langle \be, \varphi_{\bM}(\bV)\rangle^2 =\sum_{z,z'}\varphi_{\bM}(\bV)_z\,\varphi_{\bM}(\bV)_{z'}\,\be_z\be_{z'},\]
where the sum ranges over all ordered pairs $(z,z')\in\F_2^n\times\F_2^n$. We compute the conditional expectation of each term given $\bZ$. Conditional on $\bZ$, the message $\bM=m$ is fixed, so the factor $\varphi_m(\bV)_z\varphi_m(\bV)_{z'}$ depends only on $\bV$. $\be=\bF-\bG-\mu(\bZ)$ is a function of $(\bF,\bZ,\bG)$, so the factor $\be_z\be_{z'}$ depends only on $(\bF,\bG)$ after conditioning on $\bZ$. 
Conditional on $\bZ$, the subspace $\bV$ remains uniform and is independent of $(\bF,\bG)$. The conditional expectation therefore factors:
\[\E\bigl[\varphi_{\bM}(\bV)_z\,\varphi_{\bM}(\bV)_{z'}\,\be_z\be_{z'}\bigm|\bZ\bigr]=\E\bigl[\varphi_{\bM}(\bV)_z\,\varphi_{\bM}(\bV)_{z'}\bigm|\bZ\bigr]\,\E\bigl[\be_z\be_{z'}\bigm|\bZ\bigr]=g_{\bM}(z,z')\,\E\bigl[\be_z\be_{z'}\bigm|\bZ\bigr].\]
Summing over $(z,z') \in \F_2^n \times \F_2^n$ and averaging over $\bZ$ gives
\begin{equation}\label{eq:apff-err-expansion}\E[\berr^2]=\E_{\bZ}\sum_{z,z' \in \F_2^n}g_{\bM}(z,z')\,\E[\be_z\be_{z'}\mid\bZ].\end{equation}
We now bound the diagonal and off-diagonal terms of the double summation in~\eqref{eq:apff-err-expansion} separately.
\medskip

\textbf{Diagonals.} For the diagonal terms ($z = z'$), note that $\be_z=(\bF_z-\bG_z)-\E[\bF_z-\bG_z\mid\bZ]$ is the conditionally centered version of $\bF_z-\bG_z$. So
\begin{equation}\label{eq:apff-e_z-variance}
\E[\be_z^2\mid\bZ]=\Var(\bF_z-\bG_z\mid\bZ)\leq\E[(\bF_z-\bG_z)^2\mid\bZ].\end{equation}
For a fixed $m$, since $|\varphi_m(V)_z|=K^{-1/2}\mathbbm{1}[z\in U_m(V)]$, the diagonal's weights in \eqref{eq:apff-err-expansion} are
\begin{equation} \label{eq:apff-gmzz-diag}
    g_m(z,z)=\E_{\bV}\bigl[\varphi_m(\bV)_z^2\bigr]=\frac1K\Pr_{\bV}[z\in U_m(\bV)]\geq0.
\end{equation}
Because these weights are nonnegative, we may apply~\eqref{eq:apff-e_z-variance} term-by-term to the diagonal terms in~\eqref{eq:apff-err-expansion}. Thus the diagonal contribution is at most
\begin{align} \notag
&\E_{\bZ}\sum_{z\in\F_2^n}g_{\bM}(z,z)\,\E[(\bF_z-\bG_z)^2\mid\bZ]\\ \notag
&\qquad=\E_{\bZ,\bV}\left[\frac1K\sum_{z\in\F_2^n}\mathbbm{1}[z\in U_{\bM}(\bV)]\,\E[(\bF_z-\bG_z)^2\mid\bZ]\right]\\
&\qquad=\E\left[\frac1K\sum_{z\in U_{\bM}(\bV)}(\bF_z-\bG_z)^2\right]
=\eta=o(1). \label{eq:apff-diag-bound}
\end{align}
The first equality uses \eqref{eq:apff-gmzz-diag}. For the second, conditioned on $\bZ$, the random subspace $\bV$ is independent of $(\bF,\bG)$. Thus the conditional expected squared error at each coordinate can be averaged over the independently selected coset, giving exactly the expected squared error on that coset. The final equality is~\eqref{eq:apff-eta}.
\medskip

\textbf{Off-Diagonals.}
For $z\neq z'$, write $C_{z,z'}:=\Cov(\bF_z,\bF_{z'}\mid\bZ)$ and $B_{z,z'}:=\E[\bF_z\bG_{z'}\mid\bZ]$. Recall that $\mu(\bZ) = \E[\bF \mid \bZ]$. Since $\bG\mid\bZ$ is uniform (and hence $\E[\bG \mid \bZ] = 0$), expanding $\be_z=(\bF_z-\mu_z(\bZ))-\bG_z$ gives
\[\E[\be_z\be_{z'}\mid\bZ]=C_{z,z'}-B_{z,z'}-B_{z',z}.\]
By \Cref{lem:coset-character-correlations} and Cauchy--Schwarz, the $C$ contribution is at most
\[\left|\E_{\bZ}\sum_{z\neq z'}g_{\bM}(z,z')C_{z,z'}\right|\leq\sqrt{\frac3N}\left(\E_{\bZ}\sum_{z\neq z'}C_{z,z'}^2\right)^{1/2}\leq\sqrt{\frac{6c}{k}}=o(1),\]
where the second inequality uses~\eqref{eq:apff-cov}.

For the $B$ contributions, recall that for $z\neq z'$, $\E[\bG_z\bG_{z'}\mid\bZ]=0$, and hence
\[B_{z,z'}=\E[(\bF_z-\bG_z)\bG_{z'}\mid\bZ].\]
Fix $z$ and define
\[H_z:=\sum_{z':\,z'\neq z}B_{z,z'}\bG_{z'}.\]
Conditioned on $\bZ$, the coefficients $B_{z,z'}$ are fixed and the coordinates of $\bG$ are independent uniform signs. Therefore, $\E[H_z^2\mid\bZ]=\sum_{z':\,z'\neq z}B_{z,z'}^2$.
Moreover, by the definition of $B_{z,z'}$, $\E[(\bF_z-\bG_z)H_z\mid\bZ]=\sum_{z':\,z'\neq z}B_{z,z'}^2$.
Therefore, by conditional Cauchy--Schwarz,
\[\left(\sum_{z':\,z'\neq z}B_{z,z'}^2\right)^2=\E[(\bF_z-\bG_z)H_z\mid\bZ]^2\leq\E[(\bF_z-\bG_z)^2\mid\bZ]\,\E[H_z^2\mid\bZ]=\E[(\bF_z-\bG_z)^2\mid\bZ]\left(\sum_{z':\,z'\neq z}B_{z,z'}^2\right),\]
and hence
\[\sum_{z':\,z'\neq z}B_{z,z'}^2\leq\E[(\bF_z-\bG_z)^2\mid\bZ].\]
Summing over $z$ and averaging over $\bZ$,
\[\E_{\bZ}\sum_{\substack{z,z'\\z'\neq z}}B_{z,z'}^2\leq\E\|\bF-\bG\|_2^2=\Delta=o(N),\]
where the last equality and bound are from~\eqref{eq:apff-Delta}. 
By \Cref{lem:coset-character-correlations} and Cauchy--Schwarz, the $B$ contribution is at most
\[\left|\E_{\bZ}\sum_{z\neq z'}g_{\bM}(z,z')B_{z,z'}\right|\leq\sqrt{\frac3N}\left(\E_{\bZ}\sum_{z\neq z'}B_{z,z'}^2\right)^{1/2}\leq\sqrt{\frac{3\Delta}{N}}=o(1),\]
where the second inequality uses the preceding bound.
The same bound holds for the transposed term $B_{z',z}$, since summing over all ordered pairs $z\neq z'$ is unchanged by swapping $z$ and $z'$.

Combining these two mixed terms with the covariance contribution bounded above by $\sqrt{6c/k}$ and the diagonal contribution bounded by $\eta$, we obtain
\[\E[\berr^2]\leq\eta+\sqrt{\frac{6c}{k}}+2\sqrt{\frac{3\Delta}{N}}=o(1).\]
This completes the proof.
\end{proof}

\begin{proof}[Proof of \Cref{lem:coset-character-correlations}]
Fix $m\in\zo^c$, and abbreviate $U(V):=U_m(V)$ and $\varphi(V):=\varphi_m(V)$. Let $\bV'$ be an independent copy of $\bV$. For $z\neq z'$, independence gives
\[g_m(z,z')^2=\E_{\bV,\bV'}[\varphi(\bV)_z\varphi(\bV)_{z'}\varphi(\bV')_z\varphi(\bV')_{z'}].\]
Since $g_m(z,z')^2\geq0$, the triangle inequality gives
\[g_m(z,z')^2\leq\E_{\bV,\bV'}\bigl[\bigl|\varphi(\bV)_z\varphi(\bV)_{z'}\varphi(\bV')_z\varphi(\bV')_{z'}\bigr|\bigr].\]
Since $|\varphi(V)_z|=K^{-1/2}\mathbbm{1}[z\in U(V)]$ by~\eqref{eq:coset-character-vec}, the quantity inside the expectation is
\[\bigl|\varphi(\bV)_z\varphi(\bV)_{z'}\varphi(\bV')_z\varphi(\bV')_{z'}\bigr|=\frac{1}{K^2}\mathbbm{1}\bigl[z,z'\in U(\bV)\cap U(\bV')\bigr].\]
Therefore, summing over all ordered pairs $z\neq z'$,
\[\sum_{z\neq z'}g_m(z,z')^2\leq\frac{1}{K^2}\E_{\bV,\bV'}\left[\sum_{z\neq z'}\mathbbm{1}\bigl[z,z'\in U(\bV)\cap U(\bV')\bigr]\right]=\frac{1}{K^2}\E_{\bV,\bV'}\bigl[|U(\bV)\cap U(\bV')|\bigl(|U(\bV)\cap U(\bV')|-1\bigr)\bigr],\]
where the equality comes from counting the ordered pairs of distinct elements of $U(\bV)\cap U(\bV')$. The intersection of a coset of $\bV$ with a coset of $\bV'$ is either empty or a coset of $\bV\cap\bV'$, so $|U(\bV)\cap U(\bV')|\leq|\bV\cap\bV'|$. Since $t\mapsto t(t-1)$ is nondecreasing on the nonnegative integers,
\[\sum_{z\neq z'}g_m(z,z')^2\leq\frac{1}{K^2}\E_{\bV,\bV'}\bigl[|\bV\cap\bV'|\bigl(|\bV\cap\bV'|-1\bigr)\bigr].\]

We bound the expectation in this expression by conditioning on $\bV=V$, where $V\leq\F_2^n$ is any fixed $\ell$-dimensional subspace. Then $|V\cap\bV'|(|V\cap\bV'|-1)$ counts the ordered pairs $(z,z')$ of distinct points such that $z,z'\in V\cap\bV'$. Since $V$ is fixed, we may equivalently count ordered pairs of distinct points in $V$ and ask, for each pair, for the probability that both points also lie in $\bV'$.

First consider ordered pairs containing $0$. These are of the form $(0,z)$ or $(z,0)$ with $z\in V\setminus\{0\}$, so there are exactly $2(K-1)$ many. Since $0\in\bV'$ automatically, such a pair lies in $\bV'$ precisely when its nonzero point $z$ lies in $\bV'$. For every fixed nonzero $z$,
\[\Pr[z\in\bV']=\frac{K-1}{N-1}\leq\frac{K}{N}.\]

Now consider ordered pairs $(z,z')$ with $z,z'\in V\setminus\{0\}$ and $z\neq z'$. There are exactly $(K-1)(K-2)$ such pairs. Over $\F_2$, two distinct nonzero vectors are linearly independent. Hence, for any fixed such pair,
\[\Pr[z,z'\in\bV']=\frac{K-1}{N-1}\cdot\frac{K-2}{N-2}=\frac{(K-1)(K-2)}{(N-1)(N-2)}\leq\left(\frac{K}{N}\right)^2.\]
Therefore,
\[\E_{\bV'}\bigl[|V\cap\bV'|\bigl(|V\cap\bV'|-1\bigr)\bigr]\leq2(K-1)\frac{K}{N}+(K-1)(K-2)\left(\frac{K}{N}\right)^2\leq\frac{2K^2}{N}+\frac{K^4}{N^2}.\]
Finally,
\[\sum_{z\neq z'}g_m(z,z')^2\leq\frac{1}{K^2}\E_{\bV,\bV'}\bigl[|\bV\cap\bV'|\bigl(|\bV\cap\bV'|-1\bigr)\bigr]\leq\frac{2}{N}+\frac{K^2}{N^2}\leq\frac{3}{N},\]
where the last inequality holds for all sufficiently large $n$ since $K^2=o(N)$.
\end{proof}

\section{The Directional Derivative Problem}
\label{sec:DDP}

\cref{sec:PFF,sec:APFF} achieve the largest possible quantum--classical communication separation with unrestricted local computation. In this section, we ask how large of a separation can be achieved when the complexity of the quantum protocol is restricted. We answer this question through a family of total search problems called the \emph{Directional Derivative Problem}, denoted $\DDP_{n,d}$.

The first proofs we present are of the quantum and classical communication complexity of $\DDP_{n,d}$. Subsequently, we show that this result is optimal under two different regimes: (1) we introduce the Clifford hierarchy message model and show that $\DDP_{n,d}$ attains the largest possible separation for $n$-qubit level-$d$ messages. (2) we show that Alice’s phase state can be prepared using $O_d(n^d/\log n)$ one- and two-qubit gates and that this state-preparation complexity is asymptotically optimal. 

We begin by defining the communication problem that exhibits the optimal separation.
For a function $f:\F_2^n\to \F_2$ and a direction $a\in\F_2^n$, define the directional derivative of $f$ in direction $a$ by $D_a f(x) = f(x+a)+f(x)$. 
For integers $n,d\ge 1$, let $V_{n,d} \coloneqq \calP_{\le d}^0(\mathbb F_2^n)$ denote the set of zero-constant-term multilinear polynomials $q:\F_2^n\to\F_2$ of degree at most $d$. We may restrict to polynomials without a constant term without loss of generality, since adding a constant to $q$ won't change any derivative $D_aq(x)$. 

\begin{tcolorbox}[problembox]
\begin{problem}[Directional Derivative Problem, $\DDP_{n,d}$]
\label{problem:DDP}
Alice receives a polynomial $q\in V_{n,d}$, and Bob receives a nonzero direction $a\in \F_2^n\setminus\{0\}$. 
After receiving a message from Alice, Bob must output a pair $(x,b)\in \F_2^n\times \F_2$ such that $b = D_a q(x).$
\end{problem}
\end{tcolorbox}

Equivalently, $\DDP_{n,d}$ asks Bob to find any edge of the perfect matching $M_a = \bigl\{\{x,x+a\}: x\in\mathbb F_2^n\bigr\}$ together with the parity of $q$ on the two endpoints of this edge.
Thus, $\DDP_{n,d}$ is a restricted version of the Hidden Matching problem of Bar-Yossef, Jayram, and Kerenidis~\cite{BJK08-comm}. There are two restrictions. First, Alice's input is not arbitrary: it is constrained to be a degree-at-most-$d$ polynomial over $\F_2$. Second, Bob's matching is not arbitrary: it is one of the affine-shift matchings $M_a$, indexed by nonzero directions $a\in\F_2^n$. 

The main result of this section is the following:
\begin{theorem}[Restatement of \Cref{thm:intro-ddp}]\label{thm:ddp-separation} 
For every $d\ge 2$ and $n\ge d$, there is an exact quantum one-way communication protocol for $\DDP_{n,d}$ using an $n$-qubit level-$d$ message and a Clifford measurement. By contrast, under the uniform input distribution, every randomized classical one-way protocol for $\DDP_{n,d}$ with average success probability at least $2/3$ uses $\Omega_d(n^d)$ bits of communication. There is also an exact classical one-way protocol with $O_d(n^d)$ bits of communication. Moreover, Alice’s message can be prepared using $O^d(n^d/ \log n)$ one- and two-qubit gates, and this upper bound is asymptotically optimal up to a factor depending on $d$.
\end{theorem}

The rest of this section is devoted to proving \cref{thm:ddp-separation}. 
In \cref{lemma:ddp-quantum}, we prove the quantum upper bound. In \cref{thm:general-directional-lower-bound}, we prove the classical lower bound, and in \cref{cor:fixed-degree} we show that this is tight.  Finally, \cref{sec:ddp-gate-bound} gives the matching gate bounds. Together, this implies \cref{thm:ddp-separation}.

\subsection{Quantum upper bound}

We show that there is an exact one-way quantum protocol for $\DDP_{n,d}$ for all $n,d \ge 1$ using an $n$-qubit message. In particular, for $d=2$, Alice's message is a stabilizer state. The protocol follows from the Hidden Matching protocol~\cite{BJK08-comm}. We include the details for completeness.

\begin{lemma}
\label{lemma:ddp-quantum}
   For every $1 \le d \le n$, there is an exact quantum one-way communication protocol for the Directional Derivative Problem $(\DDP_{n,d})$ in which Alice's message is an $n$-qubit phase state and Bob's measurement is a Clifford circuit.  
\end{lemma}
\begin{proof}
Alice sends the $n$-qubit phase state
\[
\ket{\psi_q} = 2^{-n/2}\sum_{x\in \F_2^n} (-1)^{q(x)}\ket{x}.
\]

Let us now state Bob's protocol and prove its correctness. 
Bob receives $\ket{\psi_q}$, together with the nonzero direction $a \in \F_2^n$. 
Let $j \in [n]$ be the first index such that $a_j = 1$. 
Let $R_a = \{x\in\F_2^n\mid x_j=0\}$, where $a$ implicitly defines $j$.
For each $x\in R_a$ and $b\in\F_2$, define
\[
\ket{\phi_{x}^{b}}
= \frac{\ket{x}+(-1)^b\ket{x+a}}{\sqrt 2}.
\]
Since $a_j=1$, every edge $\{x,x+a\}$ contains exactly one endpoint whose $j$-th coordinate is $0$. 
Thus the collection $\calB_a
= \left\{ \ket{\phi_x^b} : x\in R_a,\ b\in\F_2 \right\}$
is an orthonormal basis of the $2^n$-dimensional Hilbert space.

The protocol is as follows: Bob measures $\ket{\psi_q}$ in the basis $\calB_a$. Suppose the outcome is $\ket{\phi_{x}^{b}}$. 
Bob outputs $(x, b)$. 

It remains to prove correctness. Fix a point $x \in \R_a$ and $b\in \F_2$. 
Upon measurement, the amplitude of outcome $\ket{\phi_{x}^{b}}$ is
\[
\braket{\phi_{x}^{b}}{\psi_q} 
= 2^{-n/2}\cdot \frac{1}{\sqrt 2}
\left(
(-1)^{q(x)}+(-1)^b(-1)^{q(x+a)}
\right).
\]
If $b=q(x)+q(x+a)$, then the two terms inside the parentheses are equal, so $ \abs{\braket{\phi_{x}^{b}}{\psi_q}}^2 = 2^{1-n}$.
If $b\neq q(x)+q(x+a)$, then the two terms cancel, and therefore
$\braket{ \phi_{x}^{b}}{\psi_q} = 0.$
Thus, for every $x$, the only possible outcome associated with that edge is the one with $b=q(x)+q(x+a)$.
Therefore, Bob's output bit is always correct. 

We conclude by describing an explicit circuit that implements Bob's measurement in the basis $\calB_a$. We show that the implementation is correct and is a Clifford circuit.
Let
\[
C_a = \prod_{\substack{i\neq j\\ a_i=1}}
\operatorname{CNOT}_{ji},
\]
where the product may be taken in any order, since these CNOT gates have the same control qubit and distinct target qubits. Bob applies $C_a$, then applies a Hadamard gate to qubit $j$, and finally measures all qubits in the computational basis.

We verify that this circuit implements the desired measurement. Let $x\in R_a$, so that $x_j=0$. Then $C_a\ket{x}=\ket{x}$.
On the other hand, the $j$-th bit of $x+a$ is $1$. 
Therefore the CNOT gates flip precisely those coordinates $i\neq j$ for which $a_i=1$, and hence $C_a\ket{x+a} = \ket{x+e_j}$.
It follows that
\[
C_a\ket{\phi_x^b} = \frac{\ket{x}+(-1)^b\ket{x+e_j}}{\sqrt 2}.
\]
Since $x_j=0$, this state factors as
\[
\ket{x_{\neq j}} \otimes \frac{\ket{0}+(-1)^b\ket{1}}{\sqrt 2}
\]
on the $j$-th qubit and the remaining $n-1$ qubits. Applying $H_j$ gives
\[
H_j C_a\ket{\phi_x^b} = \ket{x+b e_j}.
\]
Thus the unitary $H_jC_a$ maps the basis $\calB_a$ to the
computational basis. Therefore measuring in the basis $\calB_a$ is equivalent to applying the Clifford circuit $H_jC_a$ and then measuring in the computational basis.

Concretely, if the final computational-basis outcome is $y\in\F_2^n$, then
Bob sets $b=y_j$ and $x=y+y_j e_j$.
Then $x_j=0$, and the outcome $y$ corresponds to the basis vector
$\ket{\phi_x^b}$. Bob outputs $(x,b)$.

The circuit uses exactly $|a|-1$ CNOT gates and one Hadamard gate (where $|a|$ denotes the Hamming weight of $a$). 
In particular, Bob's measurement is a Clifford measurement of size $O(n)$.
\end{proof}

\subsection{Randomized one-way lower bound}\label{sec:ddp-classical-lower-bound}
\begin{theorem}[Classical lower bound at every degree]
\label{thm:general-directional-lower-bound}
Let $n\ge22$ and $1\le d\le n$. Choose Alice's polynomial to be uniform on $V_{n,d}$ and Bob's direction to be uniform on $\F_2^n\setminus\{0\}$, independently of Alice's input. Then every randomized classical one-way protocol for $\DDP_{n,d}$ with average success probability at least $2/3$ and communication cost $c$ satisfies
\[
    c\ge \kappa \sum_{j=1}^{\min\{d,t\}}\binom{t}{j}, \qquad t=\left\lfloor\frac n2\right\rfloor-4,
\]
where
\[
    \kappa=\frac5{41}\left(1-H_2\left(\frac{11}{18}\right)\right)>0.
\]
Numerically, $\kappa\approx 0.00438$.
\end{theorem}

\cref{thm:general-directional-lower-bound} gives a separation in two regimes: one where $d$ is fixed, and one where $d$ grows with $n$. In the first regime, treating $d$ as constant for sufficiently large $n$ shows that the minimum communication cost is at least $\kappa\sum_{j=1}^d \binom{t}{j} = \Theta_d(n^d)$. Hence any classical randomized one-way protocol for $\DDP_{n,d}$ requires $\Omega_d(n^d)$ bits of communication. In the second regime where $t \le d \le n$, the binomial sum is $2^t-1$, giving an $\Omega(2^{n/2})$ lower bound.

In particular, note that an $n$ versus $2^{n/2}$ Hidden Matching-scale separation can be achieved even with polynomials of degree $d \ge t$ (rather than the full $d = n$) and a restricted set of perfect matchings.

\begin{proof}[\textbf{Proof of \cref{thm:general-directional-lower-bound}}]
For completeness, we first prove the parameterized lower bound, then recover the numerical values in \cref{thm:general-directional-lower-bound} through substitution. Fix an integer $\ell\ge2$, a bias threshold $\gamma\in(0,1/3)$, and a density threshold $\delta$ satisfying $\frac{3}{2^{2\ell-1}}<\delta<\mu_\gamma$ where $\mu_\gamma\coloneqq\frac{1/3-\gamma}{1-\gamma}.$ Set $t \coloneqq \left\lfloor\frac n2\right\rfloor-\ell$. 

Let Alice's input $q$ be uniformly distributed over $V_{n,d}$, and let Bob's input $a$ be uniformly distributed over $\F_2^n\setminus\{0\}$, independent of $q$. It suffices to prove the result for an arbitrary deterministic $c$-bit protocol. Indeed, if a randomized protocol succeeds with probability at least $2/3$ averaged over both the input and the protocol's random coins, then there exists a fixing of the random coins for which the resulting deterministic protocol still succeeds with probability at least $2/3$.

Fix some deterministic $c$-bit protocol $\Pi$. Let $M$ be the random variable denoting Alice's message and define $H(M)$ to be the Shannon entropy of $M$. Note that since a $c$-bit message takes at most $2^c$ values, $H(M) \le c$. Under protocol $\Pi$, Alice's message is a function of her input $q$ which we denote by $\msg: V_{n,d} \rightarrow \F_2^c$. Given a message $m \in \F_2^c$ and nonzero direction $a$, Bob outputs a fixed pair $(x_{m,a}, b_{m,a})\in\F_2^n\times\F_2$. For each possible message $m$, define the message class
\[A_m \coloneqq \set{q\in V_{n,d} \mid \msg(q)=m}.\]

For a fixed message $m$, define the \emph{bias} of Bob's derivative test 
\[\beta_{m,a}\coloneqq \Abs{ \E_{q\sim A_m}\left[(-1)^{D_aq(x_{m,a})}\right]}.\] 
Note that $0\le\beta_{m,a}\le1$. If $\beta_{m,a}=0$, then the derivative bit $D_aq(x_{m,a})$ is uniform over $\F_2$ so it is perfectly unpredictable. On the other hand, if $\beta_{m,a}=1$, then the derivative bit is constant on $A_m$ so it is perfectly predictable.

For each message $m$, define the set of \emph{useful directions}
\[
    \Useful_\gamma(A_m) \coloneqq \left\{a \in \F_2^n\setminus \{0\} \mid \beta_{m,a} \geq \gamma\right\}.
\]
In other words, we say that $a$ is \emph{useful} for $A_m$ if the parity of Bob's edge $\set{x_{m,a}, x_{m,a} + a}$ has bias $\ge \gamma$ over the polynomials in $A_m$. Let $\alpha_m=\abs{A_m}/\abs{V_{n,d}}$. Since Alice's input is uniform on $V_{n,d}$, $\Pr[M = m] = \Pr[q \in A_m] = \abs{A_m}/\abs{V_{n,d}} = \alpha_m$. The entropy of Alice's message is therefore
\[
    H(M) = \sum_m \alpha_m \log_2 \frac1{\alpha_m} = \E_M\left[\log_2\frac{\abs{V_{n,d}}}{\abs{A_M}}\right].
\]
Call a message $m$ \emph{good} if the fraction of useful directions in $A_m$ is at least $\delta$. By the law of total expectation, $\E_M\left[\log_2\frac{\abs{V_{n,d}}}{\abs{A_M}}\right] = \Pr[M \text{ is good}]\cdot     \E_{M}\left[\log_2\frac{\abs{V_{n,d}}}{\abs{A_M}} \,\middle\vert\,  M \text{ is good}\right] + 
    \Pr[M \text{ is not good}]\cdot    
    \E_{M}\left[\log_2\frac{\abs{V_{n,d}}}{\abs{A_M}} \,\middle\vert\, M  \text{ is not good}\right]$. Since $A_m \subseteq V_{n,d}$ implies $\log_2\frac{\abs{V_{n,d}}}{\abs{A_M}}\ge0$,
\begin{equation}\label{eq:entropy-bound}
    \E_M\left[\log_2\frac{\abs{V_{n,d}}}{\abs{A_M}}\right]\geq  \Pr[M \text{ is good}]\cdot     \E_{M}\left[\log_2\frac{\abs{V_{n,d}}}{\abs{A_M}} \,\middle\vert\,  M \text{ is good}\right].
\end{equation}
We may therefore focus on proving a lower bound for good messages only. For $a,x,z \in F_2^n$, define the single-point linear functional $\ev_z:V_{n,d}\to\F_2$ by $\ev_z(q) \coloneqq q(z)$, and the two-point linear functional $\Delta_{x,a}:V_{n,d}\to\F_2$ by $\Delta_{x,a}(q)\coloneqq \ev_x(q) + \ev_{x+a}(q) = D_a q(x).$ We call $\ev_z$ the evaluation functional and $\Delta_{x,a}$ the derivative functional. Since every good message $m$ satisfies $\abs{\Useful_\gamma(A_m)}\ge\delta(2^n-1)$, the main technical lemma (\cref{lem:directional-rank}) gives
\[
    \dim\Span\set{\Delta_{x_{m,a},a}:a\in \Useful_\gamma(A_m)} \ge \sum_{j=1}^{\min\{d,t\}}\binom tj.
\]
Choose a basis of this span from the displayed family.  Every basis functional has bias at least $\gamma$ under the uniform distribution on $A_m$.  By \cref{lem:biased-tests}, every good message $m$ satisfies
\begin{equation}
\label{eq:good-cell-entropy-deficit}
    \log_2\frac{\abs{V_{n,d}}}{\abs{A_m}}
    \ge
    \left(1-H_2\left(\frac{1+\gamma}{2}\right)\right)\sum_{j=1}^{\min\{d,t\}}\binom tj.
\end{equation}

Fact \ref{fact:good-message-probability} states that $\Pr_M[M\text{ is good}]\ge\frac{\mu_\gamma-\delta}{1-\delta}$. Applying this and \cref{eq:good-cell-entropy-deficit} to \cref{eq:entropy-bound} gives the parameterized lower bound:
\begin{equation}\label{eq:parameterized-lower-bound}
    H(M)\ge\frac{\mu_\gamma-\delta}{1-\delta}\left(1-H_2\left(\frac{1+\gamma}{2}\right)\right)\sum_{j=1}^{\min\{d,t\}}\binom tj.
\end{equation}
We now substitute into the parameters. Take
\[
\ell=4,\:\:\:\:
\delta=\frac1{42},
\:\:\: \text{ and }\:\:\:
\gamma=\frac29.
\]
For these parameters, \cref{lem:directional-rank} gives $N = 2^{\lfloor n/2\rfloor-3}$, so $n \ge 22$ implies $N \ge 256$. The lower bound in \cref{eq:large-projection} requires $N \ge 254$, so this choice is valid. Moreover,
\[
    (1+\gamma)/2=11/18, \:\:\: \frac{3}{2^{2\ell-1}} = \frac3{128} < \frac1{42}, \:\:\:\: \mu_\gamma = \frac{1/3-2/9}{1-2/9} = \frac17, \:\:\: \text{ and }\:\:\: \frac{\mu_\gamma-\delta}{1-\delta} = \frac{1/7-1/42}{1-1/42} = \frac5{41}.
\]
Substituting these values into \cref{eq:parameterized-lower-bound} proves \cref{thm:general-directional-lower-bound}.
\end{proof}

\begin{lemma}[Directional rank lemma]
\label{lem:directional-rank}

Fix an integer $\ell\ge2$, a bias threshold $\gamma\in(0,1/3)$, and a density threshold $\delta$ satisfying $\frac{3}{2^{2\ell-1}}<\delta<\mu_\gamma$ where $\mu_\gamma\coloneqq\frac{1/3-\gamma}{1-\gamma}.$ Set $t \coloneqq \left\lfloor\frac n2\right\rfloor-\ell$, as well as $1 \le d \le n$ and $S \subseteq \F_2^n \setminus \set 0$. For each direction $a \in S$, choose an arbitrary point $x_a\in\F_2^n$. 
For sufficiently large $n$, if $\abs{S}\ge\delta(2^n-1)$, then
\[
    \dim\Span\set{\Delta_{x_a,a}:a\in S} \ge \sum_{j=1}^{\min\{d,t\}}\binom tj.
\]
\end{lemma}

\begin{proof}
For every $a\in S$, represent each derivative test $D_aq(x)$ as an edge in the graph $G(V,E)$ where $V = \F_2^n$ and $E(G) = \set{\set{x_a, x_a + a} \mid a \in S}$. Note that $\abs{E(G)} = \abs{S}$. Additionally, the edges are distinct since the direction of an edge $\{u,v\}$ is uniquely determined as $u+v$. Let $m \coloneqq \lfloor\frac n2\rfloor - (\ell - 1)$ and $N\coloneqq 2^m$. Note that this implies $\frac N2 = 2^{m-1} = 2^t$. 

Let $D$ be the maximum degree in the graph. Assume $D\ge N/2$. Let $v$ be a vertex of degree $D$. Define  $S_v = \set{a\in S \mid v \in \set{x_a, x_a + a}}$. The corresponding edges form a connected star with $D + 1 \ge 2^t$ vertices. Applying \cref{prop:connected-edge-rank} to $S_v$ gives $\dim\Span \set{\Delta_{x_a,a} \mid a \in S} \ge \dim\Span \set{\Delta_{x_a,a} \mid a \in S_v}\ge \sum_{j=1}^{\min\{d,t\}}\binom tj$, proving the lemma. We may therefore assume that $D<N/2$.

First convert $G$ to a directed graph. Orient the edges of $G$ so that for every vertex $x \in G$, the indegree $d^-(x)$ and outdegree $d^+(x)$ differ by at most one. This is possible by \cref{prop:balanced-orientation}. In order to bound the rank of the linear functionals, we project the graph onto a smaller subspace. Choose a uniformly random linear map $L:\F_2^n\to\F_2^m$ such that $L(u_e, v_e) = (Lu_e, Lv_e)$ for $(u_e, v_e) \in E(G)$. For $y,z \in \F_2^m$, define
\[
    N_{y,z} \coloneqq \abs{\set{(u_e, v_e)\in E(G):(L(u_e),L(v_e))=(y,z)}}.
\]
Define $Z_L\coloneqq\sum_{y,z\in\F_2^m}N_{y,z}^2$ as the \emph{collision count}, or the number of ordered pairs of edges $e,f \in E(G)^2$ whose projections land on the same output edge. We bound the quantity $\E_L[Z_L]$ through casework:
\begin{itemize}
    \item[1.)] $e = f$. There are exactly $\abs{E(G)} = \abs{S}$ edge self-pairs of the form $(e,e)$ for all $e \in E(G)$. Each edge has the same projection as itself, so each self-pair results in a collision with probability 1.
    \item[2.)] $e \neq f$. Let $e = (u_e, v_e)$ and $f = (u_f, v_f)$ for distinct pairs $e,f \in E(G)$. Then there is a collision iff $L(u_e + u_f) = 0$ and $L(v_e + v_f) = 0$. 
    \begin{itemize}
        \item[2.1)] $(u_e + u_f)$ and $(v_e + v_f)$ are linearly dependent. Let $C$ denote the number of such edge pairs. Over $\F_2$, this implies that either $u_e = u_f$ or $v_e = v_f$. (It is impossible to have the third case $u_e + u_f = v_e + v_f \neq 0$, since this implies that the edges have the same direction, and uniqueness forces $e = f$.) Thus, the edges $e$ and $f$ must either have the same first or second endpoint. At a given vertex $x$, there are $d^+(x)(d^+(x)-1)$ ordered pairs of distinct edges sharing $x$ as the first endpoint, and $d^-(x)(d^-(x)-1)$ ordered pairs of distinct edges sharing $x$ as the second endpoint. Thus 
        \begin{align*}
            C &= \sum_{x \in G}d^+(x)(d^+(x)-1) + d^-(x)(d^-(x)-1)\\
            & \le \sum_{x \in G}\frac{d(x)(d(x)-1)}{2} \tag{\cref{prop:balanced-orientation}}\\
            &\le \frac{D-1}{2}\sum_{x \in G}d(x) \tag{$d(x)\le D$}\\
            &\le (D-1)\cdot \abs{S} \tag{Handshake lemma}
        \end{align*}
        Since $\dim \Span \set{(u_e + u_f), (v_e + v_f)} = 1$, there is effectively one nonzero vector that needs to be mapped to zero. The collision probability is thus $\frac 1N$.
        \item[2.2)] $(u_e + u_f)$ and $(v_e + v_f)$ are linearly independent. Since there are $\abs{S}^2$ ordered pairs of edges in total, this leaves $\abs{S}^2 - \abs{S} - C$ remaining pairs of edges which collide with $\Pr_L[L(u_e + u_f) = L(v_e + v_f) = 0] = \frac1{N^2}$.
    \end{itemize}
\end{itemize} 
Separating the self-pairs and distinct pairs gives
\[
    \E_L[Z_L] = \sum_{e\in E(G)}\Pr_L[Le=Le] + \sum_{\substack{e,f\in E(G)\\e\ne f}} \Pr_L[Le=Lf].
\]
As mentioned above, the self-pairs edge contribute exactly $\abs{S}$ terms. Hence
\begin{align*}
    \E_L[Z_L] &= |S| + \sum_{e,f \in E(G), e\neq f}\Pr_L[Le = Lf] \tag{Bullet 1.}\\
    &= |S| + \frac{C}{N} + \frac{\abs{S}(\abs{S}-1) - C}{N^2} \tag{Bullet 2.}\\
    &\le |S| + \frac{C}{N} + \frac{\abs{S}^2}{N^2} \\
    &\le |S| \left(\frac32 + \frac{\abs{S}}{N^2}\right) \tag{Bullet 2.1 and $D < \frac N2$} 
\end{align*}
Now choose a projection map $L$ for which $Z_L$ is no larger than this expectation. Define $P\coloneqq \abs{\set{(Lu_e, Lv_e) \mid (u_e, v_e) \in E(G)}}$ to be the number of projected edges. By the Cauchy--Schwarz inequality, 
\begin{equation}\label{eq:collision-cs}
    \abs{S}^2=\left(\sum_{y,z}N_{y,z}\right)^2\le PZ_L.
\end{equation}
Let $\rho \coloneqq \frac{\abs{S}}{N^2}$ and $g \coloneqq \delta 2^{2\ell - 2}$, and observe that 
\begin{align*}
    \rho = \frac{\abs{S}}{N^2} &\ge \delta\left(2^{n-2m} - \frac{1}{N^2}\right) \tag{$\abs{S} \ge \delta(2^n-1)$}\\
    &\ge \delta\left(2^{2\ell-2} - \frac{1}{N^2}\right) \\
    &=g - \frac{\delta}{N^2}. 
\end{align*}
Since $x \mapsto \frac{x}{3/2 + x}$ is increasing on $[0, \infty)$ and $g > \frac  32$ by the assumption that $\delta > \frac{3}{2^{2\ell - 1}}$, substituting the $\rho$-expression into \cref{eq:collision-cs} gives 
\begin{equation}\label{eq:large-projection}
    P \ge \frac{|S|^2}{Z_L} \ge \frac{\rho}{3/2 + \rho} N^2 \ge \frac{g - \delta N^{-2}}{3/2 + g - \delta N^{-2}}N^2 = \frac{g}{3/2 + g}N^2 - O(1) > \frac12 N^2 + N. 
\end{equation}
At most $N = 2^m$ of the elements in $P$ are self-pairs. The remaining projected edge pairs form a simple undirected graph $H$ on $\F_2^m$ with $\abs{E(H)} \ge \frac{P-N}{2}$ edges. Let $k$ be the size of the largest connected component of $H$. Since $H$ has $N$ vertices, $\abs{E(H)} \le \frac{kN}{2}$. Combining the upper and lower bounds on $\abs{E(H)}$ gives $k \ge \frac{2\abs{E(H)}}{N} \ge \frac{P}{N}-1 > \frac N2 = 2^t$, where the last inequality holds from \cref{eq:large-projection}. Hence, there exists a connected component $H_0$ of the projected graph with at least $2^t$ vertices. 

Choose an input edge $\set{x_a, x_a + a}$ for every projected edge $\set{Lx_a, Lx_a + L_a}$ of $H_0$, and let $S_{H_0} \subseteq S$ be the set of directions indexing these input edges. \cref{prop:connected-edge-rank} states that the edge parities of a connected graph with at least $2^t$ vertices has rank 
\begin{equation}\label{eq:projection-rank}
    \dim\Span \set{\Delta_{Lx_a, L_a}\mid a \in S_{H_0}} \ge \sum_{j=1}^{\min\set{d,t}} \binom tj.
\end{equation}
It remains to transfer this rank bound to a bound on the original derivative functionals. Define the pullback map 
\[
    T_L:\calP_{\le d}^0(\F_2^m)\to\calP_{\le d}^0(\F_2^n), \qquad T_L(q)\coloneqq q\circ L.
\]
Substituting linear forms into a polynomial does not increase its degree and $L(0)=0$, so $T_L$ is well-defined.  For every $a\in S_{H_0}$ and $q\in\calP_{\le d}^0(\F_2^m)$,
\begin{align*}
    \bigl(T_L^*\Delta_{x_a,a}\bigr)(q) &= \Delta_{x_a,a}(q\circ L)\\
    &= q(L(x_a+a))+q(Lx_a)\\
    &= \Delta_{Lx_a,La}(q).
\end{align*}
Thus every projected edge functional is the image under $T_L^*$ of one of the original derivative functionals. Since a linear map on a set of vectors cannot increase its rank, the original set of linear functionals has at least the projected rank:
\begin{align*}
    \sum_{j=1}^{\min\set{d,t}}\binom tj &\le \dim\Span\set{\Delta_{Lx_a, La}\mid a \in S_{H_0}}\tag{\cref{eq:projection-rank}} \\
    &\le\dim\Span\set{\Delta_{x_a, a}\mid a \in S} 
\end{align*}
This completes the proof.
\end{proof}

\begin{lemma}[Entropy of biased linear functionals]
\label{lem:biased-tests}
Let $V$ be an $N$-dimensional vector space over $\F_2$, let $\varnothing\ne A\subseteq V$, and let $\lambda_1,\ldots,\lambda_r\in V^*$ be linearly independent. If $\abs{\E_{x\sim A}(-1)^{\lambda_i(x)}} \ge\gamma$ for every $i\in [r]$, then
\[
    \log_2\frac{\abs V}{\abs A} \ge \left(1-H_2\left(\frac{1+\gamma}{2}\right)\right)r.
\]
\end{lemma}

\begin{proof}
Let $X$ be uniformly distributed on $A$, and define the combined linear map
\[
\Lambda: V \to \F_2^r, \qquad \Lambda(x) \coloneqq (\lambda_1(x),\ldots,\lambda_r(x)).
\]
Set $Y \coloneqq \Lambda(X)$. $Y$ records all the outcomes of the $r$ linear tests applied to $X$. Since $\lambda_1,\ldots,\lambda_r\in V^*$ are linearly independent, $\operatorname{rank} \Lambda = r$. By the rank-nullity theorem, $\dim \ker \Lambda = N - r$. Every nonempty $\Lambda^{-1}(y) = \set{x\in V \mid \Lambda(x) = y}$ is a shift of $\ker(\Lambda)$ and thus contains $2^{N-r}$ points. Conditioned on $Y = y$, the random variable $X$ is supported on $A \cap \Lambda^{-1}(y)$, which contains at most $2^{N-r}$ points. Hence
$H(X\mid Y = y)\le N-r$ for every $y$, so 
\begin{equation}\label{eq:cond-entropy-bound}
    H(X\mid Y)\le N-r.
\end{equation}
Next, we bound the entropy of each coordinate $Y_i = \lambda_i(X)$. Let $\beta_i \coloneqq \abs{\E\left[(-1)^{Y_i}\right]}$. Since $Y_i \in \set{0,1}$,
$\E_{x \sim A}\left[(-1)^Y_i\right] = \Pr[Y_i = 0] - \Pr[Y_i = 1]= 2\Pr[Y_i = 0] - 1$. Thus $\beta_i = \abs{2\Pr[Y_i = 0] - 1}$. 
It follows that $\Pr[Y_i = 0] = \frac{1\pm \beta_i}{2}$. Since binary entropy is symmetric, i.e. $H_2(p) = H_2(1-p)$, either possibility gives 
\[
    H(Y_i) = H_2\left(\frac{1 +\beta_i}{2}\right). 
\]
Since $Y$ is determined by $X$, the chain rule gives 
\begin{align*}
    H(X) &= H(Y) + H(X\mid Y)\\
    &\le \sum_{i=1}^r H(Y_i) + N-r\\
    &= \sum_{i=1}^rH_2\left(\frac{1 +\beta_i}{2}\right) + N-r\\
    &= N - \sum_{i=1}^r\left(1 - H_2\left(\frac{1 +\beta_i}{2}\right)\right)
\end{align*}
where the first inequality holds using \cref{eq:cond-entropy-bound} and the subadditivity inequality $H(Y) \le \sum_{i=1}^r H(Y_i)$.
Since $X$ is uniform on $A$, $H(X) = \log_2\abs{A}$. Moreover, since $V$ is $N$-dimensional over $\F_2$, $N = \log_2\abs{V}.$ Rearranging therefore gives 
\[
    \log_2\frac{\abs{V}}{\abs{A}} \ge \sum_{i=1}^r \left(1 - H_2\left(\frac{1 + \beta_i}{2}\right)\right).
\]
Lastly, suppose $\beta_i \ge \gamma$ for every $i \in [r]$. The function $p \mapsto H_2\left(\frac{1+p}{2}\right)$ is decreasing on $[0,1]$, so $H_2\left(\frac{1+\beta_i}{2}\right)\le H_2\left(\frac{1+\gamma}{2}\right)$ for every $i \in [r]$. Hence 
\[
    \log_2\frac{\abs V}{\abs A} \ge \left(1-H_2\left(\frac{1+\gamma}{2}\right)\right)\cdot r.
\]

\end{proof}

\begin{proposition}[A connected graph supports many independent edge parities]
\label{prop:connected-edge-rank}
Let $H$ be a connected graph whose vertex set $A$ is contained in $\F_2^m$. If $\abs{A}\ge2^r$, then
\[
\dim\Span\set{\ev_u + \ev_v\mid\{u,v\}\in E(H)}
\ge
\sum_{j=1}^{\min\{d,r\}}\binom rj.
\]
\end{proposition}

\begin{proof}
Let $\widetilde{V} = \calP_{\le d}(\F_2^m)$ be the full polynomial space, and let $\widetilde{\operatorname{ev}}_z(q) = q(z)$ for all $z \in A$. Define \[E_{\operatorname{full}} = \Span\set{\widetilde{\operatorname{ev}}_u + \widetilde{\operatorname{ev}}_v \mid \set{u,v} \in E(H)}.\] Since every functional in $E_{\operatorname{full}}$ vanishes on constant polynomials, $E_{\operatorname{full}}$ has the same dimension as the span of the corresponding edge functionals on $\calP_{\le d}^0$. 

Fix $z_0\in A$.  For each $z\in A$, choose a path in $H$ from $z_0$ to $z$.  Summing the edge functionals along this path gives
\[
\widetilde{\operatorname{ev}}_{z_0}+\widetilde{\operatorname{ev}}_z,
\]
since every intermediate vertex is added twice and cancels over $\F_2$.  It follows that $\widetilde{\operatorname{ev}}_{z_0} + \widetilde{\operatorname{ev}}_z$ belongs to $E_{\operatorname{full}}$. Therefore
\[
\dim E_{\operatorname{full}} = \dim \Span\set{\widetilde{\operatorname{ev}}_z\mid z\in A}-1\ge \sum_{j=1}^{\min(d,r)}\binom rj,
\]
where the last inequality follows from the Hilbert-function bound {\cite[Corollary~4]{GGHNY24}}. Specifically, the span of the evaluation map is the degree-$d$ Hilbert function of $A$. Golovnev, Guo, Hatami, Nagargoje, and Yan prove that, for every nonempty $A\subseteq\F_2^m$, this rank is at least
\[
\sum_{j=0}^{\min\{d,\lfloor\log_2\abs{A}\rfloor\}}
\binom{\lfloor\log_2\abs{A}\rfloor}{j}.
\]
Since $\abs{A}\ge2^r$, the result follows. 
\end{proof}

\begin{proposition}[Balanced orientation]
\label{prop:balanced-orientation}
Every finite graph admits an orientation such that, at every vertex $v$,
\[
\abs{d^+(v)-d^-(v)}\le1.
\]
\end{proposition}

\begin{proof}
By the Handshake Lemma, the number of odd-degree vertices in any graph is always even. In each connected component, group the vertices of odd degree into pairs of two and add one auxiliary edge for each pair. The resulting multigraph has only even degrees, and hence admits an Euler tour in every connected component. Orient each edge along the traversal of the Euler tour so that the indegree and outdegree agree at every vertex, then delete the auxiliary edges. Deleting at most one auxiliary edge incident to each vertex leaves at most a difference of one.
\end{proof}

\begin{fact} \label{fact:good-message-probability}
    \begin{equation}
    \Pr_M[M\text{ is good}]\ge\frac{\mu_\gamma-\delta}{1-\delta}.
\end{equation}
\end{fact}

\begin{proof}
    Let $\nu_m$ be the fraction of useful directions in $A_m$.
    Observe that $\forall \: u,v\in\F_2$, $\mathbf{1}[u=v] = \frac{1+(-1)^{u+v}}2.$ Applying this fact to Bob's success probability gives
    \begin{align*}
        &\Pr_{q\sim A_m}
        \left[
        b_{m,a}=D_aq(x_{m,a})
        \right] \\
        &\qquad=
        \frac12
        +
        \frac12(-1)^{b_{m,a}}
        \E_{q\sim A_m}\left[
        (-1)^{D_aq(x_{m,a})}\right] \\
        &\qquad\le
        \frac{1+\beta_{m,a}}2.
    \end{align*}
    Since the protocol succeeds with probability at least $2/3$, averaging the preceding inequality over all messages $m$ and nonzero directions $a$ gives $\frac23 \le \frac12+\frac12\E_{m,a}[\beta_{m,a}]$. Therefore
    \begin{equation}
    \label{eq:average-bias-1}
        \E_{m,a}[\beta_{m,a}] \ge \frac13.
    \end{equation}
    Simultaneously, the average bias is upper bounded by
    \begin{equation}
    \label{eq:average-bias-2}
        \E_a[\beta_{m,a}]\le 1\cdot\nu_m + \gamma\cdot (1-\nu_m)
    \end{equation}
    Combining \cref{eq:average-bias-1} and \cref{eq:average-bias-2} lower bounds the fraction of useful directions by
    \begin{equation}
    \label{eq:average-useful-directions}
        \E_m[\nu_m] \ge \frac{1/3-\gamma}{1-\gamma} = \mu_\gamma.
    \end{equation}
    Combining this with an upper bound on the number of useful directions gives
    \[
        \mu_\gamma \le \E_m[\nu_m]\le 1\cdot \Pr_M[M\text{ is good}]+\delta\cdot (1-\Pr_M[M\text{ is good}]),
    \]
    which is a rearrangement of the desired inequality. 
\end{proof}

\subsection{Optimality in the Clifford hierarchy message model}
We now place the Directional Derivative Problem within a broader one-way communication model based on the Clifford hierarchy and show that its separation in this regime is also optimal. 

Let $\calP_m$ denote the $m$-qubit Pauli group. For $k \ge 1$, define the $k$th level of the Clifford hierarchy~\cite{gottesman1999demonstrating} as
\begin{align*}
    \calC_1(m) &= \set{e^{i\theta}P \mid P \in \calP_m, \theta \in R}\\
    \calC_k(m) &= \set{U \mid UPU^\dagger \in \calC_{k-1}(m) \text{ for all } P \in \calP_m}
\end{align*}
where $U$ ranges over all $m$-qubit unitaries. In particular, $\calC_2(m)$ is the Clifford group. Call a \emph{level-k message} a state $U\ket{0^m}$ where $U\in \calC_k(m)$ is some $m$-qubit unitary from the $k$-th level of the Clifford hierarchy. The \emph{level-$k$ message model} considers protocols in which Alice sends a level-$k$ message on every input she receives, while Bob’s measurement is allowed to be arbitrary.
We first verify that the phase state messages used in \cref{lemma:ddp-quantum} are level-$d$ messages, and hence $\DDP_{n,d}$ admits a protocol in the level-$d$ message model.
\begin{lemma}[Hierarchy level of degree-at-most-$d$ phase states]
    For $d \ge 2$, let $q: \F_2^n \to F_2$ be a polynomial of degree at most $d$. Then the $n$-qubit phase state $\ket{\psi_q} \coloneqq 2^{-n/2}\sum_{x\in \F_2^n} (-1)^{q(x)}\ket{x}$ is equal to $U_q\ket {0^n}$ for some $U_q \in \calC_d(n)$. 
\end{lemma}
\begin{proof}
    Let $W_q\ket x \coloneqq (-1)^{q(x)}\ket x$. We claim that for every $1\le d \le n$ and $q \in V_{n,d}$, $W_q \in \calC_d(n)$. Consequently, for every $d \ge 2$, $\ket{\psi_q} = W_qH^{\otimes n}\ket{0^n}$ is a level-$d$ message. We prove this via induction on $d$. For $d=1$, $W_q$ is a Pauli up to global phase and hence lies in the Clifford hierarchy. Note that $q \in V_{n,d}$ implies $\deg q \le d$ since Alice's input $q$ belongs to the space of polynomials of degree at most $d$. Since $D_aq$ is a derivative which eliminates the highest-degree term of $q$ in $\F_2$, $\deg(D_aq) \le d-1$. Moreover, its constant term can be removed since $W_{D_aq} = (-1)^{D_aq(0)}W_{D_aq - D_aq(0)}$. Assume by induction that $W_{D_aq} \in \calC_{d-1}(n)$. For all $a,b \in \F_2^n$, define $X^a \ket x = \ket{x + a}$ and $Z^b \ket x = (-1)^{b\cdot x} \ket x$. Then \begin{align*}
    W_q(X^aZ^b)W^\dagger_q\ket x  &= (-1)^{q(x)}W_q(X^aZ^b)\ket x\\
    &= (-1)^{q(x) + b\cdot x} W_q \ket {x + a}\\
    &= (-1)^{q(x) + b\cdot x + q(x + a)}\ket{x + a}\\
    &=(-1)^{D_aq(x) + b\cdot x}\ket {x+a}\\
    &= X^aZ^bW_{D_aq}\ket x
\end{align*}
Because every level $d\ge 2$ of the Clifford hierarchy is closed under Clifford conjugation and multiplication by a Clifford on either side, $W_q \in \calC_d(n)$.  Therefore  \[\ket{\psi_q} = W_qH^{\otimes n}\ket{0^n}, \:\:\: \text{ where }\:\:\:  W_q H^{\otimes n} \in \calC_d(n).\]Note that for $d=2$, $ W_qH^{\otimes n}$ is a Clifford gate so $\ket{\psi_q}$ is a stabilizer state. 
\end{proof}

A natural question is: among problems that admit an $m$-qubit level-$k$ protocol, how much classical communication can they require? The next lemma shows that the amount of classical communication needed to simulate one-way protocols where Alice sends a level-$k$ message, for every level $k$, is $O_k(m^k)$, where Bob's measurement can be arbitrary. 

\begin{lemma}[Classical simulation of a level-$k$ message] \label{lem:ch-simulation-bound}
For all $k \ge 1$, every one-way protocol in which Alice sends a level-$k$ message has a classical randomized one-way simulation using $O_k(m^k)$ bits. 
\end{lemma}
\begin{proof}
Let $a_k(m)$ be the number of unitaries in $\calC_k(m)$ up to global phase. We use induction. For $k = 1$, $a_1(m) = 4^m$. For $k \ge 2$, every fixed unitary $U \in \calC_k(m)$ determines the ordered list
\begin{align*}
    UX_jU^\dagger, UZ_j U^\dagger \tag{$j = 1, \cdots, m$}.
\end{align*}
This ordered list determines $U$ up to global phase, since two unitaries with the same list must differ by an operator that commutes with every Pauli (i.e. a scalar). Therefore, counting the number of possible lists provides a valid upper bound on $a_k(m)$. We now count the number of possible lists. Each entry is obtained by conjugating a Pauli generator by $U$, so is a Hermitian involution. Two such entries that differ by a global phase must either have opposite signs or be identical, so there are two operators for each equivalence class up to global phase. Therefore each entry has at most $2a_{k-1}(m)$ values. Since there are $2m$ entries, 
\[
    a_k(m) \le (2a_{k-1}(m))^{2m}.
\]
Let $b_k(m) \coloneqq \log_2a_k(m)$. Then 
\begin{align*}
    b_1(m) &= 2m\\
    b_k(m) &\le 2m(b_{k-1}(m) + 1)
\end{align*}
implies $b_k(m) = O_k(m^k)$.

The protocol is as follows: For any fixed $k \ge 1$, Alice and Bob pre-agree on an enumeration of the $a_k(m)$ unitaries up to global phase. During communication, Alice sends the index of the class of her preparation unitary. Bob then identifies the corresponding state and samples a measurement outcome which is invariant under global phase. This occurs with $O_k(m^k)$ bits.
\end{proof}

\begin{corollary}[Fixed-degree classical complexity]\label{cor:fixed-degree}
For every fixed $d\ge1$, the randomized classical one-way communication complexity of $\DDP_{n,d}$ is $\Theta_d(n^d)$. More precisely, every protocol succeeding with probability $\ge 2/3$ on uniform inputs and using $c$ bits of communication satisfies 
\[
    c \ge \left(\kappa 2^{-d} + o_d(1)\right)\binom nd,
\]
where $\kappa$ is the absolute constant in \cref{thm:general-directional-lower-bound}.
\end{corollary}

\begin{proof}
For fixed $d$, the leading term in the sum from \cref{thm:general-directional-lower-bound} is
\[
\binom{\lfloor n/2\rfloor-4}{d}
=
\left(2^{-d}+o_d(1)\right)\binom nd.
\]
This gives the lower bound. For the upper bound, Alice sends the coefficients of all nonconstant monomials of degree at most $d$. There are $\sum_{j=1}^d \binom nj$ coefficients, and each coefficient requires one bit. Bob reconstructs $q$ and outputs $(0, q(a))$. This is valid because $q(0) = 0$ and $D_aq(0) = q(a) + q(0) = q(a)$. 
\end{proof}

\begin{corollary}[Optimality and hierarchy-level tradeoff]\label{cor:ddp-tradeoff}
    Fix $d \ge 2$ and $1 \le k \le d$. Any one-way protocol for $\DDP_{n,d}$ with $m$-qubit and level-$k$ messages, succeeding with probability $\ge 2/3$ on uniform inputs, must have \[m = \Omega_{d,k}(n^{d/k}).\]
\end{corollary}
\begin{proof}
Combining the lower bound from \cref{cor:fixed-degree} and the upper bound from \cref{lem:ch-simulation-bound}  gives 
\[
    \Omega_d(n^d) \le O_k(m^k)
\]
Hence $m = \Omega_{d,k}(n^{d/k})$. For $k=d$, this gives the lower bound $m = \Omega_d(n)$, matching the $n$-qubit protocol of \cref{lemma:ddp-quantum} asymptotically.  For $k < d$, no $n$-qubit level-$k$ protocol can achieve the stated success probability for sufficiently large $n$. Hence, an $n$-qubit level-$d$ protocol has optimal qubit communication up to a factor depending on $d$, proving that $\DDP_{n,d}$ asymptotically arrives at the largest possible quantum--classical separation within the Clifford hierarchy message model.
\end{proof}

\subsection{Optimality under state-preparation complexity}\label{sec:ddp-gate-bound}
In this section, we show that the worst-case state-preparation complexity of our protocol in \cref{lemma:ddp-quantum} is asymptotically optimal. Fix a finite gate set of one- and two-qubit gates. If every $n$-qubit message can be prepared by a circuit of at most $G$ gates, then Alice can send a description of her preparation circuit using $O(G \log n)$ bits. Bob can then simulate his measurement classically. For a protocol solving $\DDP_{n,d}$ with success probability at least 2/3 on uniform inputs, our $\Omega(n^d)$ classical lower bound therefore implies $G = \Omega_d(n^d/ \log n)$. We now give a matching upper bound.
\begin{corollary}[Synthesis of bounded-degree phase states]\label{cor:ddp-gate-bound}
    Fix $d \ge 2$, and let $q: \F_2^n \to \F_2$ be a polynomial with degree at most $d$. Every $n$-qubit phase state 
    \[
        \ket{\psi_q} = 2^{-n/2}\sum_{x \in \F_2^n}(-1)^{q(x)}\ket x
    \]
    can be prepared using $O_d(n^d/\log n)$ one- and two-qubit gates up to global phase from a finite gate set depending only on $d$.
\end{corollary}

\begin{proof}
    Construct the decomposition
    \[
        q(x) = \sum_{\substack{S\subseteq [n]\\\abs{S}\le d-2}}x_Sq_S(x_{[n]\setminus S}),
    \]
    where $x_S \coloneqq \prod_{i \in S}x_i$ with $x_{\varnothing} = 1$ and $q_S$ is the quadratic polynomial left after factoring out $x_S$ from all monomials whose first $|S|$ indices are $S$. In more detail, first place all degree-at-most-two monomials of $q(x)$ in $q_{\varnothing}$. The decomposition includes these terms since $x_{\varnothing}q_{\varnothing} = q_{\varnothing}$. Finally, for each monomial $x_{i_1}\cdots x_{i_r}$ of degree $r \ge 3$ with $i_1 < \cdots < i_r$, assign
    \[
        x_{i_1}\cdots x_{i_r} = \underbrace{x_{i_1}\cdots x_{i_{r-2}}}_{x_S}\underbrace{x_{i_{r-1}}x_{i_r}}_{\text{term added to }q_S}
    \]
    where $S = \set{i_1,\ldots,i_{r-2}}$. 
    Each quadratic phase 
    \[
        W_{q_S}\ket y \coloneqq (-1)^{q_S(y)}\ket y
    \]
    is Clifford. By~\cite{Aaronson_2004}, there is an $O(n^2/\log n)$-gate circuit $C_S$ over $\set{H,P,\mathrm{CNOT}}$ that implements $W_{q_S}$ up to a scalar. Namely, $C_S\ket y = \omega_S (-1)^{q_S(y)}\ket y$ where $\omega_S$ is an eighth-root of unity. Let $R \coloneqq (HP)^3 = e^{i\pi/4}I$. By appending a constant number of copies of $R, R^{-1}$ circuits to $C_S$, we obtain a circuit $\widetilde{C_S}$ that implements $W_{q_S}$ exactly. This phase normalization incurs only $O(1)$ gates. Conditioning $\widetilde{C_S}$ on the qubits indexed by $S$ and applying the exact decompositions of \cite{Barenco_1995} replaces each gate of $\widetilde{C_S}$ by $O_d(1)$ one-qubit and CNOT gates.

    There are $O_d(n^{d-2})$ sets $S$. Applying $H^{\otimes n}$ before these circuits therefore prepares $\ket{\psi_q}$ with 
    \[
        O(n) + O_d(n^{d-2})O_d(n^2/\log n) = O_d(n^d/\log n)
    \]
    gates. 
\end{proof}

\appendix

\printbibliography

\section{The Graph Hadamard Support Problem}\label{sec:ghs}

Recall that setting $d=2$ as per \cref{sec:DDP} gives an $n$ versus $\Theta(n^2)$ separation using stabilizer-state messages and Clifford measurements. In this appendix, we give an alternative construction achieving the same asymptotic separation. This problem has two additional features:
\begin{itemize}
    \item[1)] Bob's measurement is a \emph{product} measurement, implemented using only single-qubit Hadamard gates. This makes the protocol particularly relevant for the quantum information supremacy setting discussed in \cref{subsec:toward}.
    \item[2)] The classical lower bound employs a different proof than \cref{sec:ddp-classical-lower-bound}, which we include because it may be useful independently. To our knowledge, there is no reduction in either direction between $\DDP_{n,2}$ and $\GHS$.
\end{itemize}

We begin by defining the communication problem that exhibits the separation. For a symmetric zero-diagonal adjacency matrix $A\in\F_2^{n\times n}$, define the $n$-qubit graph state
\begin{equation}
    \ket{\psi_{A}}
    =
    2^{-n/2}\sum_{x\in\F_2^n}
    (-1)^{q(x)}\ket{x}
\end{equation}
where $q(x)=\sum_{1\le u<v\le n} A_{uv}x_ux_v$ is a quadratic polynomial in $x$ (we will later view this as the polynomial $q(A,x)$ where each monomial term has bidegree $(1,2)$ in the entries of $(A,x)$ respectively). For $S\subseteq[n]$, let $H_S$ denote the Clifford unitary that applies a Hadamard gate to every qubit in $S$ and the identity elsewhere. For an $n$-qubit state $\ket\phi$ and $n$-qubit quantum circuit $Q$, let $\supp(Q\ket\phi)$ denote the set of all bitstrings $x \in \F_2^n$ such that $|\bra{x}Q\ket\phi|^2 > 0$. We call $\supp(Q\ket\phi)$ the \emph{support} of $Q\ket\phi$. 

\begin{tcolorbox}[problembox]
\begin{problem}[Graph Hadamard Support Problem, $\GHS_n$]
    Alice receives a symmetric zero-diagonal matrix $A\in\F_2^{n\times n}$ specifying the graph state $\ket{\psi_A}$, and Bob receives a subset $S \subseteq [n]$ specifying the partial Hadamard operator $H_S$. After receiving a message from Alice, Bob must output a string $z \in \F_2^n$ such that $z \in \supp(H_S\ket{\psi_A})$.
\label{problem:GHS}
\end{problem}
\end{tcolorbox}
The main result of this section is that the Graph Hadamard Support Problem $\GHS_n$ exhibits an optimal $n$ versus $\Theta(n^2)$ quantum--classical separation in one-way communication. 

\begin{theorem}\label{thm:ghs-separation}
    There is an exact one-way communication protocol for $\GHS_n$ using $n$ qubits. Furthermore, every randomized one-way protocol for $\GHS_n$ that succeeds with probability at least $2/3$ uses $\Omega(n^2)$ bits of communication. Lastly, there exists an exact deterministic classical one-way protocol that uses $O(n^2)$ bits of communication.
\end{theorem}

The rest of this section focuses on proving \cref{thm:ghs-separation}. In \cref{sub:ghs-upper-bounds}, we prove the quantum and classical upper bounds. In \cref{sub:ghs-classical-lower-bound}, we prove the classical lower bound. Together, they imply \cref{thm:ghs-separation}. 

\subsection{Upper bounds}\label{sub:ghs-upper-bounds}
We show that there is an exact classical one-way protocol for $\GHS_n$ using $O(n^2)$-bits. Alice sends the $\binom n2$ upper-triangular entries of $A$. Bob then reconstructs $\ket{\psi_A}$ from its classical description and computes the amplitudes of $H_S\ket{\psi_A}$, outputting any bitstring with nonzero amplitude.

On the other hand, there is an exact quantum one-way protocol using only $n$ qubits: Alice sends $\ket{\psi_A}$, and Bob applies $H_S$ and measures in the computational basis, obtaining an outcome that lies in the support with probability 1.

\subsection{Randomized one-way lower bound}\label{sub:ghs-classical-lower-bound}

We now prove the technical part of \cref{thm:ghs-separation}, the classical lower bound. We will use the following definitions throughout the next two proofs.

\begin{definition}[Zero set of an ideal]\label{def:zero-set}
    For polynomials $p_1,\cdots,p_s \subseteq \F_2[x_1,\cdots,x_n]$, let $I = \langle p_1,\cdots p_s\rangle$.  Define the zero set of $I$ to be \[ \calV(I) = \{x \in \F_2^n \mid p_i(x) = 0 \: \forall \: i \in [s]\}.\] Equivalently, for a single polynomial $p$ let \[\calV(p) = \{x \in \F_2^n \mid p(x) = 0\}.\]
\end{definition}

\begin{definition}[Standard monomials]\label{def:standard-monomials}
    For a fixed monomial order $\prec$ and nonzero polynomial $f\in \F_2[x_1,\cdots,x_n]$, let $\operatorname{LM}_\prec (f)$ denote the leading monomial of $f$ with respect to $\prec$. For an ideal $I \subseteq \F_2[x_1,\cdots,x_n]$, define its initial ideal by $\operatorname{in}_\prec (I) \coloneqq \langle \operatorname{LM}_\prec(f) \mid f\neq 0 \in I\rangle$. Then the set \[\Delta_\prec(I)\coloneqq \{M \mid M \text{ is a monomial and } M\notin\operatorname{in}_\prec(I)\}\] consists of the \emph{standard monomials} of $I$.
\end{definition}

\begin{definition}[Binary relative entropy]\label{def:binentropy}
    For $p,q \in (0,1)$, define the binary relative entropy \[H_2(p||q) = p\log_2\frac pq + (1-p) \log_2 \frac{1-p}{1-q}\]
\end{definition}

\begin{fact}[Number of stabilizer states \cite{Aaronson_2004}]
The number of $n$-qubit stabilizer states is \begin{equation}\label{eq:num-stabilizer-states}
   |\Stab_n| = 2^n \prod_{j=1}^n (2^j + 1) = 2^{(\frac12 + o(1))n^2}.
\end{equation}
\end{fact}

\begin{fact}[Support criterion]\label{fact:support-criterion}
Let $\bar{S}\coloneqq[n]\setminus S$. A string $z\in\F_2^n$ lies in $\supp(H_S\ket{\psi_{A}})$ iff 
\begin{equation}
    q(A_{S,S},y) + y^\top\bigl(z_S+A_{S,\bar{S}}z_{\bar{S}}\bigr) = 0
    \qquad
    \text{for every }y\in\ker(A_{S,S}).
\end{equation}
A proof of Fact~\ref{fact:support-criterion} can be found at the end of the section.
\end{fact}
    
We give two lower bounds for $\GHS_n$: one for exact deterministic protocols in \cref{thm:ghs-exact-lower-bound}, and one for bounded-error randomized protocols in \cref{thm:ghs-approx-lower-bound}. 

\begin{theorem} \label{thm:ghs-exact-lower-bound}
    Any classical deterministic one-way protocol that solves $\GHS_n$ with success probability 1 requires $\Omega(n^2)$ bits of communication.
\end{theorem}
\begin{proof}
    We restrict Bob's input to partial Hadamard circuits $H_S$ with $|S| = 3$, where $S$ ranges over all $\binom n3$ qubit triples. Since this is a restriction of the input domain of $\GHS_n$, a lower bound for these inputs is a lower bound for $\GHS_n$.
    
    Set $N \coloneqq \binom n2$ and identify Alice's input $x_A \in F_2^N$ with the upper-triangular entries of $A$. Define $\msg: \F_2^N \rightarrow \F_2^c$ via $\msg(\mathbf{x}) = m$, where $m$ is the message Alice sends on input $\mathbf{x}$ in the arbitrary deterministic protocol $\Pi$. For every message $m$, define the following ``message class''
    \[C_m = \{\mathbf{x} \mid \msg(\mathbf{x}) = m\}.\]
    Since $A$ has $\binom n2$ bits of freedom, the size of the set of all possible Alice inputs is
    \begin{equation}\label{eq:num-graph-states}
        \sum_m |C_m| = 2^{\binom n2} = 2^{(\frac12+o(1))n^2}.
    \end{equation}
    
    The key technical insight is that each message class $C_m$ cannot contain too many inputs. We show in \cref{lem:ghs-message-class-upper-bound} that for every message $m$,
        \begin{equation}
        |C_m| \le 2^{(\frac14 + o(1))n^2}.
        \end{equation}
    Combined with \cref{eq:num-graph-states}, there are at least $2^{(\frac14 - o(1))n^2}$ distinct messages, so Alice must send at least $(1/4 - o(1))n^2 = \Omega(n^2)$ bits, completing the proof of the lower bound.
\end{proof}

\begin{lemma}\label{lem:ghs-message-class-upper-bound}
    For every message m,
    \begin{equation}
    |C_m| \le 2^{(\frac14 + o(1))n^2}.
    \end{equation}
\end{lemma}
\begin{proof}
    For a fixed $m$ and $S$, let $z = z_m(S)$ be Bob's output on message $m$ and triple $S = \{i,j,k\}$ for $1 \le i < j < k \le n$. For an arbitrary adjacency matrix $A \in \F_2^{n\times n}$, let $\mathbf{x}_A$ be the associated Alice input and define
    \begin{equation}
        A_{S,S}=
        \begin{pmatrix}
        0&A_{i,j}&A_{i,k}\\
        A_{i,j}&0&A_{j,k}\\
        A_{i,k}&A_{j,k}&0
        \end{pmatrix}
    \end{equation}
    where the off-diagonal matrix entries are Boolean variables. Set  
    \begin{equation}
        y \coloneqq (A_{j,k},A_{i,k},A_{i,j}) \in \ker{A_{S,S}}
    \end{equation}
    where one can check that indeed $A_{S,S}y^\top = 0 \pmod 2$. Substituting $y$ into Fact~\ref{fact:support-criterion} gives the following ``support polynomial'' over edge variables:
    \begin{equation}\label{eq:ghs-triple-exact-constraint}
        p_S^m(A)
        \coloneqq
        3A_{ij}A_{jk}A_{ik}
        +\left(\sum_{b \not\in S}(A_{ib}A_{jk} + A_{jb}A_{ik} + A_{kb}A_{ij})z_b\right)
        +(z_k A_{ij} + z_j A_{ik} + z_i A_{jk})
        =0 
    \end{equation}
    Recall that $q$ has bidegree $(1,2)$, and since the entries of $y$ are linear in $A$, $q(A_{S,S},y)$ is now a cubic polynomial in $A$. Since $A$ is arbitrary and the protocol is exact, we must have $z \in \bigcap_{\mathbf{x}_A \in C_m} \supp(H_S \ket{\psi_A}).$ This implies \[p_S^m(A) = 0 \qquad \forall \: \mathbf{x}_A \in C_m.\]
    
    Since $S$ is an arbitrary triple, using \cref{def:zero-set} gives $C_m \subseteq \cap_S \calV(p_S^m)$. This implies    \begin{equation}\label{eq:message-upper-bound}
        |C_m| \le |\bigcap_S \calV(p_S^m)|.
    \end{equation}

    This shows that the number of stabilizer inputs for each message class is upper bounded by the number of common zeros of its corresponding support polynomial. By \cite{geil2017boundingnumbercommonzeros}, the number of zeros of a system of polynomial equations can be upper bounded by the number of its corresponding standard monomials. Namely, let $I_m \coloneqq \langle p_S^m : |S| = 3\rangle$. Applying Corollary 1 of \cite{geil2017boundingnumbercommonzeros} gives
    \begin{equation}\label{eq:zeros-upper-bound}
        |\calV(I_m)| \le |\Delta (I_m + \langle A_{ij}^2-A_{ij}: 1 \le i < j \le n\rangle)|,
    \end{equation}
    where $A_{ij}^2-A_{ij}$ gives the idempotent Boolean relation. 
    Fix any degree-based monomial order. We make the following two observations: By a direct application of \cref{def:standard-monomials}, $\operatorname{LM}(A_{ij}^2-A_{ij}) = A_{ij}^2$ implies all standard monomials must be squarefree. Additionally, by \cref{eq:ghs-triple-exact-constraint},  \[\operatorname{LM}(p_S^m) = A_{ij}A_{ik}A_{jk}.\]

    Define a triangle to be a set of three vertices with all three edges present; a graph is triangle-free if it contains no triangles. Let $M = \prod_{1\leq i < j\leq n} A_{ij}^{a_{ij}}$ be any standard monomial given boolean values $a_{ij}$, and define the edge set $E(M) = \{(i,j) : a_{ij} = 1\}$. Since no standard monomial can be divisible by the product of any three edges of a triangle, this is equivalent to saying that the induced graph $G_M([n],E(M))$ must be triangle-free. 
    
    Let  $\operatorname{tf}(n)$ be the number of triangle-free graphs on $n$-vertices. This gives an upper bound on the total number of squarefree standard monomials,
    \begin{equation}\label{eq:monomials-upper-bound}
        |\Delta (I_m + \langle A_{ij}^2-A_{ij}: 1 \le i < j \le n\rangle)| \le \operatorname{tf}(n).
    \end{equation}
    Combining \cref{eq:message-upper-bound} - \cref{eq:monomials-upper-bound} gives $|C_m|\le \operatorname{tf}(n)$ for every $m$, and by the Erd\"{o}s--Kleitman--Rothschild theorem on the number of triangle-free graphs \cite{erdos1976asymptotic},
    \[
        \operatorname{tf}(n)=2^{\frac{n^2}{4}+o(n^2)}.
    \]
\end{proof}

\begin{theorem} \label{thm:ghs-approx-lower-bound}
    Any classical randomized one-way protocol that solves $\GHS_n$ with success probability at least $2/3$ requires $\Omega(n^2)$ bits of communication.
\end{theorem}

\begin{proof}
    By Yao's Principle \cite{yao1983lower}, it suffices to construct a distribution $\calD$ over instances of $\GHS_n$ and prove that any arbitrary deterministic one-way protocol $\Pi$ with distributional error at most $1/3$ under $\calD$ uses $\Omega(n^2)$ bits of communication.
    
    \paragraph{The input distribution} We first define the hard input distribution $\calD$. Fix $r=6$, so that $\beta \coloneqq \left(\frac78\right)^r<\frac12$, and assume $n\geq 3r$. By \cref{prop:block-packing} below, there is a set $\calB=\{B_1,\cdots,B_t\} \subseteq \{B \subseteq [n] \mid |B| = 3r\}$ with $|B_i \cap B_j| \leq 1$ for all $i \neq j$ and $t = \Theta(n^2)$. Partition each $B_i$ into $r$ disjoint triples, $\delta_{i,1},\cdots,\delta_{i,r}$. We call each $\delta{i,j}$ a \emph{triangle gadget}. Let
    \[
        E \coloneqq \bigcup_{i\in[t],l\in[r]}\binom{\delta_{i,l}}2.
    \]
    A union bound gives $|E| \leq 3rt$. Note that distinct blocks share no common edge since they intersect in at most one vertex. Therefore the triangle gadgets inside each block are vertex-disjoint, so $N\coloneqq |E|=3rt=\Theta(n^2)$. Sample $A \in \F_2^{n\times n}$ by letting $A_{ij} \sim \operatorname{Unif}\{0,1\}$ independently for each $(i,j) \in E$, and $A_{ij} = 0$ for $(i,j) \notin E$ (with $A_{ij} = A_{ji}$ and $A_{ii} = 0$ for all $i,j$). Alice receives a classical description of the graph state $\mathbf{x} = \text{desc}(\ket{\psi_A})$. 
    Independently, sample $i\sim\operatorname{Unif}[t]$ and let Bob's input be $\mathbf{y} = \text{desc}(H_{B_i})$, where $H_{B_i}= \prod_{j \in B_i}H_j$. Together, these define $\calD$.
    
    \paragraph{Analysis} Fix an arbitrary deterministic $c$-bit one-way protocol $\Pi$. Under this protocol, Alice's message is a function of her input $\mathbf{x} = \operatorname{desc}(\ket{\psi_A})$ which we denote $\msg: \F_2^N \rightarrow \F_2^c$, and conditioned on receiving a message $m \in \F_2^c$, Bob's output is a function of his input $\mathbf{y}_i = \operatorname{desc}(H_{B_i})$ which we denote $y_m: \F_2^{O(n^2)} \rightarrow \F_2^n$. For each possible message $m$, let
    \begin{equation}
        C_m \coloneqq \{\mathbf{x} \in \F_2^N \mid \msg(\mathbf{x})=m\}.
    \end{equation}
    be the ``message class'' of Alice's inputs which are in the pre-image of $m$ under $\msg$. Define the error of the protocol on a fixed Alice input $\mathbf{x}$ as 
    \begin{equation}
        e_\mathbf{x} = \Pr_{i\sim\operatorname{Unif}[t]} \left[ \Pi(\mathbf{x},\mathbf{y}_i) \notin \GHS_n(\mathbf{x},\mathbf{y}_i) \right]
    \end{equation}
    which is the fraction of Bob inputs for which the protocol yields an incorrect output on Alice's input $\mathbf{x}$. Additionally, define $E_i^m(\mathbf{x}) = \{y_m(\mathbf{y_i}) \in \GHS_n(\mathbf{x},\mathbf{y}_i)\}$ and $X_m(\mathbf{x}) \coloneqq \sum_{i=1}^t \mathbf 1[E_i^m(\mathbf{x})]$. The overall distributional error is then $\operatorname{Err}_{\Pi}(\calD) = \E_{\mathbf{x}}[e_\mathbf{x}] \le 1/3$. For every message $m$, define a ``good'' message class as
    \begin{equation}
        C_m^* = \{\mathbf{x} \in \F_2^N \mid \msg(\mathbf{x})=m \text{ and } e_\mathbf{x} \le \frac65\cdot \frac13\}
    \end{equation} 
    By Markov's inequality, $\Pr[e_\mathbf{x} \ge \frac65\cdot \frac13] \le \frac56$. Therefore, 
    \begin{equation}\label{eq:total-good-messages}
        \sum_m |C_m^*| \ge 2^{N}/6.
    \end{equation}
    In  \cref{lem:largest-message-class} below, we prove that there exists a constant $d > 0$ such that for every message $m$,
    \begin{equation}
        |C_m^*| \le 2^{N-d t}.
    \end{equation}
    Combining \cref{lem:largest-message-class} and \cref{eq:total-good-messages} gives that the number of distinct good messages is at least $\frac{2^N}{6\cdot 2^{N-dt}} = 2^{dt}/6$. Since a $c$-bit protocol has at most $2^c$ messages, $2^c \ge 2^{dt}/6$ with $t=\Theta(n^2)$ implies $c \ge dt-\log6 = \Omega(n^2)$, which concludes the proof of \cref{thm:ghs-approx-lower-bound}.
\end{proof}

\begin{lemma}\label{lem:largest-message-class}
    There exists a constant $d > 0$ such that for every message $m$,
    \begin{equation}
        |C_m^*| \le 2^{N-d t}.
    \end{equation}
\end{lemma}
\begin{proof}
    Fix a message $m$. The main proposition is that the probability the protocol under fixed $m$ succeeds simultaneously on any collection of $J$ blocks is multiplicatively bounded in $|J|$. Namely, in \cref{prop:mult-success} below we show that for every $J\subseteq[t]$, 
    \begin{equation}\label{eq:block-intersection-bound}
        \Pr_{\mathbf{x}\sim \calD_\mathbf{x}}\left[\bigcap_{i\in J}E_i^m(\mathbf{x})\right]\le\beta^{|J|}
    \end{equation}
    \cref{prop:block-tail} uses the above result to bound the success probability of Bob's outputs on most blocks: we show that there exists a constant $d >0$ such that
    \begin{equation}\label{eq:Xm-tail}
        \Pr_{\mathbf{x}\sim \calD_\mathbf{x}}\left[X_m(\mathbf{x})\ge\frac t2 \right] \le 2^{-dt}.
    \end{equation}
    Since we only require $X_m(\mathbf{x}) \geq \frac35t$ and $1-e_\mathbf{x} \ge 1-\frac65\cdot\frac13 = \frac35 > \frac12$, we have
    \[
        \frac{|C_m^*|}{2^N} = \Pr_{\mathbf{x}\sim \calD_\mathbf{x}}\left[\mathbf{x} \in C_m^* \right] \le \Pr_{\mathbf{x}\sim \calD_\mathbf{x}}\left[\frac{X_m(\mathbf{x})}{t}\ge\frac35 \right] \leq \Pr_{\mathbf{x}\sim \calD_\mathbf{x}}\left[X_m(\mathbf{x})\ge\frac t2 \right] \le 2^{-d t}
    \]
    so
    \[
        |C_m^*| \le 2^{N-dt}.
    \]
\end{proof}

\begin{proposition} \label{prop:mult-success}
    For every message $m$ and every $J\subseteq[t]$,
    \[
        \Pr_{\mathbf{x}\sim \calD_\mathbf{x}}\left[\bigcap_{i\in J}E_i^m(\mathbf{x})\right]\le\left(\frac78\right)^{r|J|}.
    \]
\end{proposition}

\begin{proof}
    Fix a message $m$ and $J \subseteq [t]$. For each $i \in J$, let $z_i = y_m(y_i)$ be Bob's output given message $m$ and input block $B_i$. For every matrix $A \in \F_2^{n\times n}$ in the support of the hard distribution and triangle $\delta_{i,l}=\{u,v,w\}$ define 
    \begin{equation}
        s^{i,l}_j(A) = \begin{cases}
        A_{vw} & j = u\\
        A_{uw} & j = v\\
        A_{uv} & j = w\\
        0 & j \in B_i \setminus \delta_{i,l}\\
    \end{cases}.
    \end{equation} The vector $s^{i,l}(A)$ lies in $\ker(A_{B_i,B_i})$ for every such $A$:
    \[
        \begin{pmatrix}
            0&A_{uv}&A_{uw}\\
            A_{uv}&0&A_{vw}\\
            A_{uw}&A_{vw}&0
        \end{pmatrix}
        \begin{pmatrix}
            A_{vw}\\ A_{uw}\\ A_{uv}
        \end{pmatrix}
        =0 \pmod 2.
    \]
    Substituting $y = s^{i,l}(A)$ into Fact~\ref{fact:support-criterion} gives the cubic polynomial,
    \begin{equation} \label{eq:polygadget}
        p_{i,l}^m(A) = A_{uv}A_{uw}A_{vw} + \sum_{y\notin B_i} (A_{vw}A_{uy}+A_{uw}A_{vy}+A_{uv}A_{wy})z_y^i + A_{vw}z_u^i + A_{uw}z_v^i + A_{uv}z_w^i.
    \end{equation}
    By definition, $E_i^m(\mathbf{x}_A)$ implies $p_{i,l}^m(A) = 0$ for every $l \in [r]$. Hence 
    \begin{equation}
        \bigcap_{i\in J}E_i^m(\mathbf{x}_A) \Longrightarrow A \in \bigcap_{i\in J, l \in [r]} \calV (p_{i,l}^m).
    \end{equation}
    By \cref{def:zero-set}, letting $I_J^m = \left\langle p_{i,l}^m : i\in J,\ l\in[r] \right\rangle$ gives 
    \begin{equation}\label{eq:approx-message-upper-bound}
        \Pr_{\mathbf{x}_A\sim \calD_\mathbf{x}}\left[ \bigcap_{i\in J}E_i^m(\mathbf{x}_A) \right] \le \frac{ |\calV(I_J^m)| }{2^N}.
    \end{equation}
    \cref{eq:approx-message-upper-bound} shows that the success probability of all Bob inputs indexed by $J$ is upper bounded by the fraction of common zeros of the corresponding support polynomials. By \cite{geil2017boundingnumbercommonzeros}, the number of zeros of a system of polynomial equations can be upper bounded by the number of corresponding standard monomials. Namely,
    \begin{equation}\label{eq:approx-boolean-footprint}
        |\calV(I_J^m)| \le \left| \Delta\left( I_J^m + \langle A_e^2-A_e:e\in E \rangle \right)
        \right|.
    \end{equation}
    Under a degree-based monomial ordering, the leading monomial of $p_{i,l}^m$ is $A_{uv}A_{uw}A_{vw}$, the product of the three edges of $\delta_{i,l}$. Let $q\coloneqq r|J|.$
    By \cref{eq:approx-boolean-footprint}, the set of common zeros to $p_{i,l}^m = 0$ for all $i \in J, l \in [r]$ is at most the number of squarefree standard monomials not divisible by the leading monomial. As a consequence, each triangle gadget may not contain all three edges, so there are at most 7 standard monomials per local gadget. By construction of the hard distribution, the $q$ leading monomials correspond to $q$ edge-disjoint triangle gadgets which gives at most $7^q$ squarefree monomials. There are $2^{N-3q}$ ways to assign Boolean values to the remaining $N-3q$ edges. Thus
    \begin{equation}
        |\calV(I_J^m)| \le 7^q2^{N-3q}.
    \end{equation}
    Dividing by $2^N$,
    \begin{equation}
        \Pr_{\mathbf{x}_A\sim \calD_\mathbf{x}}\left[ \bigcap_{i\in J}E_i^m(\mathbf{x}_A) \right] \le \left(\frac78\right)^q = \left(\frac78\right)^{r|J|}.
    \end{equation}
\end{proof}

\begin{proposition} \label{prop:block-tail}
For every message $m$,
    \begin{equation}    
        \Pr_{\mathbf{x} \sim \calD_\mathbf{x}}\left[X_m(\mathbf{x})\ge\frac t2 \right] \le 2^{-tH_2(1/2\|\beta)}.
    \end{equation}
\end{proposition}

\begin{proof}
    Let $I_i=\mathbf 1[E_i^m(\mathbf{x})]$ for all $i \in [t]$. For every constant $\lambda>0$,
    \begin{align}
        \E[\exp(\lambda X_m)] &= \E\left[\prod_{i=1}^t\exp(\lambda I_i)\right] \tag{$X_m(\mathbf{x}) = \sum_{i=1}^t \mathbf{1}[E_i^m(\mathbf{x})]$}\\
        &= \E\left[\prod_{i=1}^t\bigl(1+(\exp(\lambda)-1)I_i\bigr)\right]\\
        &= \E\left[\prod_{i=1}^t\bigl(1+\alpha I_i\bigr)\right] \tag{$\alpha\coloneqq\exp(\lambda)-1$}\\
        &= \sum_{K\subseteq[t]}\left(\alpha^{|K|} \cdot \E\left[\prod_{i\in K}I_i\right]\right)\\
        &\le \sum_{K\subseteq[t]} \left(\alpha^{|K|}\cdot \beta^{|K|}\right) \tag{\cref{prop:mult-success}}\\
        &= \bigl(1+\alpha\beta\bigr)^t \label{eq:bin-var}
    \end{align}
    By Markov's inequality,
    \begin{align}
        \Pr_{\mathbf{x} \sim \calD_\mathbf{x}}\left[X_m(\mathbf{x})\ge\frac t2 \right] & \le \exp(-\lambda t/2)\E\left[\exp(\lambda X_m)\right]\\
        &\le \exp(-\lambda t/2)(1 + \alpha \beta)^t \tag{\cref{eq:bin-var}}\\
        &= (2\sqrt{\beta(1-\beta)})^t\tag{$\lambda \coloneqq \ln\left(\frac{1-\beta}{\beta}\right)$}
    \end{align}
    Where in the last step $\lambda > 0$ since $\beta < \frac12$. By \cref{def:binentropy}, $H_2(1/2||\beta) = \frac12\log_2 \frac {1/2}{\beta} + \frac12 \log_2 \frac{1/2}{1-\beta} = -\log_2(2\sqrt{\beta(1-\beta)})$. Hence,
    \[
        \Pr\left[X_m\ge\frac t2\right] \le 2^{-tH_2(1/2\|\beta)}.
    \]
\end{proof}

\begin{proposition}[Block packing] \label{prop:block-packing}
    For every fixed $k\ge 2$ and for sufficiently large $n$, there exists a family of ``blocks''
    \[
        \calB\subseteq\{B \subseteq [n] \mid |B| = k\}
    \]
    with $|\calB|=\Theta_k(n^2)$ such that $|B\cap B'|\le 1 \text{ for all distinct } B,B'\in\calB.$
\end{proposition}

\begin{proof}
   We first prove the lower bound by a greedy construction. Let $C_0 = \{B \subseteq [n] \mid |B| = k\}$ denote the initial collection of  size $k$ subsets of $[n]$. Choose any block $B_1 \in C_0$. Let $Y_1 \subseteq C_0$ be the set of all blocks $B'$ such that $|B_1\cap B'|\ge 2$. Choose $B_2 \in C_1 = C_0\setminus Y_1$. Repeat this procedure iteratively until no candidates remain. By a union bound argument, each selected block eliminates at most
    \[
        \binom{k}{2}\binom{n-2}{k-2} = O_k(n^{k-2})
    \]
    blocks. Since there are initially 
    \[
    \binom nk = \Theta_k(n^k)
    \]
    possible blocks, the greedy procedure selects at least $\Omega_k(n^2)$ blocks:
    \begin{equation}
        \frac{\binom nk}{\binom k2 \binom{n-2}{k-2}} = \frac{\Theta_k(n^k)}{ O_k(n^{k-2})} = \Omega_k(n^2).
    \end{equation} Conversely, every block contains $\binom k2$ unordered pairs, and no pair can occur in two distinct blocks. Hence
    \[
        |\calB|\binom k2 \le \binom n2,
    \]
    so $|\calB|=O_k(n^2)$.
\end{proof}

\paragraph{Proof of Fact~\ref{fact:support-criterion}}
For symmetric zero-diagonal adjacency matrix $A \in F_2^{n\times n}$ and $z\in \F_2^n$, let
\[
\ket{\psi_A} = \frac{1}{\sqrt{2^n}}\sum_{x\in \F_2^n}(-1)^{\sum_{i < j}A_{ij}x_ix_j}\ket x, \qquad S \subseteq [n]
\]
and define $s \coloneqq |S|$, $\bar{S}\coloneqq [n]\backslash S$ and $\mathcal{K} \coloneqq \ker(A_{S,S})$. Relabel the graph state vertices and correspondingly permute the rows and columns of $A$ so that the Hadamards act on the leftmost (first) $s$ qubits.
\begin{align*}
    \bra zH_S\ket{\psi_A} &= \bra z H^{\otimes s}\frac{1}{\sqrt{2^n}}\sum_{x\in \F_2^n}(-1)^{q_A(x)}\ket x \tag{$q_A(x) \coloneqq \sum\limits_{i < j} A_{ij}x_ix_j$}\\
    &= \bra z \frac{1}{\sqrt{2^n}}\sum_{x\in \F_2^n}(-1)^{q_A(x)}H^{\otimes s}\ket x\\
    &= \bra z \frac{1}{\sqrt{2^{n+s}}}\sum_{x\in \F_2^n}(-1)^{q_A(x)}\sum_{y \in \F_2^s}(-1)^{x_S^\top y}\ket {y} \ket{x_{\bar S}}\\
    &= \frac{1}{\sqrt{2^{n+s}}}(-1)^{q_{A_{\bar{S},\bar{S}}}(z_{\bar{S}})}\sum_{x\in \F_2^s}(-1)^{q_{A_{S,S}}(x) + r(x)} \tag{$r(x) \coloneqq x^\top (z_S + A_{S,\bar{S}}z_{\bar{S}})$}
\end{align*}
Let $f(x) \coloneqq q_{A_{S,S}}(x) + r(x)$, and define $\calX$ to be the ``unfixed" part of the amplitude. Note that $f(x+y) = f(x)+f(y) + x^\top A_{S,S}y$ for $x,y \in \F_2^S$.
\begin{align*}
    \calX &\coloneqq \frac{1}{\sqrt{2^{n+s}}}\sum_{x \in \F_2^s}(-1)^{f(x)}\\
    &= \frac{1}{\sqrt{2^{n+s}}}\sum_{x \in \F_2^s}(-1)^{f(x+y)} \tag{$\forall y \in \F_2^s$}\\
    &= \frac{1}{\sqrt{2^{n+s}}}\sum_{x \in \F_2^s}(-1)^{f(x) + f(y)} \tag{$\forall y \in \mathcal{K}$}\\
    &= \calX \cdot (-1)^{f(y)} \tag{$\forall y \in \mathcal{K}$}
\end{align*}
Hence $\calX \neq 0 \Rightarrow f(y) = 0 \:\:\forall y\in \mathcal{K}$. Conversely, if $f(y) = 0 \:\:\forall y\in \mathcal{K}$, then $\calX \neq 0$:
\begin{align*}
    \calX^2 &= \frac{1}{2^{n+s}}\sum_{x,z \in \F_2^s}(-1)^{f(x) + f(z)} \\
    &= \frac{1}{2^{n+s}}\sum_{x,y \in \F_2^s}(-1)^{f(x) + f(x+y)} \\
    &= \frac{1}{2^{n+s}}\sum_{y \in \F_2^s}(-1)^{f(y)}\sum_{x\in \F_2^S}(-1)^{x^\top A_{S,S}y} \\
    &= \frac{1}{2^{n+s}}\sum_{y \in \mathcal{K}}(-1)^{f(y)}\cdot 2^s \\
    &= \frac{1}{2^{n}}\sum_{y \in \mathcal{K}}(-1)^{f(y)} \\
    &= \frac{|\mathcal{K}|}{2^{n}} > 0 \tag{$f(y) = 0 \:\: \forall \: y \in \mathcal{K}$}
\end{align*}
So $\bra zH_S\ket{\psi_A} \neq 0 \Longleftrightarrow \calX \neq 0 \Longleftrightarrow f(y) = 0 \:\:\forall y\in \ker(A_{S,S})$, as intended.

\end{document}